\documentclass[letter,11pt]{article}

\newcommand{\bigoh}{{\cal O}}
\newcommand{\cO}{{\cal O}}

\newcommand{\inv}[0]{\mathsf{-1}}

\usepackage[shortlabels]{enumitem}
\usepackage{thm-restate}
\usepackage{lmodern}
\usepackage[dvipsnames,svgnames,x11names]{xcolor}
\usepackage{graphicx}
\usepackage{tikz}
\usetikzlibrary{decorations.shapes}
\usepackage{soul}

\usepackage{fullpage}
\usepackage{algorithmicx}
\usepackage[ruled,noline,linesnumbered]{algorithm2e}
\usepackage{boxedminipage}
\usepackage{stmaryrd}

\usepackage{amsmath,amssymb,amsthm}
\usepackage[retainorgcmds]{IEEEtrantools}
\usepackage{pdfpages}
\usepackage{lastpage}

\usepackage{authblk}

\usepackage{hyperref}
\hypersetup{
    pdfauthor={Neeldhara Misra},     
    pdfnewwindow=true,      
    colorlinks=true,       
    linkcolor=SlateBlue,          
    citecolor=FireBrick,        
    filecolor=magenta,      
    urlcolor=Emerald           
}

\usepackage{url}

\usepackage[backgroundcolor=Moccasin,colorinlistoftodos]{todonotes}

\makeatletter
 \providecommand\@dotsep{5}
 \def\listtodoname{}
 \def\listoftodos{\@starttoc{tdo}\listtodoname}
 \makeatother

\newcounter{nmcomment}

\newcommand{\FPT}{\textrm{\sf FPT}}

\newcommand{\NPC}{\textrm{\textup{NP-complete}}\xspace}

\newcommand{\defparproblem}[4]{
  \vspace{3mm}
\noindent\fbox{
  \begin{minipage}{.95\textwidth}
  \begin{tabular*}{\textwidth}{@{\extracolsep{\fill}}lr} \textsc{#1}  & {\bf{Parameter:}} #3 \\ \end{tabular*}
  {\bf{Input:}} #2  \\
  {\bf{Question:}} #4
  \end{minipage}
  }
  \vspace{2mm}
}

\newcommand{\CC}{{\mathcal C}}

\newcommand{\EE}{{\mathcal E}}

\newcommand{\HH}{{\mathcal H}}

\newcommand{\OO}{{\mathcal O}}

\newcommand{\wtilde}{\widetilde}

\newcommand{\cS}{{\mathcal{S}}}

\newcommand{\cF}{{\mathcal{F}}}

\newcommand{\cR}{\mathcal{R}}

\newcommand{\cA}{\mathcal{A}}

\newcommand{\cB}{\mathcal{B}}

\newcommand{\cH}{\mathcal{H}}

\newcommand{\cT}{\mathcal{T}}

\newcommand{\cX}{\mathcal{X}}

\newcommand{\cL}{\mathcal{L}}
\newcommand{\cZ}{\mathcal{Z}}

\newcommand{\prob}[1]{\mathrm{Prob}[~#1~]}
\newcommand{\mean}[1]{\mathrm{\bf E}[~#1~]}

\usepackage{libertine}
\newtheorem{theorem}{Theorem}[section]
\newtheorem{lemma}[theorem]{Lemma}
\newtheorem{claim}[theorem]{Claim}
\newtheorem{corollary}[theorem]{Corollary}
\newtheorem{definition}[theorem]{Definition}
\newtheorem{observation}[theorem]{Observation}
\newtheorem{proposition}[theorem]{Proposition}

\renewcommand*\overline[1]{%
   \vbox{%
     \hrule height 0.5pt
     \kern0.25ex
     \hbox{%
       \kern-0.15em
       \ifmmode#1\else\ensuremath{#1}\fi
       \kern-0.1em
     }
   }
}

\usetikzlibrary{shapes.geometric,matrix,shapes,arrows,positioning,chains}

\usepackage{url}
\usepackage{microtype}
\definecolor{brg}{rgb}{0.0, 0.26, 0.15}

\usepackage{tcolorbox}
\usepackage{tikz}

\newcommand{\thinblackbox}[1]{
\begin{center}
\begin{tcolorbox}[colback=white!30, colframe=black!10,width=15cm]
#1
\end{tcolorbox}
\end{center}
}

\newcommand{\pname}{\textsc}
\newcommand{\ProblemFormat}[1]{\pname{#1}}
\newcommand{\ProblemIndex}[1]{\index{problem!\ProblemFormat{#1}}}
\newcommand{\ProblemName}[1]{\ProblemFormat{#1}\ProblemIndex{#1}{}\xspace}

 \newcommand{\probVDH}{\ProblemName{Vertex Deletion to $\cH$}}

\newcommand{\probPVD}{\ProblemName{Planar Vertex Deletion}}
\newcommand{\probCVD}{\ProblemName{Chordal Vertex Deletion}}
\newcommand{\probIVD}{\ProblemName{Interval Vertex Deletion}}
\newcommand{\probOCT}{\ProblemName{Odd Cycle Transversal}}
\newcommand{\probFVS}{\ProblemName{Feedback Vertex Set}}
 \newcommand{\probVC}{\ProblemName{Vertex Cover}}

\newcommand{\oct}{\mathop{\textsc{oct}}}
\newcommand{\defparproblemR}[4]{
  \vspace{1mm}
\noindent\colbox{
  \begin{minipage}{0.94\textwidth}
  \begin{tabular*}{\textwidth}{@{\extracolsep{\fill}}lr} #1  &
    {\textbf{Parameter:}} #3 \\ \end{tabular*}
  {\textbf{Input:}} #2  \\
  {\textbf{Task:}} #4
  \end{minipage}
  }
}

 \newcommand{\fptsemips}{{\sf FPT-Semi-PS}}
 \newcommand{\semips}{{\sf SemiPS}}

\usepackage{graphicx,color}
\usepackage{empheq}
\usepackage{tcolorbox}
\definecolor{blueish}{rgb}{0.122, 0.435, 0.698}
\definecolor{dagstuhlyellow}{rgb}{0.99,0.78,0.07}
\definecolor{lightgray}{rgb}{0.9,0.9,0.9}
\newtcbox{\colbox}{
size=title,
  nobeforeafter,
  colframe=white,
  colback=blue!5!white,
  arc=10pt,
  tcbox raise base}

\newcommand{\hh}{\ensuremath{\mathcal{H}}}

\renewcommand{\tilde}{\widetilde}

\newenvironment{tightcenter}
 {\parskip=0pt\par\nopagebreak\centering}
 {\par\noindent\ignorespacesafterend}

\usepackage{tikz}
\usetikzlibrary{calc}
\usepackage{xargs}
\usepackage{xifthen}
\usepackage{framed}

\usepackage{ctable}

\newlength{\RoundedBoxWidth}
\newsavebox{\GrayRoundedBox}
\newenvironment{GrayBox}[1]%
   {\setlength{\RoundedBoxWidth}{\textwidth-4.5ex}
    \def\boxheading{#1}
    \begin{lrbox}{\GrayRoundedBox}
       \begin{minipage}{\RoundedBoxWidth}%
   }{%
       \end{minipage}
    \end{lrbox}%
    \begin{tightcenter}%
    \begin{tikzpicture}%
       \node(Text)[draw=black!90,fill=white,rounded corners,%
             inner sep=2ex,text width=\RoundedBoxWidth]%
             {\usebox{\GrayRoundedBox}};
        \coordinate(x) at (current bounding box.north west);
        \node [draw=white,rectangle,inner sep=3pt,anchor=north west,fill=white]
        at ($(x)+(6pt,.75em)$) {\boxheading};
    \end{tikzpicture}
    \end{tightcenter}\vspace{0pt}%
    \ignorespacesafterend
}

\title{Meta-theorems for Parameterized Streaming Algorithms }
\author{Daniel Lokshtanov\thanks{University of California, Santa Barbara, USA. \texttt{daniello@ucsb.edu}} ~~
Pranabendu Misra\thanks{Chennai Mathematical Institute. \texttt{pranabendu@cmi.ac.in}}~~
Fahad Panolan\thanks{Indian Institute of Technology, Hyderabad. \texttt{fahad@iith.ac.in}} ~~   
M. S. Ramanujan\thanks{University of Warwick, UK. \texttt{r.maadapuzhi-sridharan@warwick.ac.uk}}\and Saket Saurabh\thanks{Department of Informatics, University of Bergen, Norway and  Institute of Mathematical Sciences, HBNI,  India,
\texttt{saket@imsc.res.in}} ~~Meirav Zehavi\thanks{Ben-Gurion University of the Negev, Israel. \texttt{meiravze@bgu.ac.il}
}
}

\begin{document}
 \date{}
\begin{titlepage}
\def\thepage{}
\thispagestyle{empty}
\maketitle
\begin{abstract}

The streaming model was introduced to parameterized complexity independently by Fafianie and Kratsch [MFCS14] and by Chitnis, Cormode, Hajiaghayi and Monemizadeh
 [SODA15]. Subsequently, it was broadened by Chitnis, Cormode, Esfandiari, Hajiaghayi and Monemizadeh [SPAA15] and  by Chitnis, Cormode, Esfandiari, Hajiaghayi, McGregor, Monemizadeh and Vorotnikova [SODA16].
Despite its strong motivation, the applicability of the streaming model to central problems in parameterized complexity has remained, for almost a decade, quite limited. Indeed, due to simple $\Omega(n)$-space lower bounds for many of these problems, the $k^{\OO(1)}\cdot \mathsf{polylog}(n)$-space requirement in the model is  too strict.


Thus, we explore {\em semi-streaming} algorithms for parameterized graph problems,  and present the first systematic study of this topic. Crucially, we aim to construct succinct representations of the input on which optimal post-processing time complexity can be achieved.  

\begin{itemize}
	\item We devise meta-theorems specifically designed for parameterized streaming and demonstrate their applicability by obtaining the first $k^{\bigoh(1)}\cdot n\cdot \mathsf{polylog}(n)$-space streaming algorithms for well-studied problems such as {\sc Feedback Vertex Set on Tournaments}, {\sc Cluster Vertex Deletion}, {\sc Proper Interval Vertex Deletion} and {\sc Block Vertex Deletion}. 
	In the process, we demonstrate a fundamental connection between semi-streaming algorithms for recognizing graphs in a graph class $\cH$ and semi-streaming algorithms for the problem of vertex deletion into $\cH$.

	\item We present an algorithmic machinery for obtaining streaming algorithms for cut problems and exemplify this by giving the first $k^{\bigoh(1)}\cdot n\cdot \mathsf{polylog}(n)$-space streaming algorithms for {\sc Graph Bipartitization}, {\sc Multiway Cut} and {\sc Subset Feedback Vertex Set}. 
\end{itemize}

\end{abstract}

\newpage
\small{
\tableofcontents}

\end{titlepage}
\newpage

\section{Introduction}\label{sec:introduction}

The \textit{Parameterized Streaming model} was proposed independently by Fafianie and Kratsch~\cite{fafianie2014streaming} and Chitnis, Cormode, Hajiaghayi and Monemizadeh~\cite{ChitnisCHM15} with the goal of studying space-bounded parameterized algorithms for NP-complete problems.
In this setting, the space is restricted to $\widetilde{\OO}(k^{\OO(1)})$ (that is, $\OO(k^{\OO(1)}\cdot\mathsf{polylog}(n))$), where the parameter $k$ is a non-negative integer that aims to express some structure in the input. A  feature of this model is that it allows one to design (exact) streaming algorithms for certain NP-complete graph problems such as {\sc Vertex Cover} (where the parameter $k$ is the size of the solution, i.e., the vertex cover).

Unfortunately, even allowing $\OO(k^{\OO(1)}\cdot\mathsf{polylog}(n))$ space can only lead to limited success
as numerous NP-complete graph problems (when parameterized by the solution size) 
require space $\widetilde{\Omega}(n)$ even when $k$ is a fixed constant (e.g., $k=0$)~\cite{fafianie2014streaming,ChitnisCHM15}. \footnote{For instance, consider the classic {\NPC} problem {\sc Graph Bipartization}, where the goal is to determine, given a graph $G$ and number $k$, whether removing $k$ vertices from $G$ results in a bipartite graph. For $k=0$, this problem is nothing but testing whether $G$ is bipartite, for which there is a $\Omega(n\log n)$-space lower bound for one-pass algorithms~\cite{SunW15}.}
Subsequently,  Chitnis and Cormode~\cite{ChitnisC19} made an important advance by defining a hierarchy of complexity classes for parameterized streaming; among these classes, they defined the {\semips} (parameterized semi-streaming) class. This class, in the spirit of standard semi-streaming~\cite{muthukrishnan2005data,FeigenbaumKMSZ05}, includes all graph problems that can be solved by an algorithm that uses $\widetilde{\OO}(f(k)n)$ space for some computable function $f$ of $k$. 

The class {\semips}
 allows unbounded computational power, both while processing the stream and in post-processing. Notice that this implies that several problems that are not expected to have fixed-parameter algorithms (e.g., problems that are W[1]-hard on planar graphs) are contained in {\semips}, since the entire input can be stored and then solved in post-processing. Thus, the model seems too powerful to combine well with the usual meaning of efficiency in parameterized algorithms (i.e., solvability in time $f(k)\cdot n^{\cO(1)}$).  
   This leaves a gap in the area which must be addressed.
	That is, what would be an appropriate refinement of {\semips} that captures problems in {\FPT}?
We bridge this gap in the state of the art by initiating the study of parameterized streaming algorithms where the space complexity is bounded by $\widetilde{\OO}(f(k)n)$ {\em and} the time complexity is bounded by $g(k)n^{\bigoh(1)}$ at every edge update and in post-processing, for computable functions $f$ and $g$. We call such algorithms {\em fixed-parameter semi-streaming algorithms} (or FPSS algorithms).

The study of parameterized streaming algorithms has remained in a relatively early stage, with the state of the art predominantly focused on studying individual problems. 
In contrast, there has been remarkable progress in incorporating various other domains into the framework of parameterized complexity, such as: 
Dynamic Graph Algorithms~\cite{DvorakKT14,AlmanMW20,ChenCDFHNPPSWZ21,MajewskiPS23,korhonen2023dynamic}, Approximation Algorithms~\cite{DemaineHK05,LokshtanovPRS17,GuptaLL18,GrandoniKW19,LokshtanovSS20,ChalermsookCKLM20} and Sensitivity Oracles~\cite{BiloCCC0LSW22,AlmanH22,PilipczukSSTV22}. 
The wide interest in these domains is also a consequence of various unique challenges that arise in their settings. 
The setting of parameterized streaming algorithms offers its own unique challenges as well. For instance,   
since we cannot store all the edges incident to a vertex, it becomes challenging to utilize many of the standard tools and techniques of Parameterized Complexity.

In this paper, we significantly advance the study of parameterized streaming algorithms by giving \emph{meta-theorems}, which lead to the first FPSS algorithms for large classes of problems studied in Parameterized Complexity. This includes algorithms for {\sc Vertex Deletion to $\cal H$}, where $\cal H$ can be a graph class characterized by a finite forbidden family, or it can be a hereditary graph class. We also give a framework for various graph-cut problems and apply it to obtain the first FPSS algorithms for {\sc Graph Bipartization} (also called {\sc Odd Cycle Transversal}), {\sc Multiway Cut} and {\sc Subset Feedback Vertex Set}. 
Another key contribution is to show how, in several cases, one can essentially ``reconstruct'' the input graph, providing ``query-access'' to its edge-set using only $\widetilde{\OO}(f(k)n)$ space, thereby allowing the use of various algorithmic tools in Parameterized Complexity. Of course these tools must then be employed within limited space themselves, which in itself is a non-trivial problem that we address in this paper. 

\section{Our Contributions and Methodology}\label{sec:ourContribution}

Our  algorithmic contribution consists of (i) two meta-theorems that yield FPSS algorithms for 
several basic graph optimization problems, and (ii) a methodology to obtain FPSS algorithms for well-studied graph cut problems. Our results in (i) are arguably the first general purpose theorems applicable to dense graphs in this line of research, and using the machinery we develop in (ii), we resolve an open problem of Chitnis and Cormode~\cite{ChitnisC19} on  whether {\sc Odd Cycle Transversal} has an $\widetilde{\OO}(f(k)n)$-space streaming algorithm.  To obtain our results, we introduce novel sparsification methodologies and combinatorial results  that could be of independent interest or of use in the design of other parameterized semi-streaming algorithms.

We next describe the graph optimization problems we focus on, which are all {\em vertex-deletion problems}. 
Let $\cal H$ be a family of graphs.
The canonical vertex-deletion problem corresponding to the family $\cal H$ is defined as follows.

\medskip

\defparproblemR{\probVDH}{A graph $G$,  and an integer
  $k$.}{$k$}{Decide whether there is a vertex set $S$ of size at most $k$ such that $G-S\in \cH$.}

\medskip
\noindent

This family of problems includes fundamental problems in graph theory and combinatorial optimization, e.g., \probPVD,  \probOCT,  \probCVD,  \probIVD,  \probFVS, and \probVC (corresponding to $\cH$ being the class of planar, bipartite, chordal, interval, acyclic or edgeless graphs,  respectively). Many vertex-deletion problems are well known to be {\NPC}~\cite{LewisY80,Yannakakis78}.   Therefore, they have been studied extensively within various algorithmic paradigms such as approximation algorithms, parameterized complexity, and algorithms on restricted input classes~\cite{FominLMS12,Fujito98,LundY94}. 
However, when the input graph is too large to fit into the available memory, these paradigms on their own are insufficient. This naturally motivates the study of these problems in the streaming model.


Importantly, in the streaming setting, for {\probVDH}, even the special case of $k=0$ (i.e., {\em recognition of graphs in $\HH$}) is already non-trivial for many natural choices of $\cH$. That is, just the question of determining whether a given graph $G$ belongs to $\cH$ becomes significantly much harder in the streaming setting compared to the static (i.e., non-streaming) setting and in some cases, it is provably impossible to achieve any non-trivial upper bounds  (e.g., as we show in this paper, for chordal graphs). So, there is a significant challenge in developing space-bounded (i.e., $\wtilde\bigoh(n)$-space) algorithms for recognition of various graph classes.
Our first meta theorem
provides some evidence as to why the recognition problem is challenging for many natural graph classes, by drawing a fundamental correspondence between recognition of graphs in $\cH$ and solving {\probVDH}. 

\subsection{Our First Meta-theorem: $\cH$ is Characterized by the Absence of Finitely Many Induced Subgraphs} \label{sec:firstAlgorithmicResult}
Our first result can be encapsulated in the following surprising message, where $\hh$ is defined by excluding a finite number of forbidden graphs as induced subgraphs.

\medskip
\begin{tcolorbox}[colback=green!5!white,colframe=white!100!black]
\begin{center}
{\bf Vertex deletion to $\cH$ is not harder than recognition for $\cH$.}
\end{center}
\end{tcolorbox}
\medskip

\noindent
That is, we show that {\probVDH} has an FPSS algorithm if and only if graphs in $\cH$ can be recognized with a semi-streaming algorithm. {Equivalently, one could say that in the streaming model, just checking whether a given graph is a member of $\cH$ appears to be {\em as difficult} as solving the seemingly much more general {\probVDH} problem. 
}

To be precise, we prove the following theorem.

\begin{restatable}{theorem}{Hcovering}\label{thm:Hcovering}
Let $\cH$ be a family of graphs defined by a finite number of forbidden induced subgraphs such that $\cH$ admits a deterministic/randomized $p$-pass recognition algorithm in the turnstile (resp.~insertion-only) model for some $p\in\mathbb{N}$. Then, 
    {\probVDH} admits a deterministic/randomized $p$-pass $\wtilde{\bigoh}(k^{\bigoh(1)}\cdot n)$-space streaming algorithm in the turnstile (resp.~insertion-only) model with post-processing time $2^{\bigoh(k)}\cdot n^{\bigoh(1)}$.
\end{restatable}

In the above statement,  a $p$-pass recognition algorithm for $\cH$ is a $p$-pass $\widetilde{\OO}(n)$-space streaming algorithm with polynomially bounded time between edge updates and in post-processing that, given a graph $G$, correctly concludes whether or not $G\in \cH$. The turnstile model permits edge additions and removals whereas the insertion-only model only permits the former. So, our result guarantees that if one can recognize whether a given graph belongs to $\cH$, then one can also solve {\probVDH}  
within the same number of passes and using only $\wtilde{\bigoh}(k^{\bigoh(1)}\cdot n)$-space. Notably, {\em the dependency of the running time we attain on $k$ is the best possible under the Exponential Time Hypothesis {\sc (ETH)}}--- for each of the specific {\probVDH} problems for which we draw corollaries from Theorem~\ref{thm:Hcovering} (and in fact, even for the much simpler {\sc Vertex Cover} problem), it is known that there does not exist a $2^{o(k)}n^{\OO(1)}$-time algorithm under ETH even in the static setting~\cite{CyganFKLMPPS15}.

As a corollary of Theorem~\ref{thm:Hcovering}, we get the first FPSS algorithms for well-studied problems such as 
	{\sc Feedback Vertex Set on Tournaments} (FVST), {\sc Split Vertex Deletion} (SVD), {\sc Threshold Vertex Deletion} (TVD) and  {\sc Cluster Vertex Deletion} (CVD). We refer the reader to the Appendix for the formal descriptions of these problems.

	\smallskip
	
	\noindent{\bf Proof overview for Theorem~\ref{thm:Hcovering}.}  Suppose that the premise holds for the graph class $\cH$ and let $(G,k)$ be the given instance of {\probVDH}. Moreover, let $\cR$ be the finite set of graphs excluded by graphs in $\cH$ as induced subgraphs.
Suppose that the instance is a yes-instance and let $S$ be a hypothetical inclusionwise-minimal solution. Then, there is a set $\cT$ of size $|S|$, which is a set of $\cR$-subgraphs (subgraphs isomorphic to graphs in $\cR$) of $G$ such that for every graph in $\cT$, there is a unique vertex of $S$ that it intersects.
Let $V(\cT)$ denote the union of the vertex sets of the subgraphs in $\cT$ and notice that $|V(\cT)|\leq dk$, where $d$ is the maximum number of vertices among the graphs in $\cR$. Note that $d$ is a constant in our setting. We now construct an $(n,dk,(dk)^2)$-splitter family $\cF_1$ with $\wtilde\bigoh(k^{\bigoh(1)})$ functions and guess a function $f_1\in \cF_1$ which is injective on $V(\cT)$. We refer the reader to Definition~\ref{def:splitterFamilies} in Section~\ref{sec:preliminaries} for a formal definition of splitter families. At this point, it is sufficient for the reader to know that (i) for positive integers $n,x,y$, an $(n,x,y)$-splitter family is a set of functions mapping $[n]$ to $[y]$ such that every subset of $[n]$ of size $x$ is injectively mapped to $[y]$ by at least one of these functions and (ii) there are efficient ways to construct a small enough $(n,x,x^2)$-splitter family.

Returning to our description, we treat $f_1$ as a function that colors $V(G)$ with at most $(dk)^2$ colors. Now, in a single pass, for every set $J$ of at most $d$ colors, we run the recognition algorithm on the subgraph of $G$ induced by the colors in $J$. For each of those subgraphs which have been determined to {\em not} be in $\cH$, we use a second family of splitters (of size $\bigoh(\log n)$) to compute a set of at most $d$ ``candidate'' vertices that could potentially be part of a minimal solution for the original instance. This gives us a set $Z$ of size $k^{\bigoh(1)}$ which contains $S$. At this point, we can obtain a  $2^{\bigoh(k\log k)}n$-time post-processing as follows -- guess a subset $S'$ of $Z$ of size at most $k$ and use the family $\cF_1$ to check whether there is an $\cR$-subgraph of $G$ disjoint from $S'$. Notice that if such an $\cR$-subgraph exists, then there would be a function in $\cF_1$ that is injective on $S'$ plus the vertices of this subgraph and since the subgraph only spans at most $d$ color classes of this function, we would have already stored a bit identifying whether or not this subgraph is in $\cH$.

In order to improve the post-processing time to a single-exponential in $k$ (matching asymptotic lower bounds even in the static setting), we give a non-trivial reduction to the $d$-{\sc Hitting Set} problem. Notice that such a reduction in the static setting is trivial -- the $\cR$-subgraphs of the input graph correspond to the sets that need to be ``hit''. However, we require additional work in our case since the space constraints in our setting prohibit us from explicitly identifying the forbidden subgraphs in $G$. We overcome this obstacle by showing that for every subset of $Z$ of size at most $d$, our data structure can determine whether a solution needs to intersect it and show that it is also sufficient for a solution to intersect precisely these sets. To compute this $d$-{\sc Hitting Set} instance, we employ a further family of splitters with appropriately chosen parameters.

\subsection{Our Second Meta-theorem: $\cH$ is Characterized by the Absence of Infinitely Many Induced Subgraphs}
By strengthening  the requirements in Theorem \ref{thm:Hcovering}, we obtain a similar result even when the set of forbidden induced subgraphs for $\cH$ may be of {\em infinite} size. For example, the obstruction set that defines the class of proper interval graphs is infinite since it contains all chordless cycles. More generally, our result holds for any hereditary graph class (i.e., graph class closed under taking induced subraphs).  Here, we require {\em reconstruction} for $\cH$ rather than recognition for $\cH$, which means that for a given graph $G$, we need to determine whether $G\in \cH$ (as in recognition), but in case the answer is positive, we also need to output a {\em succinct representation of $G$}. By succinct representation, we mean a data structure that takes $\widetilde{\OO}(n)$ space and supports ``edge queries''---given a pair of vertices $u,v$, it answer whether $\{u,v\}$ is an edge in $G$. In addition to this requirement, we also suppose to be given an algorithm that solves the problem (in $\widetilde{\OO}(n)$ space) in the static setting, which is clearly an easier task than attaining the same result in the streaming setting. Specifically, we prove the following theorem.

\begin{restatable}{theorem}{reductionLemma}\label{lem:reductionLemma}
Let $\cH$ be a hereditary graph class such that:
\begin{enumerate}
\item\label{item:reduction1} $\cH$ admits a $p$-pass deterministic/randomized reconstruction algorithm for some $p\in\mathbb{N}$ in the turnstile (resp.~insertion-only) model, and
\item\label{item:reduction2} {\probVDH} admits a deterministic/randomized $\widetilde{\OO}(n\cdot g(k))$-space $f(k)\cdot n^{\OO(1)}$-time (static) algorithm where $g$ and $f$ are some computable functions of $k$. 
\end{enumerate}
Then,  {\probVDH} admits a $p$-pass deterministic/randomized $\widetilde{\OO}(n\cdot g(k)\cdot k^{\OO(1)})$-space streaming algorithm with post-processing time $2^k\cdot f(k) \cdot n^{\OO(1)}$  in the turnstile (resp.~insertion-only) model.
\end{restatable}

We demonstrate the applicability of our theorems for the classes of proper interval graphs and block graphs, where the size of the obstruction set is infinite. For each one of these two problems, we design a reconstruction algorithm that works in $\widetilde{\OO}(n)$ space in the static setting. 
Overall, this yields the following corollaries.

\begin{corollary}\label{cor:pivdIntro}
{\sc Proper Interval Vertex Deletion} admits a randomized $\OO(\log^2 n)$-pass $\widetilde{\OO}(k^{\OO(1)}n)$-space streaming algorithm with $2^{\OO(k)}\cdot n^{\OO(1)}$ post-processing time in the turnstile model. 
\end{corollary}

\begin{corollary}\label{cor:blockIntro}
{\sc Block Vertex Deletion} admits a randomized 1-pass $\widetilde{\OO}(k^{\OO(1)}n)$-space streaming algorithm with $2^{\OO(k)}\cdot n^{\OO(1)}$ post-processing time in the turnstile model. 
\end{corollary}

\noindent{\bf Proof overview for Theorem~\ref{lem:reductionLemma}.} 
Using a construction of an $(n,k,2)$-separating family, we begin by obtaining a family of $k^{\OO(1)}\log n$ vertex subsets of $G$, $F_1,F_2,\ldots,F_t$, with the property that for every pair of vertices $u,v$, and every vertex subset $S$ of size at most $k$ that excludes $u,v$, there exists an $F_i$ that contains $u,v$ and is disjoint from $S$. For each induced subgraph $G_i=G[F_i]$, we call the given reconstruction algorithm and attempt to reconstruct $G_i$ (where all calls are done simultaneously). For all pairs of vertices $u,v$ that never occur together  in any $G_i$ that we managed to reconstruct, we prove that every solution to $(G,k)$ must contain at least one vertex from each of these pairs. So, by considering the union $\widetilde{G}$ of the graphs $G_i$ that we managed to reconstruct, we are able to reconstruct (implicitly) $G$ up to not knowing whether there exist edges between the aforementioned pairs of vertices. However, these pairs of vertices impose ``vertex-cover constraints'', and we can directly consider every possibility of getting rid of them by branching (into at most $2^k$ cases). After that, we  simulate the given parameterized algorithm on $\widetilde{G}$ minus the vertices already deleted.

\medskip
\noindent{\bf Application to Proper Interval Vertex Deletion.} In light of Theorem \ref{lem:reductionLemma}, we need to design a reconstruction algorithm as well as a (static) parameterized algorithm that uses $\widetilde{\OO}(n)$ space. 

\smallskip
\noindent{\it Reconstruction.} Without loss of generality, we can focus only on connected graphs, and we observe that whenever we remove the closed neighborhood of a vertex in a connected proper interval graph, we get at most two connected components. The latter observation motivates the definition of a ``middle vertex'', which is a vertex such that each of the (at most two) connected components that result from the removal of its closed neighborhood has at most $\frac{9}{10}n$ vertices. We note that a randomly picked vertex has probability at least $\frac{4}{5}$ to be a middle vertex. This gives rise to a divide-and-conquer strategy where, at each step, we will aim to ``organize'' the unit intervals corresponding to vertices in the closed neighborhood of a middle vertex, and then recurse to organize those corresponding to vertices in each of the connected components resulting from their removal. However, these three tasks are not independent, and therefore we need to define an annotated version of the reconstruction problem. Here, we are given, in addition to $G$, a partial order on its vertex set, and we need to determine if $G$ admits a unit interval model (and is therefore a unit interval graph) where for pairs of vertices $u<v$, the interval of $u$ starts before that of $v$.

Now, suppose we handle this annotated version, being at some intermediate recursive call, and call the graph corresponding to it by $G$ (which is an induced subgraph of the original input graph). Let $v$ be the middle vertex, $M$ be its closed neighborhood, and $L$ and $R$ be the two connected components of $G-M$ (which may be empty), so that $<$ indicates that $L$ should be ordered before $R$ (or else their naming is chosen arbitrarily). We note that in various arguments, conflicts may arise---for example, here we may have a vertex in one component that is smaller than a vertex in the other (by $<$), and vice versa, in which case we can directly conclude that the original input graph is not a proper interval graph. It turns out that for $L$ (and symmetrically for $R$), it sufficient just to solve the problem while refining $<$ so that vertices with higher number of neighbors in $M$ will be bigger than those having lower number of neighbors in $M$ (with the opposite for $R$); clearly, for all vertices, the number of neighbors in $L,M$ and $R$ can be computed in one pass. We prove that whichever solution (being a unit interval model) will be returned for $L$ and $R$, it should be possible to ``patch'' it with any unit interval ordering of $M$ given that it complies with the following (and that the original graph is a proper interval graphs): we refine the ordering $<$ for $M$ so that vertices with higher degree in $L$ will be smaller, while vertices with higher degree in $R$ will be bigger. As the resolution of $L$ and $R$ is given by the recursive calls, let us now explain how to resolve $M$.

To handle $M$, we first prove that in any unit interval model of it (if one exists), when we go over the vertices from left to right, their degree in $M$ is non-decreasing, then reaches some pick (that is the degree of $v$), and after that becomes non-increasing.  If $L$ and $R$ are empty, then we show that can choose as the leftmost vertex any vertex in $M$ that is smallest by $<$ and has minimum degree in $M$ among such vertices, and if $L$ is not empty (and it is immaterial if $R$ is empty or not), we can choose as the leftmost vertex any vertex in $M$ that has highest degree in $L$ and is simultaneously smallest by $<$, and has smallest degree in $M$ within these vertices (but possibly not within all vertices in $M$). Let this vertex be $a$. Then, all neighbors of $a$ are ordered from left to right non-decreasingly by their degree in $M$ (taking into account some information from the ordering of $L$, which we will neglect in this brief overview, in order to break ties), and after them all non-neighbors of $a$ are ordered from left to right non-increasingly by their degree in $M$. Clearly, to attain the neighbors of $a$, we just need one more pass.

Overall, the above yields a potential unit interval model for original input graph $G$ -- in particular, we argue that if $G$ is a proper interval graph, then this must be a model of it. Lastly, the model is verified by another pass on the stream, checking that $G$ and the graph that the model represents are indeed the same graphs.

\smallskip
\noindent{\it Static Algorithms.} 
Here, we essentially show that the known $2^{\OO(k)}n^{\OO(1)}$ algorithm by van 't Hof and Villanger~\cite{HofV13}  can be implemented in $\widetilde{\OO}(k^{\OO(1)}n)$ space.

\medskip
\noindent{\bf Application to Block Vertex Deletion.} As before, we start with reconstruction and then turn to the static algorithms (which are here, unlike the case of proper interval graphs, non-trivial).

\smallskip
\noindent{\it Reconstruction.}  We provide reconstruction algorithms for a graph class that is broader than the class of block graphs. Here, a $t$-flow graph is one where between every pair of non-adjacent vertices, the number of vertex-disjoint paths is at most $t$, and a $t$-block graph is a $t$-flow graph that is chordal. A block graph is just a $1$-block graph. For the reconstruction of $t$-flow graphs, the algorithm starts similarly to the proof of Theorem \ref{lem:reductionLemma}. We first attain a family of $t^{\OO(1)}\log n$ vertex subsets of $G$, $F_1,F_2,\ldots,F_r$, with the property that for every pair of vertices $u,v$ and vertex subset of size at most $t$, there exists an $F_i$ that contains $u,v$ and is disjoint from $S$.  Using one pass on the stream, for each induced subgraph $G_i=G[F_i]$, we compute a spanning forest, and for each vertex in $G$, we compute its degree (in $G$). Then, we let $G'$ be the union of these spanning forests. In post-processing, we construct (implicitly) a graph $\widetilde{G}$ from $G'$ as follows (only $G'$ is kept explicitly). For every pair of non-adjacent vertices $u,v$ in $G'$ and every vertex subset of $G$ if size at most $t$, we choose some $G_i$ that contains both $u,v$ and no vertex from $S$ (which can be shown to exist), and if it holds that $u$ and $v$ are in the same connected component of $G_i$, then we add the edge $\{u,v\}$. If every vertex has the same degree in $G$ and $\widetilde{G}$, then we conclude the $G$ is a $t$-flow graph with reconstruction (stored in an implicit manner) $\widetilde{G}$, and else we conclude that $G$ is not a $t$-flow graph.  For correctness, we note that it can be shown that, necessarily, $E(G)\subseteq E(\widetilde{G})$, and that in case $G$ is a $t$-flow graph, then also $E(\widetilde{G})\subseteq E(G)$ (else it may not be true). The degree test is used to prove that reverse direction, where the algorithm returns true, in which case the degree equalities and the fact that $E(G)\subseteq E(\widetilde{G})$ imply that $G=\widetilde{G}$ (and we construct $\widetilde{G}$ in a way that ensures it is a $t$-flow graph.

To adapt the result to reconstruct $t$-block graphs, we first reconstruct them as $t$-flow graphs (because every $t$-block graph is in particular a $t$-flow graph), and then, having already this reconstruction at hand, we check whether they have a {\em perfect elimination ordering} in polynomial time and $\widetilde{\OO}(n)$ space.

\smallskip
\noindent{\it Static Algorithms.} We build on the work of Agrawal, Kolay, Lokshtanov and Saurabh~\cite{AgrawalKLS16} for {\sc Block Vertex Deletion}. However, implementing their algorithm in space $\widetilde{\OO}(k^{\OO(1)}n)$ is a non-trivial task because their $2^{\OO(k)}n^{\OO(1)}$-time algorithm for this problem uses space (even worse than) $\Omega(n^2)$. In their algorithm, they first hit small obstructions by branching (which can clearly be done in  $\widetilde{\OO}(k^{\OO(1)}n)$ space), and then they know that the resulting graph contains only $\OO(n^2)$ maximal cliques (however, there exist graphs where this bound is tight). They then construct an auxiliary bipartite graph $H$ whose vertex set consists of $V(G)$ and a vertex for each maximal clique in $G$, and where a vertex in $V(G)$ is adjacent to all cliques that contain it. It is argued that any subset $S\subseteq V(G)$ is a block graph deletion set in $G$ if and only if it is a feedback vertex set in $H$. So, the problem reduces to seeking a minimum feedback vertex in $H$ that avoids the vertices representing cliques. This  can be simply done by an algorithm for {\sc Feedback Vertex Set} (with undeletable vertices). However, in our case, we cannot even explicitly write the vertex set of $H$, which can be of size $\Omega(n^2)$. So, instead, we employ a very particular known sampling algorithm for {\sc Feedback Vertex Set}. This is an algorithm by Becker, Bar-Yehuda and Geiger~\cite{BeckerBG00}, which, after reducing the graph to another graph that has minimum degree $3$, uniformly at random selects an edge and randomly chooses an endpoint for it, and that endpoint is deleted and added to the solution. It is argued that, for any feedback vertex set, with probability at least $1/2$, the selected edge has at least one endpoint from the feedback vertex set.

To simulate this algorithm without constructing $H$, we show {\em (i)} how to compute the degree of every vertex in $V(G)$ in $H$, and {\em (ii)} how to apply the reduction rules to ensure that its minimum degree is $3$. We select a vertex in $V(G)$ (to delete and insert to our solution) with probability proportional to its degree. This already gives us a $2^{\OO(k)}n$-time $\widetilde{\OO}(k^{\OO(1)}n)$-space algorithm for {\sc Block Vertex Deletion}.

\subsection{Our Third Algorithmic Result: A Framework for Cut Problems}

We devise sparsification algorithms that, in combination with random sampling is used to deal with additional graph problems that are not encompassed by the previous theorems. Of special interest to us in this set of problems is {\sc Odd Cycle Transversal} (OCT). In this problem, one aims to decide whether there is a set of at most $k$ vertices in the given graph whose deletion leaves a bipartite graph. In other words, the ``obstruction'' set for OCT is the set of all odd cycles. Chitnis and Cormode~\cite{ChitnisC19} ask whether there exists a $\wtilde{\bigoh}(g(k)\cdot n)$-space streaming algorithm for OCT {\em even} if one were to allow $f(k)$-passes. Furthermore, they point out that already for $k=1$, the existence of such an algorithm is open (for $k=0$, this is nothing but testing bipartiteness, for which a single-pass semi-streaming algorithm is known~\cite{AhnGM12}).
Using our framework, we obtain the following result.

\begin{restatable}{theorem}{oct}
\label{thm:oct}
There is a 1-pass $\wtilde{\bigoh}(k^{\bigoh(1)}\cdot n)$-space randomized streaming algorithm for {\sc OCT} with $\wtilde{\bigoh}(3^{k}k^{\bigoh(1)} n)$-time post-processing.
\end{restatable}

Thus, Theorem~\ref{thm:oct} affirmatively answers the open problem of Chitnis and Cormode~\cite{ChitnisC19}
in its  most general form by giving a  1-pass (instead of the $f(k)$-passes they asked for) semi-streaming algorithm for {\sc Odd Cycle Transversal}.

\paragraph{Overview of our framework.} Our framework has two high-level steps. The first is generic and the second is problem-specific.

\medskip
\noindent
{\bf Step 1:} The first step of our framework is a sampling primitive of Guha, McGregor and Tench~\cite{GuhaMT15}.
 The main idea here is to sample roughly $\wtilde\bigoh(k^{\bigoh(1)})$ vertex subsets $V_1,\dots, V_\ell$ of the input graph and argue that with good probability, the subgraph defined by the union of the edges in the sampled graphs preserves important properties of the input graph. In our case, we
prove that at least one ``no-witness'' of every non-solution is preserved in the sampled subgraph.

More precisely, we show that for every set $S\subseteq V(G)$ that is disjoint from some forbidden substructure in $G$ (e.g., odd cycles) and hence not a solution, 
the subgraph of $G$ defined by the union of the sampled subgraphs 
also contains a forbidden substructure disjoint from $S$. This allows us to identify a set of $\wtilde\bigoh(k^{\bigoh(1)})$ subgraphs that can always provide a witness when a vertex set is {\em not} a solution for the problem at hand. For problems closed under taking subgraphs, this implies that solving the problem on the sampled subgraph is sufficient. 

\noindent
{\bf Step 2:} At the end of Step 1, however, we are still left with a major obstacle. That is, the sampled subgraphs may still be dense and it is far from obvious how one could ``sparsify'' them while preserving the aforementioned properties. Therefore, as the second step in our template, one needs to  also provide a problem-specific $\wtilde\bigoh(k^{\bigoh(1)}n)$-space sparsification step that allows us to reduce the number of edges we keep from each of the sampled subgraphs to $\bigoh(n)$, while ensuring that non-solutions can still be witnessed by a substructure in the union of the sparsified subgraphs. If this is achieved, one may invoke existing linear-time (static) FPT algorithms for the problem on the sparsified instance and show that a solution for this reduced instance is also a solution for the original instance.

We next briefly sketch our sparsification procedure for {\sc Odd Cycle Transversal} that leads to Theorem~\ref{thm:oct}. 
The bipartite double cover of a graph $G$ is the bipartite graph with two copies of the original vertex set $V$ (say $V_a$ and $V_b$) and two copies of each edge (for an edge $(u,v)$ in the original graph, there are two edges $(u_a,v_b)$ and $(u_b,v_a)$, where for each $z\in \{u,v\}$ and $x\in \{a,b\}$, $z_x$ is the copy of $z$ contained in $V_x$). Ahn, Guha and McGregor~\cite{AhnGM12} used this auxiliary graph in their bipartiteness-testing algorithm by observing that $G$ is bipartite precisely when its bipartite double cover has exactly twice as many connected components as $G$. In our work, we conduct a closer examination of bipartite double covers and exploit the fact that odd-length closed walks in $G$  (which must exist if $G$ is non-bipartite) that contain a vertex $v$ correspond precisely to $v_a$-$v_b$ paths in the bipartite double cover of $G$. We build upon this fact to show that the edges preserved by computing a dynamic connectivity sketch for all the sampled subgraphs together preserve odd-length closed walks in the union of these subgraphs. Finally, we argue that solving OCT with a static FPT algorithm on the sparsified instance is equivalent to solving OCT on the original instance.

\paragraph{Applying the framework to other cut problems.}  Recall that each specific application of the framework boils down to designing, for the cut problem at hand, {\bf (i)} a problem-specific sparsification subroutine and {\bf (ii)} a linear-space (or ideally, linear-time) FPT algorithm for use in post-processing. For instance, consider the classic {\sc Subset Feedback Vertex Set} and {\sc Multiway Cut} problems. We show that these two problems also fall under the same framework. For  {\sc Subset Feedback Vertex Set}, the obstruction set is the family of all cycles that contain ``terminals''---the objective of  {\sc Subset Feedback Vertex Set} is to determine whether one can choose at most $k$ vertices that intersect all cycles in the input graph $G$ that pass through at least one designated vertex, called a terminal (the input contains, in addition to $G$ and $k$, a subset of vertices called terminals). 
 In the case of {\sc Multiway Cut} the obstruction set is the family of all paths that connect pairs of ``terminals''. In this problem, the input is $G$, $k$ and a vertex subset called terminals and the objective is to determine whether a set of at most $k$ vertices intersects every path in $G$ between a pair of terminals. For both these problems, we show how to combine the sampling primitive behind our algorithm for {\sc Odd Cycle Transversal} along with problem-specific sparsifications for these (to handle Point {\bf (i)} above) and the invocation of existing linear-time FPT algorithms (to handle Point {\bf (ii)} above). Roughly speaking, in the sparsification step for {\sc Subset FVS}, we maintain a dynamic connectivity sketch for the sampled subgraphs along with the edges incident on terminals and for {\sc Multiway Cut}, we show that it is sufficient to maintain a dynamic connectivity sketch for the sampled subgraphs
  This leads us to the following two results.

\begin{restatable}{theorem}{subsetFVS}\label{thm:subsetFVS}
	There is a 1-pass $\wtilde{\bigoh}(k^{\bigoh(1)}\cdot n)$-space randomized streaming algorithm for  {\sc Subset Feedback Vertex Set} with $\wtilde{\bigoh}(2^{\bigoh(k)}\cdot n)$-time post-processing.
\end{restatable}

\begin{restatable}{theorem}{multiwaycut}\label{thm:multiwaycut}
	There is a 1-pass $\wtilde{\bigoh}(k^{\bigoh(1)}\cdot n)$-space randomized streaming algorithm for  {\sc Multiway Cut} with $\wtilde{\bigoh}(4^kk^{\bigoh(1)}\cdot n)$-time post-processing.
\end{restatable}

\subsection{Refinement of the Class {\semips}}\label{sec:modelInOurContributions}

   Our algorithmic results demonstrate that the notion of fixed-parameter semi-streaming algorithms  
is widely applicable to parameterized versions of graph optimization problems.
As a result, we obtain an associated natural complexity class (which we call {\fptsemips}) that is the analogue of the class {\sf FPT} (fixed-parameter tractable problems) in the semi-streaming setting. In the same spirit, we also introduce notions of {\em semi-streaming kernelization} (and {\em semi-streaming compression}), where the algorithm uses $\widetilde{\OO}(k^{\bigoh(1)}n)$ space overall and polynomial time at each edge update and in post-processing, and eventually outputs an equivalent instance of size $f(k)$ of the same problem (or a different problem, respectively). This output is called a kernel (or a compression, respectively).
We show that
our definitions are robust and faithfully reflect their analogues in the static setting. That is, we show (see Theorem~\ref{thm:FPT-semi-PS-Kernels-equivalence}) that:

\medskip

\begin{tcolorbox}[colback=green!5!white,colframe=white!100!black]
\begin{center}
{\bf A (decidable) parameterized  problem has an FPSS algorithm if and only if it has a semi-streaming kernelization.}\end{center}
\end{tcolorbox}

\medskip

Notice that {\fptsemips} is contained in {\sf FPT}. We show that there are problems in {\sf FPT} that are not in {\fptsemips} by proving a $\Omega(n^2)$-space lower bound for the {\sc Chordal Vertex Deletion} problem even when $k=0$, that is, for the problem of recognizing whether a given graph is chordal.

In summary, our contributions in this paper are grounded in a novel exploration of parameterized streaming algorithms. By developing unified algorithms through meta-theorems, proposing a complexity class and demonstrating its richness through containment of numerous well-studied parameterized graph problems, our work significantly extends the boundaries of the field. Our advances also naturally point to a large number of open questions in parameterized streaming (see Section~\ref{sec:conclusion} for a discussion).


\section{Related Work}
 Recently, Chakrabarti et.al~\cite{ChakrabartiG0V20} studied \emph{vertex-ordering problems} in digraphs such as testing {\sc Acyclicity} and computing {\sc Feedback Vertex (Arc) sets}. In brief, they gave space lower bounds of $\Omega(n^{1+1/p})$ in general digraphs and $\Omega(n/p)$ in Tournaments for these problems, where $p$ is the number of passes. Moreover, they analyze the post-processing complexity of their algorithm and show that both the time and space complexity of their post-processing are essentially optimal.
We refer to ~\cite{mcgregor2014graph,ChakrabartiG0V20} for more details and further references.

Fafianie and Kratsch~\cite{fafianie2014streaming} introduced a notion of streaming kernelization, which is an algorithm that takes polynomial time, $\Omega(p(k) \log n)$ space and outputs an equivalent instance whose size is bounded polynomially in the parameter $k$, where $p$ is some polynomial function of $k$, and $n$ is the input size.
In this setting they gave streaming kernelizations for certain problems such as {\sc $d$-Hitting Set} in $1$-pass and {\sc Edge Dominating Set} in 2-passes. On the other hand, they  obtained $\Omega(m)$-space lower-bounds for 1-pass streaming kernelization for many problems including {\sc Feedback Vertex Set}, {\sc Odd Cycle Transversal}, {\sc Cluster Vertex Deletion} and interestingly, {\sc Edge Dominating Set}. Here $m$ denotes the number of edges. Further, they obtained multi-pass lower bounds for {\sc Cluster Editing} and {\sc Chordal Completion} -- they gave a $\Omega(n/t)$-space lower bound for  $t$-pass algorithms.  Chitnis et.al~\cite{ChitnisCHM15,ChitnisCEHMMV16} studied parameterized streaming algorithms for {\sc Vertex Cover} and the more general {\sc $d$-Hitting Set} problem and gave a $\tilde\OO(k^d)$-space streaming kernelization even on dynamic streams.
We remark that in \cite{ChitnisCHM15},  a $\tilde\OO(nk)$-space 1-pass streaming FPT algorithm for {\sc Feedback Vertex Set} was presented, along with an $\Omega(n)$ lower-bound.

In the previously discussed work of Chitnis and Cormode~\cite{ChitnisC19}, they proposed a hierarchy of complexity classes for parameterized streaming. In particular they defined the classes {\sf FPS}, {\sf SubPs}, {\sf SemiPS}, {\sf SupPS} and {\sf BrutePS} that bound space usage to $\tilde\OO(f(k))$, $f(k)\cdot o(n)$, $\tilde\OO(f(k) \cdot n)$, $\tilde\OO(f(k) \cdot n^{1+\epsilon})$ for $0 \leq \epsilon < 1$ and $\tilde\OO(n^2)$, respectively. In their setting there is no restriction on time, i.e. an unbounded amount of time may be spent during the streaming phase and post-processing (although this was not exploited for their positive results).  Then it is clear that every decidable graph problem lies in {\sf BrutePS}, since we can store the entire graph and solve it via brute force computation. They show that certain problems such as {\sc Dominating Set} and {\sc Girth} are tight for {\sf BrutePS} in 1-pass, i.e. they require $\Omega(n^2)$ space. 
%
Similarly, they show that {\sc Feedback Vertex Set} and {\sc $k$-Path} are tight for {\semips}, i.e. require $\tilde\OO(f(k) \cdot n)$ space. They also proved a general result placing all minor-bidimensional problems~\cite{DemaineH05} in the class {\semips}. Their argument implies that these problems are also typically in the class {\fptsemips} depending on whether or not they have a static linear-space fixed-parameter algorithm. However, the problems we consider in this paper are {\em not} minor-bidimensional and so, this theorem is inapplicable in our case. 
Furthermore, as mentioned earlier, {\sc Vertex Cover} and {\sc $d$-Hitting Set} lie in {\sf FPS}, i.e. they require $\tilde\OO(f(k))$ space for some function $f$ of $k$ alone.

Let us remark that, in light of this result for {\sc $d$-Hitting Set}, one could be tempted to conclude that {\probVDH}  (for $\cH$ characterized by finite forbidden induced subgraphs) can be reduced to an instance of {\sc $d$-Hitting Set} in the natural way (i.e., forbidden subgraphs of the input graph correspond to sets in the {\sc $d$-Hitting Set} instance).
However, in the streaming setting, with only $\tilde\OO(f(k)\cdot n)$-space available, this reduction is no longer feasible since we are unable to store all edges and enumerate all obstructions.
As an illustrative example, consider the question of recognizing whether an input graph is $k$-vertex deletions away from being a cluster graph in the semi-streaming setting. Recall that cluster graphs are exactly the induced-$P_3$-free graphs. Since the edges forming a $P_3$ may arrive far apart in time, one cannot hope to  simply reduce the problem to a 3-Hitting Set instance and solve it, unless one stores the entire graph.
 Therefore, in spite of the $\tilde\OO(k^d)$-space streaming kernelization for {\sc $d$-Hitting Set}~\cite{ChitnisCEHMMV16}, this  approach cannot be used to resolve such graph modification problems. In this work we give a novel data structure (see Section~\ref{sec:coveringBoundedObstructions} and the discussion following the proof of Theorem~\ref{thm:Hcovering}) using which one can indeed reduce the given instance of any {\probVDH} problem (for such $\cH$) to an instance of {\sc $d$-Hitting Set} that is bounded polynomially in $k$, using $\wtilde\bigoh(k^{\bigoh(1)}n)$ space.

Finally, it is important to mention the
 work of Feigenbaum et.al.~\cite{FeigenbaumKSV02}, who also considered time complexity in the context of streaming algorithms. In particular, they defined a class called PASST$(s,t)$ (stands for probably approximately correct streaming space complexity $s$ and time complexity $t$) that additionally bounds the per-item processing time during the streaming phase. However, in this work also the post-processing time is not bounded. Chitnis and Cormode (see Remark 43~\cite{ChitnisC19}) also consider restricting the per-item processing time and the post-processing time to only $n^{\OO(1)}$.  However, while this is suitable for streaming kernelization algorithms~\cite{fafianie2014streaming}, it is too restrictive when asking for parameterized streaming algorithms.

\section{Preliminaries}\label{sec:preliminaries}
When clear from the context, we use $m$ and $n$ to denote the number of edges and vertices respectively.

\subsection{Splitters and Separating Families}
\begin{definition}{\rm \cite{Naor:1995:SND:795662.796315}}
\label{def:splitterFamilies}
	Consider $n,k,\ell\in {\mathbb N}$ such that $k\leq n,\ell$.  An {\em $(n, k, \ell)$-splitter} $\cF$ is a family of functions from $[n]$ to $[\ell]$ such that for every set $S \subseteq [n]$ of size $k$, there exists a function $f \in \cF$ that is injective on $S$.
	\end{definition}

	\begin{proposition}[Splitter construction] {\rm \cite{Naor:1995:SND:795662.796315}}
	\label{prop:splitterConstruction}
	For every $n, k \geq 1$, one can construct an $(n, k, k^2)$- splitter of size $\OO(k^6 \log k \log n)$ in time (and space) $k^{\bigoh(1)}n\log n$.

	\end{proposition}

	We also require the following construction of a family of subsets over a set.
\begin{definition}[$(n,k,\ell)$-Separating Family] \label{def:separatingFamily}
    Let $U$ be a universe of $n$ elements. An \emph{$(n,k,\ell)$-separating family} over $U$ is a family $\cal F$ of subsets of $U$, such that for any pair of disjoint subsets $A$ and $B$ of $U$ where $|A| \leq k$ and $|B| \leq \ell$, there exists $F \in \cal F$ such that $A \cap F = \emptyset$ and $B \subseteq F$.
\end{definition}

\begin{lemma}
    Let $U$ be a universe of $n$ elements, and let $\ell \leq k$ be two integers. There exists an $(n,k,\ell)$-separating family of size $\OO((k+\ell)^{2\ell + 6} \log (k + \ell) \log n)$ that can be enumerated in $\OO((k + \ell)^{2\ell + \OO(1)} \cdot n \log n) $ time (and space).
\end{lemma}
\begin{proof}
    Let $q = \ell + k$. We begin by enumerating an $(n,q,q^2)$-splitter $\cal H$ over $U$. Recall that $\cal H$ is a collection of functions from $U$ to $\{1,2, \ldots, q^2\}$ such that, for any subset $X$ of $U$ of size at most $q$, there exists a function $h \in \cal H$ such that $h(i) \neq h(j)$ for any two distinct $i,j \in X$. From Proposition~\ref{prop:splitterConstruction}, we have a construction of $\cal H$ containing $\OO(q^6 \log q \log n)$ functions in $\OO(q^{\OO(1)} \cdot n \log n)$ time (and space).
    For each function $h \in \cal H$, we enumerate the following family of subsets of $U$. For each subset $Y$ of size $\ell$ of $\{1,2,\ldots,q^2\}$ we output $h^{-1}(Y) \subseteq U$. Note that, for each $h \in \cal H$ we enumerate at most $q^{2\ell}$ subsets of $U$, and we denote them by ${\cal F}_h$. We define the $(n,k,\ell)$-separating family $\cal F$ as the union of these subsets, i.e. ${\cal F} = \bigcup_{h \in {\cal H}} {\cal F}_h$. Observe that $\cal F$ contains $\OO(q^{2\ell + 6} \log q \log n)$ subsets of $U$, and it can be enumerated in $\OO(q^{2\ell} \cdot q^{\OO(1)} \cdot n \log n)$ time (and space). 

    It only remains to argue that $\cal F$ is indeed a $(n,k,\ell)$-separating family over $U$. Towards this, consider any two disjoint subsets $A$ and $B$ of $U$ of size at most $k$ and $\ell$, respectively. Then there is a function $h \in H$ such that $h(i) \neq h(j)$ for all $i\neq j \in A \cup B$.
    Let $Y = \{h(i) \mid i \in B\}$, then observe that $h^{-1}(Y) \in F$ is disjoint from $A$ and contains $B$. This holds for any choice of the subsets $A$ and $B$, and hence $\cal F$ is a $(n,k,\ell)$-separating family over $U$.
\end{proof}

We require the following corollary of the above lemma, where we wish to separate pairs of vertices from vertex subsets of size at most $k$ in a graph.
\begin{corollary}\label{cor:graph-separating-family}
    Let $G$ be a graph on $n$ vertices. Then, there is a family $\cal F$ of $\OO((k+2)^{10} \log(k+2) \log n)$ subsets of $V(G)$ such that for any pair of vertices $u, v$ and any subset $X$ of at most $k$ vertices  where $u,v \notin X$, there exists $F \in \cal F$ such that $u,v \in F$ and $X \cap F = \emptyset$. Such a family $\cal F$ can be enumerated in $\OO(k^{\OO(1)}\cdot n \log n)$ time (and space).
\end{corollary}


\subsection{Semi-streaming Algorithms}
We assume standard word RAM model of computation with words of bitlength $\bigoh (\log n)$, where $n$ is the vertex count of the input graph. Vertex labels are assumed to fit within single machine words and can be operated on in $\bigoh(1)$-time.


\paragraph{Insertion-only streams.}
 Let $P$ be a problem parameterized by  $k\in {\mathbb N}$. Let $(I,k)$ be an instance of $P$ that has an input $X=\{x_1,\dots, x_i,\dots, x_m\}$ with input size $|I| = m$. Let $\cS$ be a stream of ${\sf INSERT}(x_i)$ (i.e., the insertion of an element $x_i$) operations of underlying instance $(I, k)$. In particular, the stream $\cS$ is a permutation $X' = \{x'_1 , \dots, x'_i , \dots,  x'_m \}$ for $x'_i \in  X$ of an input $X$.

\paragraph{Dynamic/turnstile streams.}
Let $P$ be a problem parameterized by  $k\in {\mathbb N}$. Let $(I,k)$ be an instance of $P$ that has an input $X=\{x_1,\dots, x_i,\dots, x_m\}$ with input size $|I| = m$. We say that stream $\cS$ is a turnstile parameterized stream if $\cS$ is a stream of ${\sf INSERT}(x_i)$ (i.e., the insertion of an element $x_i$) and ${\sf DELETE}(x_i)$ (i.e., the deletion of an element $x_i$) operations applying to the underlying instance $(I,k)$ of $P$.

Throughout this paper, when we refer to a (semi-)streaming algorithm without explicitly mentioning the number of passes required, then the number of passes is 1. Similarly, when the type of input stream (whether it is insertion-only or turnstile) is not explicitly mentioned, then it is to be understood that the input stream being referred to is a turnstile stream. We also assume that in either stream, when the input is an instance $(I,k)$ of a parameterized graph problem $P$, the algorithm is aware of the vertex set of the input graph. Moreover, $k$ is first given in the stream in unary and never deleted. It is only then that the stream begins providing $I$. Note that the elements in $I$ can potentially undergo deletion and reinstertions, depending on the input model.

\begin{definition}[$k$-sparse recovery algorithm] A {\em $k$-sparse recovery algorithm} is a data structure which accepts insertions and deletions of elements from $[n]$ and recovers all elements of the stream if, at the recovery time, the number of elements stored in it is at most $k$.
\end{definition}

We require the following result of Barkay, Porat and Shalem~\cite{BarkayPS15} which we have specialized to our setting.

\begin{proposition}[Lemma 9, \cite{BarkayPS15}]\label{prop:kSparseRecovery}
	There is a deterministic structure, SRS with parameter $k$, denoted by SRS$_k$, that that keeps a sketch of stream $I$ (comprising insertions and deletions of elements from $[n]$) and can recover all of $I$'s elements if $I$ contains at most $k$ distinct elements. It uses $\bigoh(k\log(n))$ bits of space, and is updated in $\bigoh(\log^2 k)$ operations amortized.
\end{proposition}

In other words, Barkay, Porat and Shalem~\cite{BarkayPS15} have given a deterministic $k$-sparse recovery algorithm that uses $\wtilde{\bigoh}(k)$ space.

We also require the following version of the dynamic connectivity result of Ahn, Guha and McGregor~\cite{AhnGM12}.

\begin{proposition}{\rm \cite{AhnGM12}}\label{prop:dynamicConnectivitySketch}
For every $c\in {\mathbb N}$, there exists a  1-pass $\wtilde{\bigoh}(n)$-space streaming algorithm in the turnstile model that, in post-processing, constructs a spanning forest of the input graph with probability at least $1-1/n^c$ in time $\wtilde{\bigoh}(n)$.
\end{proposition}

\subsection{Graphs and Graph Classes}\label{sec:graphClassPrelims}
Fix a set of graphs $\cR$. Any  graph that does not contain a graph in $\cR$ as an induced subgraph is called an $\cR$-free graph.
%
 We use $d_{\cR}$ to denote the maximum number of vertices among the graphs in $\cR$.
We say that $S$ is an {\em $\cR$-deletion set} of a graph $G$ if $G-S$ is $\cR$-free. We say that a subgraph of $G$ is an {\em $\cR$-subgraph} if it is isomorphic to a graph in $\cR$.

In the rest of this section, we recall the various graph classes we consider in this paper and characterizations of these classes that we use in our algorithms.

{\em Acyclic tournaments} are precisely those tournaments that exclude directed triangles~\cite{diestel2000graphtheory}. {\em Split graphs} are graphs whose vertex set can be partitioned into two disjoint sets, one of which is a clique and the other is an independent set. Split graphs are also characterized by the exclusion of $\{2K_2,C_4,C_5\}$~\cite{Golumbicbook}. {\em Cluster graphs}
are graphs where every connected component is a clique. These are also characterized by the exclusion of $\{P_3\}$~\cite{Golumbicbook}.
%
%
{\em Threshold graphs} are precisely those graphs that are characterized by the exclusion of $\{2K_2,C_4,P_4\}$ respectively~\cite{Golumbicbook}.
{\em Block graph} are graphs in which every biconnected component is a clique. It is known that they are characterized by the exclusion of
$\{D_4, C_{\ell + 4,\, \ell \geq 0} \}$~\cite{brandstadt1999graph}. Here $D_4 = K_4 - e$ for some edge $e \in E(K_4)$, and $C_{\ell+4}$ denotes an induced cycle on $\ell + 4$ vertices.

\section{From Recognition to  {\probVDH} with Finitely Many Obstructions}\label{sec:coveringBoundedObstructions}
Recall that for a class $\cH$ of graphs, in the {\probVDH} problem, the input is a graph $G$ and integer $k$ and the objective is to decide whether there is a set $S\subseteq V(G)$ of size at most $k$ such that $G-S\in \cH$. The standard parameterization is $k$, the size of the solution. %
In this section, we only focus on the case where $\cH$ is characterized by a finite set of forbidden induced subgraphs.

We begin by formally capturing the notion of an efficient semi-streaming algorithm that recognizes graphs in $\cH$.

\begin{definition}\label{def:sparseRecognitionAlgorithms}
A class of graphs $\cH$ admits a {\em $p$-pass recognition algorithm} if there exists a $p$-pass $\widetilde{\OO}(n)$-space streaming algorithm with polynomially bounded time  between edge updates and in post-processing that, given a graph $G$, correctly concludes whether or not $G\in \cH$.
\end{definition}

\subsection{The Fixed-parameter Semi-streaming Algorithm for {\probVDH}}


We are now ready to prove our first meta theorem.

%

\Hcovering*

\begin{proof}
%
%
Let $\cR$ denote the finite set of graphs excluded by graphs in $\cH$ as induced subgraphs.
In what follows, let $d=d_\cR$ denote the size of the largest graph in the fixed set of graphs $\cR$.
Note that $d$ is a constant in this setting.
Recall that we know $V(G)$ (and hence $n$) and $d$ apriori, while $k$ is provided at the beginning of the stream.
 Let $\alpha=max\{dk,k+d\}$ and $\beta=d^2 \cdot \binom{\alpha}{d}\cdot \alpha^{c}\ln n+d$ (where $c$ is the constant in the $\bigoh(\cdot)$ notation in the size of the splitter family as given in Proposition~\ref{prop:splitterConstruction}).

 \begin{enumerate}[(i)]
 	\item We construct an $(n,\alpha,\alpha^2)$-splitter family $\cF_1$ of  size at most $\alpha^{c}\ln n$, which can be constructed in time $\alpha^{\bigoh(1)}\cdot n\ln n$ (Proposition~\ref{prop:splitterConstruction}).
 	\item We construct an $(n,d+1,(d+1)^2)$-splitter family $\cF_2$ of  size $d^{\bigoh(1)}\ln n$, which can be constructed in time $d^{\bigoh(1)}\cdot n\ln n$.
 	\item We construct an $(n,\beta,\beta^2)$-splitter family $\cF_3$ of  size at most $\beta^{c}\ln n$, which can be constructed in time $\beta^{\bigoh(1)}\cdot n\ln n$.
 \end{enumerate}

Before describing our (post-)processing steps, we define some useful notation.

 \begin{itemize}

 \item   	 For  every  $f_1\in \cF_1$, and $J \subseteq [\alpha^2]$, we denote by $G_{f_1,J}$ the graph $G[f_1^\inv(J)]$, i.e., the subgraph induced by those vertices of $G$ whose image under $f_1$ is contained in $J$.
 \item   Similarly, for every  $f_3\in \cF_3$, and $J\subseteq [\beta^2]$, we denote by $G_{f_3,J}$ the graph $G[f_3^\inv(J)]$, i.e., the subgraph induced by those vertices of $G$ whose image under $f_3$ is contained in $J$.
 \item   	 For  every $f_2\in \cF_2$, $i\in [(d+1)^2]$ and graph $G_{f_1,J}$, we denote by
 $G_{f_2,i,f_1,J}$ the graph $G_{f_1,J}-f_2^\inv(i)$. That is, the graph obtained from $G_{f_1,J}$ by deleting those vertices whose image under $f_2$ is $i$.
 \item Let $\cX$ denote the set: $$\{G_{f_1,J}\mid f_1\in \cF_1, J\subseteq [\alpha^2], |J|\leq d\}$$ $$\bigcup \{G_{f_2,i,f_1,J}\mid f_1\in \cF_1,f_2\in \cF_2,i\in [(d+1)^2],  J\subseteq [\alpha^2], |J|\leq d\}$$ $$\bigcup \{G_{f_3,J}\mid f_3\in \cF_3,  J\subseteq [\beta^2], |J|\leq d\}.$$
 Notice that $|\cX|\leq |\cF_1|\cdot d\cdot \binom{\alpha^2}{d}  + |\cF_1|\cdot |\cF_2|\cdot d\cdot \binom{\alpha^2}{d}\cdot (d+1)^2 + |\cF_3|\cdot d\cdot \binom{\beta^2}{d}=\wtilde{\bigoh}((kd)^{\bigoh(d)})=\wtilde{\bigoh}(k^{\bigoh(1)})$.
  	\end{itemize}


\smallskip
 \noindent
 {\bf Processing the stream:}
 We process the graph stream by running the $p$-pass recognition algorithm (call this algorithm, $\cA$) assumed in the premise of the theorem for each of the graphs in $\cX$. Since $|\cX|=\wtilde{\bigoh}(k^{\bigoh(1)})$, the space used by our algorithm is $\wtilde{\bigoh}(k^{\bigoh(1)}\cdot n)$ as required. Moreover, for each graph $G'\in \cX$, the premise guarantees that we can decide whether $G'\in \cH$ by running a polynomial-time post-processing algorithm on the data-structure constructed by Algorithm $\cA$, denoted by $\langle G'\rangle$. We call this post-processing algorithm, Algorithm $\cB$. We say that $\cB(\langle G'\rangle)=\top$ if $G'\in \cH$ and $\cB(\langle G'\rangle)=\bot$ otherwise.

We are now ready to describe our fixed-parameter post-processing algorithm.

\smallskip
 \noindent
 {\bf Post-processing:}
 For every $f_1\in \cF_1$ we construct a vertex set $Z_{f_1}$ as follows. For every $v\in V(G)$,  we add $v$ to the set $Z_{f_1}$ if and only if there exists $J\subseteq [\alpha^2]$ of size at most $d$ such that (i) $v\in V(G_{f_1,J})$, (ii)  $\cB(\langle G_{f_1,J}\rangle)=\bot$ and (iii) for every $f_2\in \cF_2$, $\cB(\langle G_{f_2,f_2(v),f_1,J}\rangle)=\top$.
  Clearly, computing $Z_{f_1}$ takes  $k^{\bigoh(1)}\cdot n^{\bigoh(1)}$-time. Therefore, computing the set $\cZ=\{Z_{f_1}\mid f_1\in \cF_1\}$ can be done in $k^{\bigoh(1)}\cdot n^{\bigoh(1)}$-time and additional $\wtilde{\bigoh}(k^{\bigoh(1)}n)$-space.

  Let $Z^\star=\cup_{Z\in \cZ}Z$.  We now construct a $d$-set system $\EE$ with universe $Z^\star$ as follows. For every $L\subseteq Z^\star$ of size at most $d$, we add $L$ to $\EE$ if and only if there is an $f_3\in \cF_3$ and a set $J\subseteq [\beta^2]$ of size at most $d$ such that $J$ is disjoint from $f_3(Z^\star\setminus L)$ and $\cB(\langle G_{f_3,J} \rangle)=\bot$. This completes the construction of $\EE$. Notice that the number of sets in $\EE$ is bounded by $d\cdot |Z^\star|^{d}$.

  Finally, we execute the standard linear-space $O^*(d^k)$-time branching FPT algorithm for {\sc $d$-Hitting Set} on the instance $(Z^\star,\EE,k)$ (see, for example, \cite{CyganFKLMPPS15}) and return the answer returned by this execution.
  This completes the description of the algorithm.
   We now argue the correctness of this algorithm.

 We begin by bounding the size of $Z^\star$.

 \begin{claim}\label{clm:thmHcoveringBoundOnZStar}
 	$|Z^\star|\leq d^2 \cdot \binom{\alpha^2}{d}\cdot \alpha^{c}\ln n$.
 \end{claim}

 \begin{proof}
 To prove the claim, it suffices to prove that for every $f_1\in \cF_1$ and $J\subseteq [\alpha^2]$ such that $|J|\leq d$, at most $d$ vertices of $G_{f_1,J}$ can be added to the set $Z_{f_1}$. The claim then follows from the bound of $d\cdot \binom{\alpha^2}{d}$ on possible values of $J$ and the bound of $\alpha^c\ln n$ on the size of $|\cF_1|$.

 We say that a vertex $v$ {\em is contributed to} $Z_{f_1}$ by a set $J\subseteq [\alpha^2]$ of size at most $d$ if  $v\in V(G_{f_1,J})$, $\cB(\langle G_{f_1,J}\rangle)=\bot$, and  for every $f_2\in \cF_2$, $\cB(\langle G_{f_2,f_2(v),f_1,J}\rangle)=\top$.
  We now argue that each $J$ contributes at most $d$ vertices to $Z_{f_1}$. To do so, we show that a vertex of $G_{f_1,J}$ is contributed to $Z_{f_1}$ by $J$ precisely if it intersects every $\cR$-subgraph of $G_{f_1,J}$, that is, $G_{f_1,J}-v\in \cH$. Since each graph in $\cR$ contains at most $d$ vertices, there can be at most $d$ such vertices.

 Suppose to the contrary that a vertex $v\in V(G_{f_1,J})$ is contributed to $Z_{f_1}$ by $J$ and $G_{f_1,J}-v$ is not in $\cH$. Let $H$ be an $\cR$-subgraph of $G_{f_1,J}-v$ and let $f_2\in \cF_2$ be a function which is injective on $V(H)\cup \{v\}$.
 Since $\cF_2$ is an $(n,(d+1),(d+1)^2)$-splitter, such an $f_2$ exists. However, notice that $\cB(\langle G_{f_2,f_2(v),f_1,J}\rangle)=\bot$, contradicting our assumption that $v$ is contributed to $Z_{f_1}$ by $J$. This completes the proof of the claim.
 \end{proof}

 We next observe that if $(G,k)$ is a yes-instance of {\probVDH}, then it is sufficient to look for our solution within $Z^\star$.

 \begin{claim}\label{clm:thmHcoveringZStarContainsAllSolutions}
 	 Every minimal solution of size at most $k$ in $G$ is contained in $Z^\star$.
 \end{claim}

 \begin{proof}
 	Let $S$ be an inclusionwise-minimal solution of size at most $k$ (also called an $\cR$-deletion set). That is, $G-S\in \cH$.
 	 	  Let $\cT$ denote a set of obstructions witnessing the minimality of $S$.  In other words, $\cT$ is a set of $\cR$-subgraphs of $G$ such that for every graph in $\cT$, there is a unique vertex of $S$ that it intersects. Since $S$ is a minimal $\cR$-deletion set, such a set of $|S|$ $\cR$-subgraphs must exist. For each $v\in S$, we denote by $\cT_v$ the unique graph in $\cT$ that contains $v$. Let $V(\cT)$ denote the union of the vertex sets of the subgraphs in $\cT$ and notice that $|V(\cT)|\leq dk$.

  Now, consider a function $f_1\in \cF_1$ that is injective on $V(\cT)$. Since $\cF_1$ is an $(n,dk,(dk)^2)$-splitter family, such an $f_1$ exists. We argue that $S\subseteq Z_{f_1}$.
  For each $v\in S$, let $J_v\subseteq [(dk)^2]$ denote the set $\{f_1(u)\mid u\in  V(\cT_v)\}$. Notice that $|J_v|\leq d$.   Moreover, it must be the case that for every $v\in S$ , the graph $G[f_1^\inv(J_v)]$ has at least one minimal $\cR$-deletion set of size exactly 1 (which is the vertex $v$) and at most $d$ distinct minimal $\cR$-deletion sets of size exactly 1. The former is a consequence of the fact that $G[f_1^\inv(J_v)]$ is disjoint from $S\setminus \{v\}$ and the latter is a consequence of the fact that
the graphs in $\cR$ have size bounded by $d$.

Notice that for each $v\in S$, since $v$ is an $\cR$-deletion set of $G[f_1^\inv(J_v)]$, it follows that for every $f_2\in \cF_2$, the graph $G_{f_2,f_2(v),f_1,J}\in \cH$ and therefore, it must be the case that $\cB(\langle G_{f_2,f_2(v),f_1,J}\rangle)=\top$.
In other words, the vertex $v$ is contributed to $Z_{f_1}$ by $J_v$ (see proof of Claim~\ref{clm:thmHcoveringBoundOnZStar} for the definition).  Hence, we conclude that $S\subseteq Z_{f_1}\subseteq Z^\star$.
 	This completes the proof of the claim.
 \end{proof}

 \begin{claim}\label{clm:thmHcoveringEquivalence}
 A set $S\subseteq V(G)$ is a solution for the instance $(G,k)$ of {\probVDH} if and only if  it is a solution  for the {\sc $d$-Hitting Set} instance $(Z^\star,\EE,k)$.

 \end{claim}

 \begin{proof}
 Suppose that $(G,k)$ is a yes-instance of {\probVDH} and let $S\subseteq V(G)$ be a solution. By Claim~\ref{clm:thmHcoveringZStarContainsAllSolutions}, we have that $S\subseteq Z^\star$. Now, suppose that $S$ is not a solution for the {\sc $d$-Hitting Set} instance $(Z^\star,\EE,k)$ and let $L\in \EE$ be a set disjoint from $S$. Then, by the construction of $\EE$, we have that there is a function $f_3$ and a set $J\subseteq [\beta^2]$ disjoint from $f_3(S)$ such that $\cB(\langle G_{f_3,J})=\bot$. This implies that there is a subgraph of $G$, i.e., $G_{f_3,J}$ that contains an $\cR$-subgraph and is disjoint from $S$, a contradiction to our assumption that $S$ is an $\cR$-deletion set in $G$.

Conversely, suppose that $S$ is a solution for the $d$-{\sc Hitting Set} instance $(Z^\star,\EE,k)$, but not an $\cR$-deletion set in $G$. Then, there is an $\cR$-subgraph $H$ in $G-S$. Let $L=V(H)\cap Z^\star$. Notice that  $|L|\leq d$ since $|V(H)|\leq d$. Moreover, we claim that in the construction of $\EE$, we would have added the set $L$ to $\EE$. To see this, first observe that $Z^\star\cup V(H)$ has size at most $\beta$ and so, there is a function $f_3\in \cF_3$ which is injective on $Z^\star\cup V(H)$. Therefore, for $J=f_3(V(H))$, we have that $J$ is disjoint from $f_3(Z^\star\setminus L)$ and $\cB(\langle G_{f_3,J}\rangle)=\bot$, ensuring that $L$ is contained in $\EE$.
Since $S$ is a solution for the $d$-{\sc Hitting Set} instance $(Z^\star,\EE,k)$, it follows that $S$ intersects $L$ and hence also intersects $V(H)$, a contradiction to our assumption that $S$ is not an $\cR$-deletion set in $G$.
This completes the proof of the claim.
 \end{proof}

This completes the proof of Theorem~\ref{thm:Hcovering}. 
\end{proof}

\paragraph{Further insights that can be drawn from our proof of Theorem~\ref{thm:Hcovering}. }
First of all, notice that that in the proof of this theorem, if there is a randomized $p$-pass recognition algorithm for $\cH$ where the post-processing succeeds with high probability (i.e., with probability at least $1-1/n^c$ for any $c\in {\mathbb N}$), then one obtains a randomized $p$-pass  $\wtilde{\bigoh}(k^{\bigoh(1)}n)$-space streaming algorithm for
{\probVDH} with $2^{\bigoh(k)}n^{\bigoh(1)}$-time post-processing and a
randomized $p$-pass $\wtilde{\bigoh}(k^{\bigoh(1)}\cdot n)$-space polynomial streaming compression for {\probVDH}, both of which succeed with high probability. Similarly, if we have a recognition algorithm for $\cH$ only on insertion-only streams, then we can still use Theorem~\ref{thm:Hcovering}, with the caveat that the resulting algorithms for {\probVDH} also only work for insertion-only streams.

Secondly, our algorithm can be used to obtain a {\em polynomial compression} from {\probVDH} to {\sc $d$-Hitting Set}. Indeed, if $d^k<n$, then the fixed-parameter post-processing algorithm we use can be seen to solves the problem in polynomial time, allowing us to produce a trivial equivalent instance of {\sc $d$-Hitting Set} as the output. On the other hand, if $d^k>n$, then the size of the {\sc $d$-Hitting Set} instance $(Z^\star,\EE,k)$ is already bounded by $d\cdot |Z^\star|^d\leq d \cdot d^2 \cdot \binom{\alpha^2}{d}\cdot \alpha^{c}\ln n = k^{\bigoh(1)}$ as required. 

\subsection{Corollaries}

\begin{lemma}\label{lem:sparseRecogExamples}
For each of the following graph classes $\cH$, there is a finite family $\cR$  such that $\cH$ is precisely the class of graph that exclude graphs in $\cR$ as induced subgraphs: Acyclic tournaments,  split graphs, threshold graphs and cluster graphs. Moreover, the first three graph classes have a deterministic 1-pass recognition algorithms and cluster graphs have a randomized 1-pass recognition algorithm that succeeds with probability at least $1-1/n^c$ for any given $c>0$.
\end{lemma}

\begin{proof}
As discussed in Section~\ref{sec:graphClassPrelims}, it is known that acyclic tournaments are precisely those tournaments that exclude directed triangles~\cite{diestel2000graphtheory}. Split graphs, cluster graphs and threshold graphs are characterised by the exclusion of $\{2K_2,C_4,C_5\}, \{P_3\},$ and $\{2k_2,C_4,P_4\}$ respectively~\cite{Golumbicbook}.  Furthermore, acyclic tournaments~\cite{diestel2000graphtheory}, split graphs and threshold graphs~\cite{Golumbicbook} have characterizations through their degree sequence. That is, it is sufficient to know the degrees of each vertex in order to be able to recognize graphs from these classes. As storing the degree of each vertex can be trivially done in $\wtilde{\bigoh}(1)$ space, the algorithm follows.

 For cluster graphs, it is straightforward to see that they can be recognized from a spanning forest and the degree sequence. Indeed, a graph $G$ is a cluster graph if and only if for every tree $T$ in a spanning forest of $G$ and every vertex $v\in V(T)$, the degree of $v$ in $G$ is precisely $|V(T)|-1$. In order to compute a spanning forest of the input graph with high probability, we use Proposition~\ref{prop:dynamicConnectivitySketch}.
 This completes the proof of the lemma.
\end{proof}

As immediate corollaries of Theorem~\ref{thm:Hcovering} and Lemma~\ref{lem:sparseRecogExamples}, we obtain in one shot the first 1-pass $\wtilde{\bigoh}(k^{\bigoh(1)}\cdot n)$-space streaming algorithms 
for several vertex-deletion problems such as {\sc Feedback Vertex Set on Tournaments}, {\sc Split Vertex Deletion}, {\sc Threshold Vertex Deletion} and {\sc Cluster Vertex Deletion} (for the last of which the algorithms are randomized). Moreover, our data structures enable the post-processing for these algorithms to be done by a single-exponential fixed-parameter algorithm parameterized by the solution size. Furthermore, notice that the
post-processing time used by the 1-pass recognition algorithms for the above specific graph classes (in Lemma~\ref{lem:sparseRecogExamples}) is $\wtilde{\bigoh}(n)$. Combining this fact with a closer inspection of the proof of Theorem~\ref{thm:Hcovering} indicates that the polynomial factor in the running times of the fixed-parameter post-processing routines for these problems is in fact only $\wtilde{\bigoh}(n)$. Thus, we have the following results.

\paragraph{Further impact of our techniques in the form of new algorithms for {\probVDH} in the static setting.} As discussed above, a closer  examination of the proof of Theorem~\ref{thm:Hcovering} implies an FPT algorithm for {\probVDH} where the polynomial factor in the running time is $n\cdot T(n)$ where $T(n)$ is the time required to {\em recognize} a graph in $\cH$.
In the standard $2^{\bigoh(k)}n^{\bigoh(1)}$ branching algorithm for {\probVDH}, this factor is at least $T'(n)$ where $T'(n)$ is the time required  to {\em compute an obstruction}, i.e., an $\cR$-subgraph in the input graph (assuming $\cR$ is the set of forbidden induced subgraphs for $\cH$). The computation of obstructions for various graph classes is typically achieved through a {\em certifying recognition algorithm}, i.e., an algorithm that returns an obstruction if it concludes that the input is not in the graph class. However, designing certifying recognition algorithms that are nearly (or just as) efficient as a normal recognition algorithm is a non-trivial task. Indeed, Heggernes and Kratsch~\cite{HeggernesK07} note that usually different and sometimes deep insights are needed to produce useful certificates.
Therefore, our approach allows one to improve the polynomial dependence in the standard $2^{\bigoh(k)}n^{\bigoh(1)}$ branching algorithm for {\probVDH} in cases where recognition of $\cH$ is efficient but the time required to {\em find} an obstruction is worse by a factor of $n$ or more.

\section{From Reconstruction to  {\probVDH} with Infinitely Many Obstructions}\label{sec:reconstruction}


In this section, we present a meta theorem to reduce the design of streaming algorithms for deletion to hereditary graph classes to the design of reconstruction algorithms for these classes.

Towards the presentation of our theorem, we first need to define the meaning of reconstruction. This is done in the following two definitions.

\begin{definition}
Given an $n$-vertex graph $G$, a {\em succinct representation of $G$}  is a data structure that uses $\widetilde{\OO}(n)$ space and supports an $\widetilde{\OO}(n)$-space polynomial-time procedure that, given two vertices $u,v\in V(G)$, correctly answers whether $\{u,v\}\in E(G)$.
\end{definition}

\begin{definition}
For an integer $p \geq 0$, a class of graphs $\cH$ admits a {\em $p$-pass reconstruction algorithm} if there exists a $p$-pass $\widetilde{\OO}(n)$-space streaming algorithm with polynomial post-processing time that, given a graph $G$, correctly concludes whether $G\in \cH$, and in case $G\in{\cH}$, outputs a succinct representation of~$G$.
\end{definition}

We are now ready to prove our meta theorem.


\reductionLemma*

\begin{proof}
Let $G$ denote the input graph, which is presented to us as a stream of edges. Let $n = |V(G)|$. We say that a vertex subset $S \subseteq V(G)$ is a \emph{solution} to $(G,k)$ if $|S| \leq k$ and $G - S \in \cH$.
Let ${\cal F} = \{F_1, F_2, \ldots, F_t\}$ be an $(n,k,2)$-separating family over $V(G)$. Note that such a family contains $t = k^{\OO(1)} \log n$ vertex subsets, and it can be constructed in $\OO(k^{\OO(1)} n \log n)$ time and space (Corollary~\ref{cor:graph-separating-family}).
Given the family $\cal F$, for each $1 \leq i \leq t$, let $G_i = G[F_i]$.
For brevity, let $\cal A$ denote the algorithm of Item~\ref{item:reduction1}, which can reconstruct a graph in the class $\cH$ in $p$-passes. And let $\cal B$ denote the static algorithm for {\probVDH} of Item~\ref{item:reduction2}.
Our streaming algorithm for {\probVDH} is as follows.
\begin{itemize}

    \item \textbf{Streaming Phase.} In the streaming phase, we construct a collection of data-structures from the input stream of edges using the reconstruction algorithm $\cal A$. We first construct the family $\cal F$ of subsets of $V(G)$. Then for each graph $G_i = G[F_i]$, where $1 \leq i \leq t$, we apply algorithm $\cal A$, in parallel. At the end of the stream, for each $G_i$, the algorithm either concludes that $G_i \in \cH$ and outputs a succinct representation of it, or else it concludes that $G_i \notin \cH$. Let $D_i$ denote the data-structure output by $\cal A$ whenever $G_i \in \cH$, and $D_i = \emptyset$ otherwise. Observe that each data-structure $D_i$ requires $\tilde\OO(n)$ space, and there are at most $t = k^{\OO(1)}\log n$ of them. We store the family $\cal F$ and the collection of these data-structures $\{D_i\}_{1 \leq i \leq t}$. Note that, we require $p$-passes and $\tilde\OO(k^{\OO(1)}\cdot g(k) \cdot n)$ space in the streaming phase.

    \item \textbf{Post-processing Phase.} Let us now describe the post-processing phase of our algorithm. Let us define a graph $\tilde{G}$ as follows: $V(\tilde{G}) = V(G)$ and $E(\tilde{G}) = \bigcup_{1 \leq i \leq t \,:\, D_i \neq \emptyset} E(G_i)$.
    Observe that $\tilde{G}$ is a subgraph of $G$, although it may not be an induced subgraph. Note that the graph $\tilde{G}$ is not explicitly constructed, but a succinct representation of $\tilde{G}$ is obtained from the data-structures $\{D_i \mid 1 \leq i \leq t\}$ as follows: given any two vertices $u,v \in V(G)$, to determine if $\{u,v\} \in E(\tilde{G})$,
    we query every $D_i$ for the edge $\{u,v\}$; if any one of these queries succeed we output that $\{u,v\}$ is an edge in $\tilde{G}$, otherwise we output that $\{u,v\}$ is a non-edge of $\tilde{G}$. Next, let $E' = \{\{u,v\} \in V \times V \mid u \neq v \text{ and } \forall~ 1 \leq i \leq t \text{ such that } D_i \neq \emptyset,~ |F_i \cap \{u,v\}| \leq 1\}$.
    Let $G' = (V(G), E')$, and note that it is not necessarily a subgraph of $G$ (although it can be shown to contain every edge in $E(G) \setminus E(\tilde{G})$). The graph $G'$ is also not explicitly constructed, but it is reconstructed from the vertex subsets $\{F_i \mid 1 \leq i \leq t, ~D_i \neq \emptyset\}$.

    To construct a solution to $(G,k)$, we do the following.
    We first enumerate every minimal vertex cover of size at most $k$ in the graph $G'$.  Note that there are at most $2^k$ minimal vertex covers of $G'$ of size at most $k$~\cite{CyganFKLMPPS15}. This is accomplished by a simple branching algorithm, that processes the edges of $G'$ one by one.
    This algorithm runs in $2^k \cdot n^{\OO(1)}$ time and $\widetilde\OO(n)$ space.
    For each vertex cover $X$ of $G'$ of size at most $k$ produced by the above enumeration, we apply the algorithm $\cal B$ to $(\tilde{G} - X, k - |X|)$, and either obtain a solution $\tilde{S}$ to $(\tilde{G}-X, k-|X|)$ or no such solution exists.
    If we do indeed find a solution $\tilde{S}$ to $(\tilde{G}-X, k-|X|)$, then we output $\tilde{S} \cup X$ as the solution to $(G,k)$. Otherwise, for every $X$ the algorithm $\cal B$ fails to find a solution to $(G-X,k-|X|)$, and  we output that there is no solution to $(G,k)$. This completes the description of the algorithm.
\end{itemize}

It clear that the above algorithm requires $p$-passes, and used $\OO(n \cdot g(k) \cdot k^{\OO(1)})$ space and $(2^k \cdot f(k)) \cdot n^{\OO(1)}$ time.
It only remains to argue the correctness of our algorithm.
Towards this, first recall that $\cal F$ has the following property: for any pair of vertices $u,v$ and any $S \subseteq V(G)$ such that $u,v \notin V(G)$ and $|S| \leq k$, there is a subset $F_i \in {\cal F}$ such that $u,v \in F_i$ and $S \cap F_i = \emptyset$. Next, we have the following claim: If $S$ is a solution to $(G,k)$ and $e = \{u,v\} \in E'$, then $|S \cap \{u,v\}| \geq 1$. That is, we claim that $S$ is a (not necessarily minimal) vertex cover for $G'$.
We prove this claim using the properties of the $(n,k,2)$-separation family $\cal F$ over $V(G)$. Suppose the claim is false, and let $\{u,v\} \in E'$ such that $S \cap \{u,v\} = \emptyset$. Then there exists $F_i \in {\cal F}$ such that $\{u,v\} \subseteq F_i$ and $S \cap F_i = \emptyset$. Observe that $G_i = G[F_i]$ is an induced subgraph of $G - S$, and $G - S \in \cH$. Since $\cH$ is hereditary, $G_i \in \cH$ as well. Therefore, $D_i \neq \emptyset$, and hence by the construction of $E'$, we have $\{u,v\} \notin E'$. This is a contradiction.

Next, consider any vertex cover $X$ of $G'$.
We claim that $\tilde{G} - X = G - X$.
First observe that, as $V(G) = V(\tilde{G})$, we have $V(G) \setminus X = V(\tilde{G}) \setminus X$.
Next suppose that there are  $u,v \in V(\tilde{G}) \setminus X$ such that $\{u,v\} \in E(G) \setminus E(\tilde{G})$. Then by the definition of $\tilde{G}$, there is no vertex subset $F_i \in \cal F$ such that $G_i \in \cH$ and $u,v \in F_i$.
In other words, for any $F_j \in \cal F$ such that $u,v \in F_j$, we have $G_j \notin \cH$ and hence $D_j = \emptyset$.
Hence, by definition $\{u,v\} \in E'$, and therefore $X \cap \{u, v\} \neq \emptyset$. But this is again a contradiction.
Hence we conclude $\tilde{G} - X = G - X$.

Now, let us argue that $(G,k)$ has a solution  if and only if our algorithm concludes that there is a solution to $(G,k)$. In the forward direction, consider a solution $S^\star$ to $(G,k)$. Let $Y \subseteq S^\star$ be a minimal subset that is a vertex-cover of $G'$, and note that $|Y| \leq k$. Then as $\cH$ is hereditary, $S^\star - Y$ is a solution to the instance $(G - Y, k - |Y|)$.
Recall that our post-processing algorithm enumerates all minimal vertex covers of $G'$ of size at most $k$, and in particular $Y$. Then for the vertex cover $Y$ of $G'$, it computes a solution to $(\tilde{G} - Y, k - |Y|)$ using algorithm $\cal B$. Here, recall that $\tilde{G} - Y = G - Y$ and hence it admits a solution $\tilde{S} = S^\star \setminus Y$ of size at most $k - |Y|$. Therefore, by invoking algorithm $\cal B$ on $(\tilde{G} - Y, k -|Y|)$, we correctly conclude that $(G,k)$ admits a solution.
In the reverse direction, suppose that our algorithm outputs $\tilde{S} \cup X$ as a solution to $(G,k)$, where $X$ is a vertex cover of $G'$ and $\tilde{S}$ is a solution to $(\tilde{G}-X, k - |X|)$. Recall that $\tilde{G} - X = G - X$, and hence $G - (\tilde{S} \cup X) = \tilde{G} - (\tilde{S} \cup X)$. Hence it follows that $G - (\tilde{S} \cup X) \in \cH$, i.e. $\tilde{S} \cup X$ is a solution $(G,k)$. This concludes the proof of this lemma.
\end{proof}

Finally, we remark that the results of this section also hold for digraphs (with the same proof).

\subsection{Proper Interval Vertex Deletion in $\OO(\log^2 n)$ Passes}

Formally, the class of proper interval graphs can be defined as follows.

\begin{definition}[Proper Interval Graph]
A graph $G$ is a {\em proper interval graph} if there exists a function $f$ that assigns to each vertex in $G$ an open interval of unit length on the real line such that every two vertices in $G$ are adjacent if and only if their intervals intersect. Such a function $f$ is called a {\em representation}, and given a vertex $v\in V(G)$, we let $\mathsf{begin}_f(v)$ and $\mathsf{end}_f(v)$ denote the beginning and end of the interval assigned to $v$ by $f$ (as if it was closed).

Equivalently, the demand that each interval will be of unit length can  be replaced by the demand that no interval will properly contain another interval.\footnote{To be more precise, when intervals are required to have unit length, then the graph is called a {\em unit interval graph}, and when intervals should not properly contain one another, it is called a {\em proper interval graph}. However, these two notions are known to be equivalent.}
\end{definition}

We remark that whenever we have a representation, we can slightly perturb it so that no two vertices will be assigned the same interval, or, more generally, no two endpoints (start or end) of intervals will coincide. Further, given a representation of the second ``type'' (where intervals can have different lengths), it is possible to construct a representation of the first ``type'' with the same ordering of the starting and ending points of all intervals assigned, and that every representation of the first ``type'' is also of the second ``type''.

The purpose of this section is to prove the following theorems. The second theorem implies the first but requires additional arguments.

\begin{theorem}\label{thm:piv}
{\sc Proper Interval Vertex Deletion} admits an $\OO(\log^2 n)$-pass semi-streaming algorithm with $2^{\OO(k)}\cdot n^{\OO(1)}$ post-processing time.
\end{theorem}

To prove this lemma, we present the reconstruction and post-processing algorithms for the class of proper interval graphs as required by Theorem  \ref{lem:reductionLemma}. 

\subsubsection{Reconstruction Algorithm}

The purpose of this section is to prove the following lemma.

\begin{lemma}\label{lem:pivReconstruction}
The class of proper interval graphs admits an $\OO(\log^2 n)$-pass reconstruction algorithm.
\end{lemma}

We first observe that it suffices to focus on the class of {\em connected} proper interval graphs. Indeed, this follows by running an algorithm for the connected case on all connected components simultaneously, and answering ``yes'' if and only if the answer to all components is ``yes''; in case the answer is ``yes'', the output representation is the union of the representations of the connected components.

\begin{observation}\label{obs:pivReduceToConnected}
The class of connected proper interval graphs admits an $x$-pass reconstruction algorithm with polynomial post-processing time, then so does the class of proper interval graphs.
\end{observation}

In particular, dealing with connected proper interval graphs allows us to make use of the following result regarding the number of connected components left after the removal of any closed neighborhood.

\begin{lemma}\label{lem:conPivTwoComps}
Let $G$ be a connected proper interval graph with representation $f$. For any $v\in V(G)$, the graph $G-N[v]$ consists of at most two connected components: one on the set of vertices whose intervals end at or before the interval of $v$ starts (according to $f$), and the other on the set of vertices whose intervals start at or after the interval of $v$ ends.
\end{lemma}

\begin{proof}
Let $L$ be the subgraph of $G$ induced by the set of vertices whose intervals end at or before the interval of $v$ starts. Let $R$ be the subgraph of $G$ induced the set of vertices whose intervals start at or after the interval of $v$ ends. Note that every vertex not in $V(L)\cup V(R)$ has an interval that intersects that of $v$, and the set that consists of these vertices is precisely $N[v]$. Further, the intervals of two vertices that belong to different graphs among $L$ and $R$ do not intersect, and hence they are not neighbors. So, to conclude the proof, it remains to argue that each graph among $L$ and $R$ is connected. We will only show this for $L$ as the proof for $R$ is symmetric. Consider two vertices $u,v\in V(L)$. Then, because $G$ is connected, there exists a (simple) path $P$ between $u$ and $v$ in $G$. If $P$ belongs to $L$, we are done. Else, $P$ can be rewritten as $u-P_1-x-y-P_2-z-w-P_3-v$ (possibly $y=z$ and $P_1,P_2,P_3$ can be empty) where $V(P_1)\cup V(P_3)\cup\{x,w\}\subseteq V(L)$ and $y,z\in N[v]$. So, {\em (i)} $\mathsf{start}_f(x)<\mathsf{start}_f(y)<\mathsf{end}_f(x)<\mathsf{start}_f(v)$, and {\em (ii)} $\mathsf{start}_f(w)<\mathsf{start}_f(z)<\mathsf{end}_f(w)<\mathsf{start}_f(v)$.

Without loss of generality, suppose that {\em (iii)} $\mathsf{start}_f(z)<\mathsf{start}_f(y)$. We claim that $\mathsf{start}_f(x)<\mathsf{start}_f(z)<\mathsf{end}_f(x)$, which will imply (by {\em (ii)}) that $f(x)$ and $f(w)$ intersect. Clearly, $\mathsf{start}_f(x)<\mathsf{start}_f(z)$ because $f(x)$ ends before $f(v)$ starts while $f(z)$ does not. Moreover, $\mathsf{start}_f(z)<\mathsf{end}_f(x)$, because otherwise, by {\em (iii)}, $\mathsf{end}_f(x)<\mathsf{start}_f(z)<\mathsf{start}_f(y)$ which is a contradiction to {\em (i)}. So, $u-P_1-x-w-P_3-v$ is a path within $L$.
\end{proof}

In order to proceed, we need to define an annotated version of the reconstruction problem, where we only seek proper interval graphs with representations that comply with a given order on the intervals assigned to vertices. This annotation will be encountered since a step (which will be encapsulated by a lemma ahead) of our algorithm will reduce a problem instance to smaller annotated instances (and will reduce these smaller annotated instances to even smaller annotated instances, and so on, based on divide on conquer).

\begin{definition}
The {\sc Annotated Proper Interval Reconstruction} problem is defined as follows. Given a graph $G$ and a partial ordering $<$ on $V(G)$, decide whether $G$ is a proper interval graph that admits a representation $f$ where for every different $u,v\in V(G)$, no endpoint of $f(u)$ coincides with an endpoint of $f(v)$, and such that if $u<v$, then we have that $\mathsf{begin}_f(u)<\mathsf{begin}_f(v)$ (such a representation is said to {\em comply with $<$}). Moreover, in case $G$ is such a graph, output a succinct representation of $G$ that is the permutation on $\{\mathsf{begin}_f(v):v\in V(G)\}\cup\{\mathsf{end}_f(v):v\in V(G)\}$ corresponding to $f$.\footnote{Here, we mean that intervals are scanned from left to right. Note that this is indeed a succinct representation because it takes space $\widetilde{\OO}(n)$, and to check whether two vertices $u,v$ are neighbors, we just need to check whether $\mathsf{begin}_f(v)<\mathsf{begin}_f(u)<\mathsf{end}_f(v)$ or $\mathsf{begin}_f(u)<\mathsf{begin}_f(v)<\mathsf{end}_f(u)$.}
\end{definition}

From now on, as we justify immediately, our objective will be to prove the following lemma.

\begin{lemma}\label{lem:annotatedReconstruction}
The {\sc Annotated Proper Interval Reconstruction} problem admits an $\OO(\log^2 n)$-pass algorithm.
\end{lemma}

Indeed, to prove Lemma \ref{lem:pivReconstruction}, it suffices to prove Lemma \ref{lem:annotatedReconstruction} because we can choose the order $<$ to define all vertices as incomparable, and thereby make $<$ immaterial.

The proof of Lemma \ref{lem:annotatedReconstruction} will be based on an inductive argument. The main part of the proof is given by Lemma \ref{lem:annotatedReconstructionStep} ahead. Before this, we state the following definition and lemmas that will be used in its proof.

\begin{definition}
Let $G$ be an $n$-vertex graph. A vertex $v$ is a {\em middle vertex} if the size of every connected component of $G-N[v]$ is at most $\frac{9}{10}n$.
\end{definition}

\begin{lemma}\label{lem:probChooseMid}
Let $G$ be a connected proper interval graph. Then, by selecting a vertex $v\in V(G)$ uniformly at random, $v$ is a middle vertex with probability at least $\frac{4}{5}$.
\end{lemma}

\begin{proof}
Let $f$ be a representation of $G$, and $n=|V(G)|$. For every vertex $v$, let $L_v$ and $R_v$ be the first and second connected components (which may be empty) defined in Lemma \ref{lem:conPivTwoComps}. Order $V(G)=\{v_1,v_2,\ldots,v_n\}$ so that $f(v_i)$ starts at or before $f(v_j)$ starts. Then, for any $i\in\{1,2,\ldots,n\}$, $|V(L_{v_i})|\leq i-1$ and $|V(R_{v_i})|\leq n-i$. So, whenever $\frac{1}{10}\leq i\leq \frac{9}{10}$, $|V(L_{v_i})|\leq \frac{9}{10}n$ and $|V(R_{v_i})|\leq \frac{9}{10}n$, and then $v_i$ is a middle vertex. So, there are at least $\frac{4}{5}n$ mid vertices, which implies the desired probability.
\end{proof}

\begin{lemma}\label{lem:verifyReconstruction}
There exists a 1-pass $\widetilde{\OO}(n)$-space polynomial-time algorithm that given two graphs $G$ and $H$ on the same vertex set $V$ and a succinct representation of $H$ (where the representation is stored in memory and not part of the stream), decides whether $G$ is equal to $H$.\footnote{That is, for any pair of vertices $u,v\in V$, we have that $u,v$ are adjacent in $G$ if and only if they are adjacent in $H$.}
\end{lemma}

\begin{proof}
The algorithm works as follows. Using one pass on the stream, it computes the degree of every vertex in $G$, and also for each edge in $G$ (when it appears in the stream), it checks that it belongs to $H$ using the succinct representation and if not it returns that $G\neq H$. Afterwards, it computes the degree of every vertex in $H$ (using the succinct representation) and checks that it equals the degree computed for $G$, and if not then it returns that $G\neq H$. At the end, it returns that $G=H$.

Clearly, this is a 1-pass $\widetilde{\OO}(n)$-space polynomial-time algorithm. Moreover, correctness follows by observing that $G=H$ if and only if $E(G)\subseteq E(H)$ and the degree of every vertex is the same in $G$ and $H$.
\end{proof}

We now turn to state and proof the aforementioned Lemma \ref{lem:annotatedReconstructionStep}.

\begin{lemma}\label{lem:annotatedReconstructionStep}
Let $n,s\in\mathbb{N}$. Suppose that {\sc Annotated Proper Interval Reconstruction} admits a $5\cdot s\cdot \log_{\frac{10}{9}} \widetilde{n}$-pass algorithm on $\widetilde{n}$-vertex graphs where $\widetilde{n}\leq \frac{9}{10}n$ with success probability at least $1-(\frac{1}{5})^s\cdot \widetilde{n}\log_{\frac{10}{9}}\widetilde{n}$. Then, {\sc Annotated Proper Interval Reconstruction} admits a $5\cdot s\cdot \log_{\frac{10}{9}} n$-pass algorithm on $n$-vertex graphs with success probability at least $1-(\frac{1}{5})^s\cdot n\log_{\frac{10}{9}} n$.\footnote{We remark that the constant $5$ was not optimized. For example, just by sacrificing modularity and not checking the validity of our reconstruction here (at each step) but only once in the end, it reduces to $4$.}
\end{lemma}

\begin{proof}
Let $\mathfrak{A}$ denote the algorithm guaranteed by the supposition of the lemma.

\medskip
\noindent{\bf Algorithm.} The algorithm works as follows.
\begin{enumerate}
\item For $i=1,2,\ldots,s$:
	\begin{enumerate}
	\item Select uniformly at random a vertex $v^\star\in V(G)$.
	\item Go over the stream once to compute $N(v^\star)$.
	\item Go over the stream once using a connectivity sketch to compute the connected components of $G-N[v^\star]$. By Lemma \ref{lem:conPivTwoComps}, there are at most two such components, which we denote by $C_1$ and $C_2$ (where possibly one or both of them are empty).
	\item If $|V(C_1)|\leq\frac{9}{10}n$ and $|V(C_2)|\leq\frac{9}{10}n$, then go directly to Step \ref{step:proceed}.
	\end{enumerate}
\item Return failure.
\item\label{step:proceed} Rename $C_1$ and $C_2$ as $L$ and $R$ (where $C_1$ can be either) such that for every pair of comparable vertices $u,w\in V(G)$ such that $u<w$, at least one of the following holds: {\em (i)} $u\in V(L)$; {\em (ii)} $w\in V(R)$; {\em (iii)} $u,w\in N[v]$. If this is not possible, then return ``no-instance''. Also, denote $M=G[N[v]]$.

\item Go over the stream once to compute for every $u\in V(G)$: $d_L(u)=|N(u)\cap V(L)|$, $d_R(u)=|N(u)\cap V(R)|$ and $d_M(u)=|N(u)\cap V(M)|$.

\item\label{step:identifyAB} Here, we consider several cases:
\begin{enumerate}
\item {\bf $V(L)\neq\emptyset$.} Among all vertices $u\in M$ with maximum $d_L(u)$ and that are smallest by $<$, let $a$ be one with minimum $d_M(u)$ (if there is more than one choice, then select one  arbitrarily).  If $a$ cannot be chosen (i.e., there is no vertex that has maximum $d_L(u)$ and is simultaneously smallest by $<$ among all vertices in $M$), then return ``no-instance''. Afterwards:
\begin{enumerate}
\item\label{step:chooseb1} Go over the stream once to compute $N(a)$.
\item\label{step:chooseb2}  If $M\setminus N[a]\neq\emptyset$, then among all vertices $u\in M\setminus N[a]$ with maximum $d_R(u)$  and that are largest by $<$, let $b$ be one with minimum $d_M(u)$ (if there is more than one choice, then select one arbitrarily).  If $b$ cannot be chosen (i.e., there is no vertex that has maximum $d_R(u)$ and is simultaneously largest by $<$  among all vertices in $M\setminus N[a]$), then return ``no-instance''.
\item\label{step:chooseb3} Else, among all vertices $u\in M$ with maximum $d_R(u)$ and that are largest by $<$, choose some vertex $b$ arbitrarily.  If $b$ cannot be chosen (i.e., there is no vertex that has maximum $d_R(u)$ and is simultaneously largest by $<$  among all vertices in $M$), then return ``no-instance''.
\end{enumerate}

\item {\bf $V(L)=\emptyset$ and $V(R)\neq\emptyset$.} Among all vertices $u\in M$ with maximum $d_R(u)$ and that are largest by $<$, let $b$ be one with minimum $d_M(u)$ (if there is more than one choice, then select one arbitrarily).  If $b$ cannot be chosen (i.e., there is no vertex that has maximum $d_R(u)$ and is simultaneously largest by $<$ among all vertices in $M$), then return ``no-instance''. Afterwards:
\begin{enumerate}
\item Go over the stream once to compute $N(b)$.
\item If $M\setminus N[b]\neq\emptyset$, then among all vertices $u\in M\setminus N[b]$ with maximum $d_L(u)$ and that are smallest by $<$, let $a$ be one with minimum $d_M(u)$ (if there is more than one choice, then select one arbitrarily).  If $a$ cannot be chosen (i.e., there is no vertex that has maximum $d_R(u)$ and is simultaneously smallest by $<$  among all vertices in $M\setminus N[b]$), then return ``no-instance''.
\item Else, among all vertices $u\in M$ with maximum $d_R(u)$ and that are smallest by $<$, choose some vertex $a$ arbitrarily.  If $a$ cannot be chosen (i.e., there is a conflict between having minimum $d_M(u)$ and being smallest by $<$), then return ``no-instance''.
\end{enumerate}

\item {\bf $V(L)=V(R)=\emptyset$.} Among all vertices $u\in M$ that are smallest by $<$, let $a$ be one of minimum $d_M(u)$ (if there is more than one choice, then select one arbitrarily). Repeat Steps \ref{step:chooseb1}, \ref{step:chooseb2} and \ref{step:chooseb3} to define $b$.
\end{enumerate}

\item\label{step:recursiveCalls} Call the algorithm in the supposition of the lemma on $L$ with $<_L$ defined as follows: for all $u,v\in V(L)$, if $u<v$ or $d_M(u)<d_M(v)$, then $u<_Lv$. If a conflict arises in the definition of $<_L$ or the algorithm returns ``no-instance'', then return ``no-instance'', and else let $\mathsf{Perm}_L$ be the permutation it returns (supposedly corresponding to some representation $g_L$ of $L$ that complies with $<_L$). Simultaneously,\footnote{That is, use one pass on $G$ to simulate one pass on $L$ as well as one pass on $R$.} call the algorithm in the supposition of the lemma on $R$ with $<_R$ defined as follows: for all $u,v\in V(R)$, if $u<v$ or $d_M(u)>d_M(v)$, then $u<_Rv$. If a  conflict arises in the definition of $<_R$ or the algorithm ``no-instance'', then return ``no-instance'', and else let $\mathsf{Perm}_R$ be the permutation it returns (supposedly corresponding to some representation $g_R$ of $R$ that complies with $<_R$).

\item\label{step:orderStart} Now, for some (unknown and possibly incorrect) representation $f$, we define an ordering on starting points, $\mathsf{begin}_f(\cdot)$, as follows. First, the starting points of the intervals of all vertices in $L$ and $R$ are ordered exactly as in $\mathsf{Perm}_L$ and $\mathsf{Perm}_R$, and all starting points of the intervals of vertices in $L$ appear before those in $M$, and for those in $M$, they appear before those in $R$. Internally, the starting points of the vertices in $M$ are ordered as follows. For every $u,v\in M$, if at least one of the following conditions is true, then $\mathsf{begin}_f(u)<\mathsf{begin}_f(v)$:
	\begin{enumerate}
	\item $u<v$ (where $<$ is the input partial order on $V(G)$);
	\item $d_L(u)>d_L(v)$;
	\item $d_R(u)<d_R(v)$;
	\item $u,v\in N(a)$ and $d_M(u)<d_M(v)$;
	\item $u,v\notin N(a)$ and $d_M(u)>d_M(v)$;
	\item $u\in N(a)$ and $v\notin N(a)$.
	\end{enumerate}
In case a conflict occurs, that is, the above conditions imply that both $\mathsf{begin}_f(u)<\mathsf{begin}_f(v)$ and $\mathsf{begin}_f(v)<\mathsf{begin}_f(u)$ for some $u,v\in V(G)$, then we return ``no-instance''. Besides this, ties are broken arbitrarily.

\item\label{step:orderEndLR} We proceed by inserting the ending points of the intervals of all vertices in $V(L)\cup V(R)$ within the ordering of the starting points. First, the ending points of the intervals of all vertices in $L$ with no neighbors in $M$ are  ordered exactly as in $\mathsf{Perm}_L$ (and are inserted before the starting points of the intervals of all vertices in $M$). Second, the ending points of the intervals of all vertices in $R$ are ordered exactly as in $\mathsf{Perm}_R$ (and are inserted after the starting points of the intervals of all vertices in $M$). Third, for every vertex $v\in V(L)$ with neighbours in $M$, insert $\mathsf{end}_f(v)$ between $\mathsf{begin}_f(x)$ and $\mathsf{begin}_f(y)$ where $x$ is the $d_M(v)$-th vertex in $M$ according to the already established ordering of starting points of intervals, and $y$ is the vertex whose starting point is ordered immediately after that of $x$ (which might not exist, in which case $\mathsf{end}_f(v)$ is just placed after all starting points of intervals of vertices in $M$ and before all those of vertices in $R$).

\item\label{step:orderEndM} Lastly, we insert the ending points of the intervals of all vertices in $V(M)$ within the ordering of the starting points. For every vertex $v\in V(M)$ with $d_R(v)>0$, insert $\mathsf{end}_f(v)$ between $\mathsf{begin}_f(x)$ and $\mathsf{begin}_f(y)$ where $x$ is the $d_R(v)$-th vertex in $R$ according to the already established ordering of starting points of intervals, and $y$ is the vertex whose starting point is ordered immediately after that of $x$ (which might not exist, in which case $\mathsf{end}_f(v)$ is just placed after all starting points). Additionally, for every vertex $v\in V(M)$ with $d_R(v)=0$, if the starting point of $f(v)$ was placed after the starting point of $f(v^\star)$, then place the ending point of $f(v)$ after the starting point of all vertices in $M$ but before those of all vertices in $R$; else, insert $\mathsf{end}_f(v)$ between $\mathsf{begin}_f(x)$ and $\mathsf{begin}_f(y)$ where $x$ is the $d_M(v)$-th vertex in $M$ according to the already established ordering of starting points of intervals, and $y$ is the vertex whose starting point is ordered immediately after that of $x$ (which might not exist, in which case $\mathsf{end}_f(v)$ is just placed after all starting points of intervals of vertices in $M$ and before all those of vertices in $R$).

\item\label{step:orderEndInternal} To complete the ordering of all starting points and ending points, it remains only to define the ordering between ending points that were inserted between the same two starting points---for every two such ending points $\mathsf{end}_f(u)$ and $\mathsf{end}_f(v)$, we order $\mathsf{end}_f(u)$ before $\mathsf{end}_f(v)$ if and only $\mathsf{begin}_f(u)$ before $\mathsf{begin}_f(v)$.

\item\label{step:end1} Check whether in the constructed succinct representation (that is the ordering of starting and ending points) no interval properly contains another. If this check fails, return ``no-instance''.

\item\label{step:end2} Go over the stream once using the algorithm in Lemma \ref{lem:verifyReconstruction} to verify whether the graph given by the constructed succinct representation yields a graph equal to $G$. If the answer is yes, then return ``yes-instance'' and the representation, and otherwise return ``no-instance''.
\end{enumerate}

This completes the description of the algorithm.

\medskip
\noindent{\bf Space, Time, Number of Passes and Success Probability.}  Clearly, the space complexity of the algorithm is $\widetilde{O}(n)$ and it works in polynomial time. Regarding the number of passes, observe that the number of passes it makes is at most \[\begin{array}{l}
2s+3+\max_{\frac{1}{10}n\leq \ell\leq \frac{9}{10}n}5s\log_{\frac{10}{9}}\ell\\
\leq 5s(1+\max_{\frac{1}{10}n\leq\ell\leq\frac{9}{10}n}\log_{\frac{10}{9}}\ell)\\
= 5s(1+\log_{\frac{10}{9}}\frac{9}{10}n)\\
=5s\log_{\frac{10}{9}}n.
\end{array}\]

Moreover, by Lemma \ref{lem:probChooseMid} and union bound, its failure probability is at most
\[\begin{array}{l}
\smallskip
(\frac{1}{5})^s+\max_{\frac{1}{10}n\leq\ell\leq\frac{9}{10}n}((\frac{1}{5})^s\cdot \ell\log_{\frac{10}{9}} \ell + (\frac{1}{5})^s\cdot (n-\ell)\log_{\frac{10}{9}} (n-\ell))\\

\smallskip
\leq (\frac{1}{5})^s+(\frac{1}{5})^s\cdot \max_{\frac{1}{10}n\leq\ell\leq\frac{9}{10}n}(\ell\log_{\frac{10}{9}} \frac{9}{10}n + (n-\ell)\log_{\frac{10}{9}}\frac{9}{10}n)\\

= (\frac{1}{5})^s+(\frac{1}{5})^s\cdot n\log_{\frac{10}{9}}\frac{9}{10}n \leq(\frac{1}{5})^s\cdot n\log_{\frac{10}{9}} n.
\end{array}\]
So, the success probability is at least $1-(\frac{1}{5})^s\cdot n\log_{\frac{10}{9}} n$ as claimed.

\medskip
\noindent{\bf Correctness: Reverse Direction.} For correctness, first observe that by Steps \ref{step:end1} and \ref{step:end2}, if the algorithm returns ``yes-instance'', then this is correct and so does the representation (because the answer and representation are verified in these steps). The difficult part in the proof is to show the other direction.

\medskip
\noindent{\bf Correctness: Forward Direction.}  In the rest of the proof, we assume that $G$ is a proper interval graph that admits a representation that complies with $<$, and aim to prove that in this case the algorithm returns ``yes-instance'' (in which case, due to our verification, the representation will be correct). To this end, we need to show that there exists some representation for $G$ that complies with $<$ and with our ordering of starting and ending points---then, necessarily, the checks in the last two steps succeed and the algorithm  returns ``yes-instance''. For this, we will have a sequence of claims that will justify the decisions made in the different steps of the algorithm. We begin by justifying Step \ref{step:proceed} with the following claim, whose correctness follows directly from the definition of a representation, and by observing that taking the mirror image of a representation also gives a representation, but which may not comply with $<$.

\begin{claim}\label{claim:proceed}
The algorithm does not answer ``no-instance'' in Step~\ref{step:proceed}, and there exists a representation of $G$ that complies with $<$ and where all intervals of vertices in $L$ begin before all intervals of vertices in $M$ begin, and all intervals of vertices in $M$ begin before all intervals of vertices in $R$ begin.
\end{claim}

Next, we claim that the vertices $a$ and $b$ have leftmost and rightmost intervals among the vertices in $M$, addressing Step \ref{step:identifyAB}.

\begin{claim}\label{claim:identifyAB}
The algorithm does not answer ``no-instance'' in Step~\ref{step:proceed}, and there exists a representation of $G$ that satisfies the conditions in Claim \ref{claim:proceed} and where $a$ (resp.~$b$) is the vertex in $M$ whose interval's starting point is leftmost (resp.~rightmost).
\end{claim}

\begin{figure}
\center
  \fbox{\includegraphics[scale=0.8]{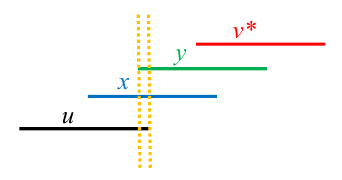}}
  \caption{The sub-interval where $y$ and $y$ intersect is flanked by dashed yellow lines, and it is a sub-interval of $x$.}
  \label{fig:NeighborhoodContainment}
\end{figure}

\begin{proof}
We only consider now representations that satisfy the conditions in Claim \ref{claim:proceed} (and we know that at least one such representation exists). Clearly, given two vertices $x,y$ in $M$ such that the interval of $x$ (in any of the representations we consider) starts before that of $y$, it must be that  {\em (i)} $x$ has at least as many neighbors in $L$ as $y$, {\em (ii)} $y$ has at least as many neighbors as $x$ in $R$, and {\em (iii)} either $x,y$ are incomparable or $x<y$. Moreover, vertices in $M$ that have the same number of neighbors in $L$ (resp.~$R$) have also the exact same neighbors in $L$ (resp.~$R$)---to see this, suppose that the interval of $y$ (resp.~$x$) starts before that of $x$ (resp.~$y$), let $u$ be a neighbor of $y$ (resp.~$x$)  in $L$, and notice that every intersection point of the intervals of $y$ (resp.~$x$) and $u$ also belongs to the interval of $x$ (resp.~$y$); see Fig.~\ref{fig:NeighborhoodContainment} (the case of $R$ is symmetric). We also note that no vertex can have neighbors in both $L$ and $R$.

We proceed to present arguments that justify the choices we make based on $d_M(\cdot)$. For this purpose, let $a'$ and $b'$ be the vertices in $M$ having the leftmost and rightmost intervals in some representation that does not assign the same interval to two different vertices. Then, notice that $V(M)\subseteq N[a']\cup N[b']$ (see Fig.~\ref{fig:DegreeFunc}).  Moreover, notice that for every $u\in N(a')\cap V(M)$, we have that $N[a']\cap V(M)\subseteq N[u]\cap V(M)$---indeed, the interval of every vertex in $N(a')\cap V(M)$ must intersect the (open) right endpoint of the interval of $a'$, and the interval of $u$ intersects it as well, and hence the intervals of the neighbor and of $u$ intersect each other (see Fig.~\ref{fig:DegreeFunc}). Symmetrically, for every $u\in N(b')\cap V(M)$, we have that $N(b')\cap V(M)\subseteq N(u)\cap V(M)$. In particular, this means that when we look at vertices from left to right according to the starting point of their intervals, $d_M(\cdot)$ is non-decreasing, and after reaching some peak, becomes non-increasing (see Fig.~\ref{fig:DegreeFunc}). Furthermore, two vertices that have minimum $d_M(\cdot)$ and are neighbors themselves, have the exact same neighbors in $M$.

\begin{figure}
\center
  \fbox{\includegraphics[scale=0.8]{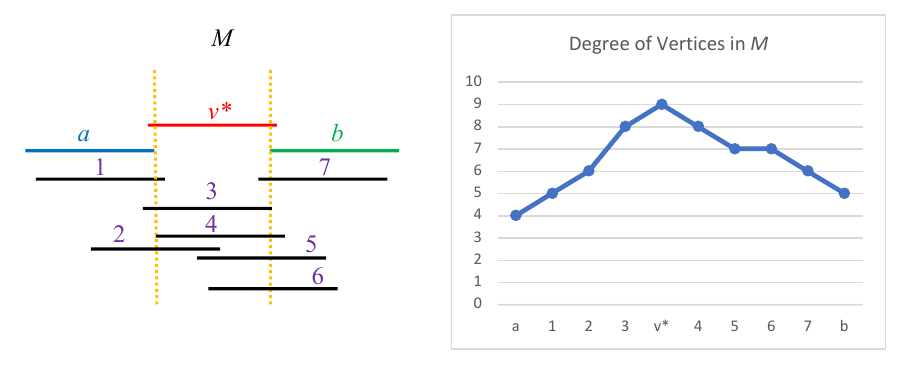}}
  \caption{Notice that $V(M)\subseteq N[a]\cup N[b]$, and for any $u\in N(a)=\{1,2,3,v^\star\}$, $N[a]\subseteq N[u]$. Additionally, the degree of each vertex, ordered from left to right, is shown on the right.}
  \label{fig:DegreeFunc}
\end{figure}

Overall, the above arguments justify our choice of $a$ and $b$ (i.e., that there exists at least one representation among those we consider that makes the same choice). For the sake of clarity, let us show this explicitly for the case where $V(L)\neq\emptyset$ (the arguments for the other two cases follows the same lines). As argued above, the leftmost vertex must have, among the vertices in $M$, maximum $d_L(\cdot)$ (which is greater than $1$) and, simultaneously, be smallest by $<$. Further, among these vertices with maximum $d_L(\cdot)$ (which, therefore, are neighbors of each other and have the same neighborhood in $L$) and that are smallest by $<$, it must have minimum $d_M(\cdot)$. All the vertices that can be chosen in this manner also have the same neighborhood also in $M$ (and no neighbors in $R$), and hence the choice among them is arbitrary. This justifies the choice of $a$. Now, in the easier case where $M\subseteq N[a]$, then $G[M]$ is in fact a clique. So, among all vertoces with largest $d_R(\cdot)$ (all of which have the same neighborhood in $R$) and which is largest by $<$, we can pick one as $b$ arbitrarily. In case $M\setminus N[a]\neq\emptyset$, then $G[M]$ is not a clique, and, in particular, the leftmost vertex is not a neighbor of $a$, yet it has minimum $d_M(\cdot)$ among all non-neighbors of $a$ (and all of the vertices that satisfy this condition have the same neighborhood in $M$). So, among the non-neighbors of $a$ vertices with maximum $d_R(\cdot)$ and that are largest by $<$, we can choose $b$ as a vertex with minimum $d_M(\cdot)$ (we remark that $b$ might not have minimum $d_M(\cdot)$ among all vertices in $M$, but it should have it among all the non-neighbors of $a$). Notice that is some choice is not possible, then we have a ``no-instance''.
\end{proof}

For the recursive calls in Step \ref{step:recursiveCalls}  on $L$ and $R$, we have the following claim.

\begin{claim}\label{claim:recursiveCalls}
The algorithm does not answer ``no-instance'' in Step~\ref{step:recursiveCalls}. Moreover, for {\em any} pair of representations $g_L$ of $L$ and $g_R$ of $R$ that comply with $<_L$ and $<_R$, respectively, there exists a representation of $G$ that satisfies the conditions in Claim \ref{claim:identifyAB} and that the restrictions of its ordering of starting and ending points to $L$ and $R$ is the same as $g_L$ and $g_R$, respectively.
\end{claim}

\begin{proof}
To prove the first part of the claim, let $f$ be a representation of $G$ that complies with $<$. Let $f_L$ and $f_R$ be the restrictions of $f$ to $L$ and $R$, respectively. Then, $f_L$ and $f_R$ are clearly representations of $L$ and $R$ that comply with the restrictions of $<$ to $V(L)$ and $V(R)$, respectively. Further, for any $u,v\in V(L)$ such that $d_L(u)<d_L(v)$, it must be that the starting point of $u$ is to the left of that of $v$---indeed, this follows from the observation that for any two vertices $x,y\in V(L)$ such that $f(x)$ starts before $f(y)$, every neighbor of $x$ in $M$ is also a neighbor of $y$ in $M$. Symmetrically, for any $u,v\in V(R)$ such that $d_L(u)>d_L(v)$, it must be that the starting point of $u$ is to the left of that of $v$. So, $f_L$ and $f_R$ comply with $<_L$ and $<_R$, respectively, which means that the recursive calls are made with yes-instances. Therefore, the algorithm does not answer ``no-instance'' in Step~\ref{step:recursiveCalls}.

To prove the second part of the claim, it suffices to show that there exists a representation $f$ of $G$ that satisfies the conditions in Claim \ref{claim:identifyAB}, and whose restriction to $L$ (resp.~$R$) is $g_L$ (resp.~$g_R$). We will only show a modification to ensure restriction to $L$ is $g_L$, since the modification for $R$ is symmetric (and can be done after the the property for $L$ is derived without altering it) Towards this, let $h$ be a representation of $G$ that satisfies the conditions in Claim \ref{claim:identifyAB}. Then, for some implicit $f$ that is meant to be a representation (which, until proved below, maybe incorrect), we define an ordering of the starting and ending points of the intervals of vertices as follows. The internal ordering between endpoints corresponding to vertices in $V(M)\cup V(R)$ is the same as in $f$, and the ordering between endpoints corresponding to vertices in $V(L)$ is the same as in $g_L$. All starting points corresponding to vertices in $V(L)$ appear before all those of vertices in $V(M)$, and all ending points corresponding to vertices in $V(L)$ appear before all those of vertices in $V(M)$ and before all starting points corresponding to  vertices in $V(R)$.
It remains to specify the relation between the ending points corresponding to $V(L)$ and the starting points corresponding to $V(M)$. To this end, let $v_1,v_2,\ldots,v_\ell$ where $\ell=|V(L)|$ be the ordering of the vertices in $L$ by the starting points (from left to right) of the intervals assigned to them by $g_L$. Let $u_1,u_2,\ldots,u_\ell$ be defined in the same way but with respect to $h$.
Additionally, let $w_1,w_2,\ldots,w_m$ where $m=|V(M)|$ be the ordering of the vertices in $M$ by the starting points (from left to right) of the intervals assigned to them by $h$. Then, for any $i\in\{1,2,\ldots,\ell\}$:
\begin{itemize}
\item If $d_M(v_i)=0$, then place the ending point of $f(v_i)$ before the starting points of all intervals (assigned by $f$) of vertices in $M$.
\item If $d_M(v_i)=|V(M)|$, then place the ending point of $f(v_i)$ after the starting points of all intervals (assigned by $f$) of vertices in $M$.
\item Else, place the ending point of $f(v_i)$ between the starting points of $f(w_{d_M(v_i)})$ and $f(w_{d_M(v_{i+1})})$.
\end{itemize}
We need to argue that no conflicts arose, that is, that we do not have two vertices $v_i,v_j\in V(L)$ such that the ordering of their ending points in $g_L$ is different than the ordering implied when we placed them between starting points corresponding to vertices in $M$. Suppose, by contradiction, that such a conflict arose. Then, there exist $v_i,v_j\in V(L)$ such that there exists $w_t\in V(M)$ so that we have ordered $\mathsf{end}_{f}(v_j)<\mathsf{start}_{f}(w_t)<\mathsf{end}_{f}(v_i)$, but $i<j$ (i.e., $\mathsf{end}_{g_L}(v_i)<\mathsf{end}_{g_L}(v_j)$). Having $\mathsf{end}_{f}(v_j)<\mathsf{start}_{f}(w_t)<\mathsf{end}_{f}(v_i)$ means that $d_M(v_j)<d_M(v_i)$, which, by the definition of $<_L$, implies that $v_j<_L v_i$. But then, it cannot be that $\mathsf{end}_{g_L}(v_i)<\mathsf{end}_{g_L}(v_j)$, because this contradicts the demand that $g_L$ complies with $<_L$. So, we have reached a contradiction, which means that no conflict arose.

\begin{figure}
\center
  \fbox{\includegraphics[scale=0.8]{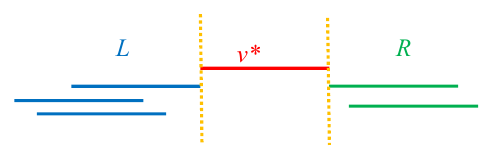}}
  \caption{Intervals of vertices in $L$ cannot intersect intervals of vertices in $R$. (Note that the intervals are open.)}
  \label{fig:NoIntersectionLRFIG}
\end{figure}

Clearly, if $f$ that corresponds to our ordering is a representation of $G$, then it should be clear that $f$ satisfies the conditions in Claim \ref{claim:identifyAB} and that the restriction of its ordering of starting and ending points to $L$ is the same as $g_L$ and $g_R$. So, it remains to prove that $f$ is a representation. Directly from the definition of the ordering for $f$, for every pair of intervals such that one belongs to a vertex in $L$ and the other to a vertex in $R$, their intervals do not intersect and indeed they are not adjacent (see Fig.~\ref{fig:NoIntersectionLRFIG}). Further, because $g_L$ and $h$ are representation, directly from the definition of the ordering for $f$, every two vertices that are both in $V(L)$ or both in $V(M)\cup V(R)$, their intervals (as assigned by $f$) intersect if and only if they are adjacent. It remains to consider the case where we have a vertex $v_i\in V(L)$ and a vertex $w_j\in V(M)$. Then, we need to show that $\mathsf{start}_f(w_j)<\mathsf{end}_f(v_i)$ if and only if $v_i$ and $w_j$ are adjacent in $G$.

To prove the above, let $t\in\{1,2,\ldots,\ell\}$ be such that $v_i=u_t$. We will claim that the placements of the ending point $u_t$ by $f$ and by $h$ among the starting points intervals of vertices in $M$ (as assigned by $f$ or $h$---the ordering of these starting points by $f$ and $h$ is the same) are the same. Given that this claim is correct, then because $h$ is a representation, we can conclude that $\mathsf{start}_f(w_j)<\mathsf{end}_f(v_i)$ if and only if $v_i$ and $w_j$ are adjacent in $G$. So, it remains to prove this claim. This follows from the observation that if $h(u_t)$ intersects $h(w_q)$ for some $w_q\in\{1,2,\ldots,m\}$, then $h(u_t)$ also intersect $h(w_p)$ for all $p\in\{1,2,\ldots,m\}$. Hence, because $u_t$ intersects exactly $d_M(u_t)$ vertices in $M$, the ending point of $h(u_t)$ must be placed within the starting points of the intervals of vertices in $M$ (as assigned by $h$)  exactly as we have placed the ending point of $f(v_i)$ within them. This completes the proof.
\end{proof}

Now, we consider the correctness of the ordering of starting points as defined in Step \ref{step:orderStart}.

\begin{claim}\label{claim:orderStart}
The algorithm does not answer ``no-instance'' in Step~\ref{step:orderStart}, and there exists a representation of $G$ that satisfies the conditions in Claim \ref{claim:recursiveCalls} and where the starting points of the intervals of all vertices are ordered in the way defined in Step \ref{step:orderStart}.
\end{claim}

\begin{proof}
Consider all represetations of $G$ that satisfy the conditions in Claim \ref{claim:identifyAB} (where we know that at least one such representation exists). So, by the condition in Claim~\ref{claim:proceed}, all intervals of vertices in $L$ begin before all intervals of vertices in $M$ begin, and all intervals of vertices in $M$ begin before all intervals of vertices in $R$ begin. With respect to the internal ordering of the starting points of the intervals of vertices in $V(L)$ and $V(R)$, correctness is guaranteed by the condition in Lemma \ref{claim:recursiveCalls}. Lastly, by following the same lines of  arguments as those given in the proof of  Claim~\ref{claim:identifyAB}, we conclude that among the representation we consider, at least one has the internal ordering of starting points of the intervals of vertices in $V(M)$ as the one defined in Step \ref{step:orderStart}.
\end{proof}

Next, we consider the correctness of the ordering of ending points for vertices in $L$ and $R$ as defined in Step \ref{step:orderEndLR}. With respect to vertices in $L$ with no neighbors in $M$, and vertices in $R$, correctness is immediate. With respect to vertices in $L$ with neighbors in $M$, clearly the endpoints should occur after the first starting point of intervals of vertices in $M$ and before all starting points of intervals of vertices in $R$ (because no vertex in $L$ can be a neighbor of a vertex in $R$; see Fig.~\ref{fig:NoIntersectionLRFIG}). Moreover, in any representation, if a vertex in $L$ has $d$ neighbors in $M$, then this $d$ must be the $d$ vertices in $M$ whose starting points are leftmost in that representation. So, this justifies our placement of all ending points of vertices in $L$ and $R$ within the starting points.

\begin{claim}\label{claim:orderEndLR}
There exists a representation of $G$ that satisfies the conditions in Claim \ref{claim:orderStart} and where the ending points of the intervals of all vertices in $V(L)\cup V(R)$ are ordered in the way defined in Step \ref{step:orderEndLR}.
\end{claim}

We proceed to consider the correctness of the ordering of ending points for vertices in $M$ as defined in Step \ref{step:orderEndM}. For vertices in $M$ that have neighbors in $R$ (i.e., $d_R(\cdot)>0$), the correctness follows the symmetric arguments as for the vertices in $L$ that have neighbors in $M$ in the previous claim. For vertices $u$ in $M$ with no neighbors in $R$, in case their intervals start after that of $v^\star$, then they must end after the ending point of the interval of $v^\star$ which is, in turn, after the starting points of all vertices in $M$. Because such $u$ have no neighbors in $R$, the ending points of their intervals must also be before the starting points of all intervals in $R$. Now, in case their intervals start before the interval of $v^\star$, then they must too intersect $a$, and so they intersect the first $d_M(u)$ vertices in $M$, which justifies the placement of the ends of their intervals.

\begin{claim}\label{claim:orderEndM}
There exists a representation of $G$ that satisfies the conditions in Claim \ref{claim:orderEndLR} and where the ending points of the intervals of all vertices in $V(M)$ are ordered in the way defined in Step \ref{step:orderEndM}.
\end{claim}

We proceed to consider the correctness of the ``internal ordering'' of ending points as defined in Step \ref{step:orderEndInternal}. This part is clear because if the starting point of the interval of one vertex is more to the left than another, so is the relation between their ending points (since all intervals are of the same unit length).

\begin{claim}\label{claim:orderEndInternal}
There exists a representation of $G$ that satisfies the conditions in Claim \ref{claim:orderEndM} and where the ending points of the intervals of all vertices are ordered in the way defined in Step \ref{step:orderEndInternal}.
\end{claim}

Overall, when we reach Claim \ref{claim:orderEndInternal}, we have, in fact, defined a strict ordering of the starting and ending points of the intervals of all vertices. As we argued before this sequence of claims, this completes the proof.
\end{proof}

We are now ready to conclude the proof of Lemma \ref{lem:annotatedReconstruction}, which, as argued earlier, implies Lemma \ref{lem:pivReconstruction}.

\begin{proof}[Proof of Lemma \ref{lem:annotatedReconstruction}]
Let $s\in\mathbb{N}$. We claim that for any $n\in\mathbb{N}_0$, {\sc Annotated Proper Interval Reconstruction} admits a $5\cdot s\cdot \log_{\frac{10}{9}} n$-pass algorithm on $n$-vertex graphs with success probability at least $1-(\frac{1}{5})^s\cdot n\log_{\frac{10}{9}} n$. This follows by induction on $n$, where at the basis $n$ is a constant (or just $0$) and hence the claim is trivial, and the step follows from Lemma \ref{lem:annotatedReconstructionStep}.

Now, we pick $s=\log_5(n^2)$. Then, the number of passes is $5\cdot \log_5(n^2)\cdot \log_{\frac{10}{9}} n=\OO(\log^2 n)$, and the success probability is  $1-(\frac{1}{5})^{\log_5(n^2)}\cdot n\log_{\frac{10}{9}} n = 1-\frac{1}{n^2}\cdot n\log_{\frac{10}{9}} n>\frac{9}{10}$ assuming that $n$ is larger than some fixed constant.
This completes the proof.
\end{proof}

\subsubsection{Post-processing Algorithm}\label{sec:ivdPost}

The purpose of this section is to prove the following lemma.
\begin{lemma}\label{lem:pivpost-processing}
{\sc Proper Interval Vertex Deletion} admits an $\widetilde{\OO}(n)$-space $2^{\OO(k)}\cdot n^{\OO(1)}$-time (static) algorithm.
\end{lemma}

\begin{figure}
\center
  \fbox{\includegraphics[scale=0.5]{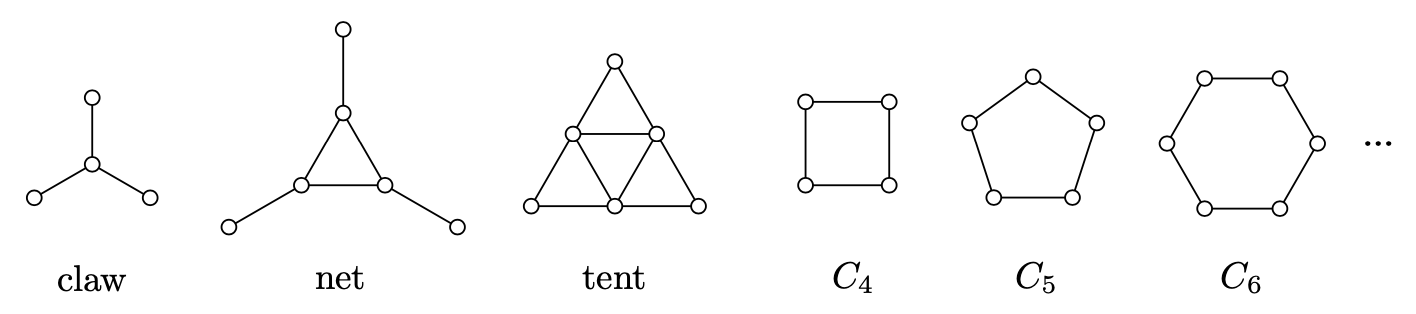}}
  \caption{The (infinite) family of forbidden induced subgraphs of proper interval graphs. The figure was taken from \cite{HofV13}.}
  \label{fig:ObstructionsProperInterval}
\end{figure}

Towards this proof, we need to recall the alternative characterization of proper interval graphs in terms of forbidden induced subgraphs. For this purpose, we remind that for any integer $i\geq 3$, $C_i$ in the induced (i.e.~chordless) cycle on $i$ vertices, and it is termed a {\em hole}. A {\em claw, net} and {\em tent} are specific graphs  on four, six and six, respectively, vertices, as illustrated in Fig.~\ref{fig:ObstructionsProperInterval}.

\begin{proposition}[\cite{wegner1967eigenschaften, brandstadt1999graph}]\label{prop:pivChar}
A graph $G$ is a proper interval graph if and only if is admits no claw, net, tent and hole.
\end{proposition}

We also need the definition of a proper circular-arc graph (see Fig.~\ref{fig:ProperCircularArc}).

\begin{definition}
A graph $G$ is a {\em proper circular-arc} if there exist a circle $C$ on the Euclidean  plane and a function $f$ that assigns to each vertex in $G$ an open interval of  $C$ such that no interval properly contains another, and every two vertices in $G$ are adjacent if and only if their intervals intersect. Such a function $f$ is called a {\em representation}, and given a vertex $v\in V(G)$, we let $\mathsf{begin}_f(v)$ and $\mathsf{end}_f(v)$ denote the beginning and end (in clockwise order) of the interval assigned to $v$ by $f$ (as if it was closed).
\end{definition}

We note the the class of proper circular-arc graphs is not equivalent to the class of unit circular-arc graphs.

\begin{figure}
\center
  \fbox{\includegraphics[scale=0.8]{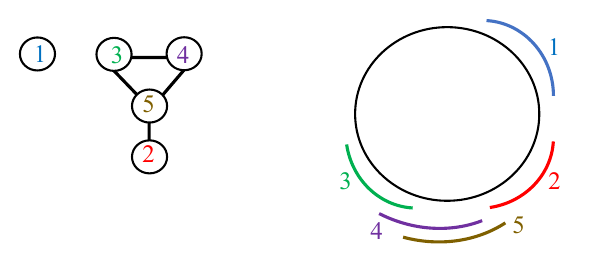}}
  \caption{A proper circular-arc graph.}
  \label{fig:ProperCircularArc}
\end{figure}

Notice that any $C_i$ (hole on $i$ vertices) is a proper circular-arc graph, and hence the class of proper arc-circular graphs is a strict superclass of the class of proper interval graphs. Our proof will build upon two results by van 't Hof and Villanger~\cite{HofV13}:

\begin{proposition}[\cite{HofV13}]\label{prop:reduceToCirc}
Let $G$ be a graph that that admits no claw, net, tent, $C_4$, $C_5$ and $C_6$ as induced subgraphs. Then, $G$ is a proper circular-arc graph.
\end{proposition}

\begin{proposition}[\cite{HofV13}]\label{prop:charSolOnCirc}
Let $G$ be a proper circular-arc graph with representation $f$. Then, any minimal subset $U\subseteq V(G)$ such that $G-U$ is a proper interval graph satisfies that there exists $v\in V(G)$ such that $U=\{u\in V(G): \mathsf{end}_f(v)\in f(u)\}$.\footnote{Note that $v\notin U$ (as $f(v)$ is an open interval).} Moreover, for every $v\in V(G)$, we have that $U=\{u\in V(G): \mathsf{end}_f(v)\in f(u)\}$ satisfies that $G-U$ is a proper interval graph.
\end{proposition}

So, Propositions \ref{prop:reduceToCirc} and \ref{prop:charSolOnCirc} present a straightforward way to prove Lemma \ref{lem:pivpost-processing} given that we know how to obtain a representation of a proper circular-arc graph in $\widetilde{\OO}(n)$ space and polynomial time. This can indeed be done---for example, thee linear (in the number of vertices and edges) algorithm for this purpose given in~\cite{DengHH96} can be seen to work in $\widetilde{\OO}(n)$ space. In what follows, we given an alternative approach, which shows that we can also solve the problem in $\widetilde{\OO}(n)$ space and polynomial time without the representation---while the proof is not short, it yields a very simple algorithm (see the proof of Lemma \ref{lem:pivpost-processing}	 ahead), which may be of independent interest.

We refine one direction of Proposition \ref{prop:charSolOnCirc} as follows.

\begin{lemma}\label{lem:charSolOnCircRef}
Let $G$ be a proper circular-arc graph with representation $f$ that has an independent set of size at least $4$. Then, any minimal subset $U\subseteq V(G)$ such that $G-U$ is a proper interval graph satisfies that there exists $v\in V(G)$ such that $U=\{u\in V(G): \mathsf{end}_f(v)\in f(u), N[v]\neq N[u]\}$.
\end{lemma}

\begin{figure}
\center
  \fbox{\includegraphics[scale=0.8]{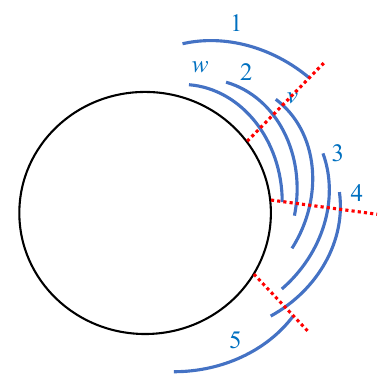}}
  \caption{A figure for the proof of Lemma \ref{lem:charSolOnCircRef}. Here, $U=\{2,v,3,4\}$, $N[w]=N[2]=N[v]=\{1,w,2,v,3,4\}$ and for all $u\notin\{w,2,v\}$, $N[u]\neq N[w]$. So, $v$ is indeed leftmost among vertices in $U$ with the same closed neighborhood as $v$. Notice that $T=\{3,4\}\subset U$.}
  \label{fig:ChangeCharacterization}
\end{figure}

\begin{proof}
Let $U$ be a minimal subset of $V(G)$ such that $G-U$ is a proper interval graph. Then, by Proposition \ref{prop:charSolOnCirc}, there exists $w\in V(G)$ such that $U=\{u\in V(G): \mathsf{end}_f(w)\in f(u)\}$. Now, let $v\in U$ be a vertex such that $N[v]=N[w]$ and when we go in clockwise order starting at $\mathsf{end}_f(v)$ (and neglecting vertices whose intervals coincide with that of $v$ apart from $v$ itself), the last end $\mathsf{end}_f(\cdot)$ that we encounter among those of all vertices in $U$ that have the same neighborhood as $w$ is that of $v$. In particular, as the maximum independent set size is at least $4$, this means that there does not exist a vertex $u$ with the same closed neighborhood as $v$ (and in particular, which is a neighbor of $w$) and yet ``ends after $v$'', that is, $\mathsf{end}_f(v)\in f(u)$.  So, $T=\{u\in V(G): \mathsf{end}_f(v)\in f(u), N[v]\neq N[u]\}=\{u\in V(G): \mathsf{end}_f(v)\in f(u)\}\subseteq U$ (see Fig.~\ref{fig:ChangeCharacterization}). If $w=v$, we are done. Else, we have that $T\subset U$ because $v\in U$ but $v\notin T$.  By Proposition \ref{prop:charSolOnCirc}, $G-T$ is a proper interval graph. However, this contradicts the minimality of $U$.
\end{proof}

Recall that the crux of the proof should be in the following lemma.

\begin{lemma}\label{lem:onCircArc}
Consider {\sc Proper Interval Vertex Deletion} on proper circular-arc graphs that admit no claw, net, tent, $C_4$, $C_5$, $C_6$ and $C_7$ as induced subgraphs. Then, it admits an $\widetilde{\OO}(n)$-space $2^{\OO(k)}\cdot n^{\OO(1)}$-time (static) algorithm.
\end{lemma}

Towards the proof of this lemma, we have the following lemma. In particular, this lemma gives a characterization of minimal solutions that does not depend on $f$.

\begin{lemma}\label{lem:alternativeCircArc}
Let $G$ be a proper circular-arc graph with representation $f$ and that has an independent set of size at least $4$. Let $v\in V(G)$, and denote $U=\{u\in V(G): \mathsf{end}_f(v)\in f(u), N[v]\neq N[u]\}$. Then, there exists $x\in V(G)$ such that $U=W_{v,x}$ where $W_{v,x}=\{u\in N[v]\cap N[x]: N[u]\neq N[v], N[u]\subseteq N[v]\cup N[x]\}$.
\end{lemma}

\begin{proof}
Let $x$ be a vertex such that $x\in U$ (and so $\mathsf{end}_f(v)\in f(x)$) and among all vertices in $U$ the end of its interval $\mathsf{end}_f(x)$ is furthest from $\mathsf{end}_f(v)$ (when distance is measure in clockwise order); see Fig.~\ref{fig:Wvx}. In what follows, we show that $U=W_{v,x}$.

\begin{figure}
\center
  \fbox{\includegraphics[scale=0.8]{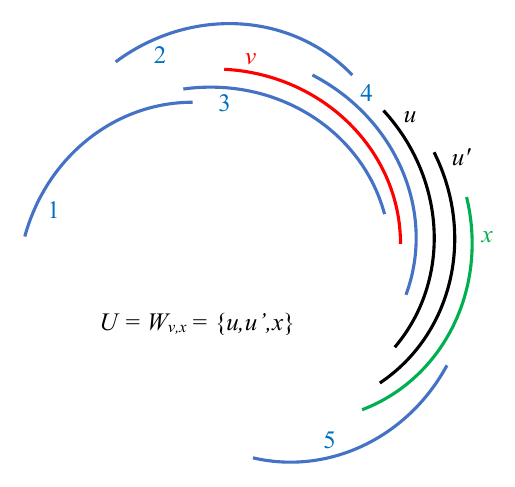}}
  \caption{A figure for the proof of Lemma \ref{lem:alternativeCircArc}. In particular, $1,2,3,5\notin U$ since $\mathsf{end}_f(v)$ does not belong to their intervals, and $4\notin U$ since $N[v]=N[4]$. Additionally, $1,2\notin W_{v,x}$ since they do not belong to $N[x]$, $3\notin W_{v,x}$ since $N[3]$ is not contained in $N[v]\cup N[x]$, $4\notin W_{v,x}$ since $N[v]=N[4]$, and $5\notin W_{v,x}$ since $4\notin N[v]$.}
  \label{fig:Wvx}
\end{figure}

In one direction, consider some $u\in U$. Then, because $\mathsf{end}_f(v)\in f(u)$, $\mathsf{end}_f(v)\in f(x)$ and $G$ has an independent set of size at least $4$, and by the choice of $w$, this means that when we go from $\mathsf{start}_f(v)$ in clockwise order, we observe the ordering $\mathsf{start}_f(v)<\mathsf{start}_f(u)\leq \mathsf{start}_f(x)< \mathsf{end}_f(v)<\mathsf{end}_f(u)\leq \mathsf{end}_f(x)$; see Fig.~\ref{fig:Wvx}. So, $f(u)$ has non-empty intersection with both $f(v)$ and $f(x)$, and it is contained in their union, which means that $u\in N[v]\cap N[x]$ and $N[u]\subseteq N[v]\cup N[x]$. Because $u\in U$, it also holds that $N[u]\neq N[v]$, and hence $u\in W_{v,x}$.

In the other direction, consider some $u\in W_{v,x}$. Because $u\in N[v]\cap N[x]$ and in particular $u\in N(v)$, we have that  $\mathsf{end}_f(v)\in f(u), \mathsf{start}_f(v)\in f(u)$ or $f(v)=f(u)$.
Because $N[u]\neq N[v]$ while $N[u]\subseteq N[v]\cup N[x]$, there exists a vertex that is a neighbor of $u$ and $x$ but not of $v$. However, by the choice of $x$ and because $G$ has an independent set of size at least $4$, in case $\mathsf{start}_f(v)\in f(u)$ or $f(v)=f(u)$, this is not possible. For intuition, see Fig.~\ref{fig:Wvx} (where it can be seen that any vertex in $N[u]\setminus N[x]$, which in this case is $5$, must be ``more to the left'' than $x$, which imply that $u$ is ``more to the left'' than $v$) and Fig.~\ref{fig:SmallPerimeter}. So, necessarily, $\mathsf{end}_f(v)\in f(u)$.
\end{proof}

\begin{figure}
\center
  \fbox{\includegraphics[scale=0.8]{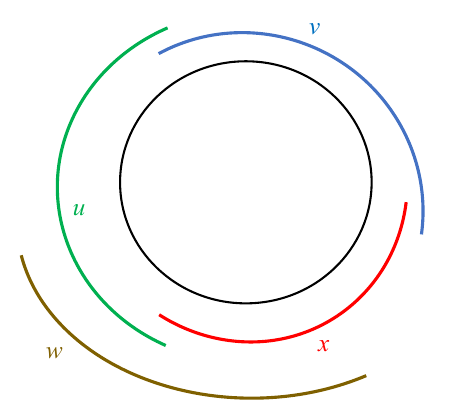}}
  \caption{A figure for the proof of Lemma \ref{lem:alternativeCircArc}. Here, there cannot be an independent set of size at least $4$ (even if additional intervals were added, preserving the property that not interval properly contains another), and we have $u\in W_{v,x}\setminus U$. Intuitively, we do have that $u$ must be ``more to the left'' than $v$, but it can do so, still being a neighbor of both $v$ and $x$, while going ``far enough'' to the left so that it can intersect the start, rather than the end, of $v$; yet, this prevents the graph from having a large independent set. }
  \label{fig:SmallPerimeter}
\end{figure}

We complement Lemma \ref{lem:alternativeCircArc} with the following simple observation.

\begin{observation}\label{obs:smallPerimeter}
Let $G$ be a proper circular-arc graph with representation $f$ such that the maximum independent set size of $G$ is at most $3$, and which admits no claw, net, tent, $C_4$, $C_5$, $C_6$ and $C_7$ as induced subgraphs. Then, $G$ is a proper interval graph.
\end{observation}

\begin{proof}
Targeting a contradiction, suppose that $G$ is not a proper interval graph. Then, as $G$ admits no claw, net, tent, $C_4$, $C_5$, $C_6$  and $C_7$ as induced subgraphs, by Proposition \ref{prop:pivChar}, it must admit some hole $C_i$ for $i\geq 8$ as an induced subgraph. As $C_i$ for any $i\geq 8$ has an independent set of size at least $4$, so does $G$, and therefore we have reached a contradiction.
\end{proof}

Now, it becomes easy to prove Lemma \ref{lem:onCircArc}:

\begin{proof}[Proof of Lemma \ref{lem:onCircArc}]
The algorithm works as follows. It first checks in $\widetilde{\OO}(n)$ space and polynomial time whether $G$ is already  proper interval graph by making use of the algorithm in Lemma~\ref{lem:pivReconstruction}, and if the answer is positive, then it returns that we have a yes-instance (with solution $\emptyset$). Notice that although the algorithm in Lemma \ref{lem:pivReconstruction} works for streams, we can simulate it by iterating over the edges in $G$.

Afterwards, for every pair of vertices $v,x\in V(G)$, it computes $W_{v,x}=\{u\in N(v)\cap N(x): N[u]\neq N[v], N(u)\subseteq N(v)\cup N(x)\}$ and if $|W_{v,x}|\leq k$, then:
\begin{itemize}
\item It checks in $\widetilde{\OO}(n)$ space and polynomial time whether $G-W_{v,x}$ is a proper interval graph by making use of the algorithm in Lemma \ref{lem:pivReconstruction}.
\item  If the answer is positive, then it returns that we have a yes-instance (with solution $W_{v,x}$).
\end{itemize}
At the end, having not returned a positive answer so far, the algorithm returns ``no-instance''. This completes the description of the algorithm.

Clearly, our algorithm uses $\widetilde{\OO}(n)$ space and polynomial time. We now argue that it is correct. In one direction, suppose that $(G,k)$ is a yes-instance, and let $U$ be a solution for it. If $G$ is already a proper interval graph, then by the correctness of the algorithm in Lemma \ref{lem:pivReconstruction}, our algorithm returns ``yes-instance''.
Else, by Observation \ref{obs:smallPerimeter}, $G$ has an independent set of size at least $4$. Then, Lemmas \ref{lem:charSolOnCircRef} and \ref{lem:alternativeCircArc} imply that there exists a pair of vertices $v,x\in V(G)$ such that $U=W_{v,x}$. By the correctness of the algorithm in Lemma \ref{lem:pivReconstruction}, in the iteration corresponding to this pair our algorithm should return ``yes-instance''.

In the other direction, suppose that our algorithm returns ``yes-instance''. By the correctness of the algorithm in Lemma \ref{lem:pivReconstruction}, either $G$ is already a proper interval graph, or our algorithm has found a subset of vertices $W$ of size at most $k$ such that $G-W$ is a proper interval graph. So, $(G,k)$ is a yes-instance.
\end{proof}

Finally, we are ready to prove Lemma \ref{lem:pivpost-processing}.

\begin{proof}[Proof of Lemma \ref{lem:pivpost-processing}]
The algorithm works as follows. We first consider every subset $X\subseteq V(G)$ of size at most $6$, and for each such subset $X$, we check whether it induces a claw, net, tent, $C_4$, $C_5$, $C_6$ or $C_7$; in case the answer is positive, we branch on which vertex to delete from $X$ (that is, the algorithm calls itself recursively on $(G-v,k-1)$ for each $v\in X$). After this, we call the algorithm in Lemma \ref{lem:onCircArc} and return the same answer.

By Proposition \ref{prop:pivChar} and because we select which vertices to delete to hit all claws, nets, tents, $C_4$'s, $C_5$'s, $C_6$'s and $C_7$'s in an exhaustive manner, there exists a solution to $(G,k)$ if and only if there exists a solution to one of the instances with whom we call the algorithm in Lemma \ref{lem:onCircArc}. Moreover, by Proposition \ref{prop:reduceToCirc}, each call is made with a proper circular-arc graph. So, by the correctness of the algorithm in Lemma \ref{lem:onCircArc}, our algorithm is correct.

For the first part of our algorithm, notice that going over every subset $X\subseteq V(G)$ of size at most $7$ and doing brute-force to check isomorphism are done in $\widetilde{\OO}(n)$ space. Moreover, the total time required by this part if $7^k\cdot n^{\OO(1)}$ (since we branch in at most $7$ ways, and in each the parameter $k$ is decreased). So, by the space and time complexities of the algorithm  in Lemma \ref{lem:onCircArc}, we conclude that the space and time complexities of our algorithm are as stated.
\end{proof}

Based on Theorem \ref{lem:reductionLemma}, \ref{lem:pivReconstruction} and \ref{lem:pivpost-processing}, we conclude the correctness of Theorem \ref{thm:piv}.


\subsection{Block Vertex Deletion in One Pass}

Recall (see Section~\ref{sec:graphClassPrelims}) that a graph $G$ is a block graph if every biconnected component of it is a clique and equivalently, if it is $\{D_4, C_{\ell + 4,\, \ell \geq 0} \}$-free.


The purpose of this section is to prove the following theorem.

\begin{theorem}\label{thm:blockGraphs}
{\sc Block Vertex Deletion} admits a $1$-pass semi-streaming algorithm with $2^{\OO(k)} \cdot n^{\OO(1)}$ post-processing time.
\end{theorem}

\subsubsection{Reconstruction Algorithm}
In this section we present a somewhat more general result on the reconstruction of block graphs. We begin by defining $t$-flow graphs and $t$-block graphs that generalize block graphs. Next, we present a 1-pass randomized streaming algorithm for reconstruction of $t$-flow graphs. Then, combined with a simple recognition algorithm, based on \emph{Perfect Elimination Ordering}, we obtain a 1-pass randomized streaming algorithm for reconstruction of $t$-block graphs. As a corollary we obtain a one-pass randomized reconstruction algorithm for block graphs.

\begin{definition}
    For an integer $t \geq 1$, a graph $G$ is called \emph{$t$-flow graph} if and only if for each pair of vertices $u,v \in V(G)$ either $\{u,v\} \in E(G)$,
    or there is a subset $S$ of at most $t$ vertices such that $u$ and $v$ are disconnected in $G - S$.
\end{definition}

\begin{definition}
    For an integer $t \geq 1$, a graph $G$ is called a \emph{$t$-block graph} if it is a chordal graph where any two maximal cliques intersect in at most $t$ vertices.
\end{definition}

Note that $1$-block graphs are the same as \emph{block graphs}.
Recall that a \emph{clique-tree decomposition} of a graph $G$ is a tree-decomposition $(T,\beta)$ such that every bag $\beta(v)$ for $v \in V(T)$ induces a clique in $G$~\cite{brandstadt1999graph}. Then recall that a graph $G$ is chordal if and only if it admits a clique-tree decomposition~\cite{brandstadt1999graph}.  Then following observation is immediate.
\begin{observation}\label{obs:chordal-t-block}
    A graph $G$ is a $t$-block graph if and only if it is a chordal $t$-flow graph.
\end{observation}

Let us begin with  a reconstruction algorithm for $t$-flow graph.

\begin{lemma}\label{lem:recon-t-flow}
    Let $t \geq 1$ be an integer. The class of $t$-flow graphs has a 1-pass randomized reconstruction algorithm, with $n^{t+\OO(1)}$ post-processing time and $\tilde\OO(n)$ space.
\end{lemma}
\begin{proof}
    Let ${\cal F} = \{V_1, V_2, \ldots, V_r\}$ be a $(n,t,2)$-separating family over $V(G)$. Note that such a family contains $r = (t+2)^{10} \log (t+2) \log n$ vertex subsets,  and it can be constructed in $\OO(t^{\OO(1)} n \log n)$ time and space (Corollary~\ref{cor:graph-separating-family}).
    Recall that, for any pair of vertices $\{u,v\}$ and any $Q \subseteq V(G)$ such that $|Q| \leq t$, there is an $V_i \in {\cal F}$ such that $u,v \in V_i$ and $V_i \cap S = \emptyset$. For $V_i \in \cal F$, let $G_i = G[V_i]$.
    We now describe our algorithm.

    \begin{itemize}

    \item \textbf{Streaming Phase.} With the family $\cal F$ in hand, our algorithm computes the following data structures in 1-pass: ~$(i)$~the degree $d_G(v)$ of every vertex $v \in V(G)$, and~$(ii)$~for each $G_i$, a (randomized) connectivity sketch $D_i$ (Proposition~\ref{prop:dynamicConnectivitySketch}). Note that $D_i$ is a spanning-forest of $G_i$ and we can enumerate $E(D_i)$. Further, we require $\tilde\OO(r \cdot n)$ space to store all these data-structures. Let $G'$ be the graph with vertex set $V(G)$ and edge-set $\cup_{1 \leq i \leq r} E(D_i)$, and note that we can explicitly construct $G'$, which requires $\tilde\OO(r \cdot n)$ space.

    \item \textbf{Reconstruction Algorithm.} Next, let us describe the post-processing algorithm to reconstruct $G$. Recall that our goal is to devise an $n^{t+\OO(1)}$-time and $\tilde\OO(n)$-space algorithm that, given a pair of vertices $u,v \in V(G)$ as input, outputs that either $\{u,v\} \in E(G)$, or $\{u,v\} \notin E(G)$, or else $G$ is not a $t$-flow graph. This algorithm is as follows. We first define a graph $\tilde{G}$ on the vertex $V(G)$, and compute the degree $d_{\tilde{G}}(v)$ of every vertex $v \in V(G)$.
    Note that this graph $\tilde{G}$ is not explicitly constructed.
    Initially, let $\tilde{G} = G'$.  Then we consider every distinct pair of non-adjacent vertices $x,y$ in $\tilde{G}$. For each vertex pair $x,y$, we determine if $\{x,y\} \in E(G)$ or not as follows. We consider every subset $S$ of at most $t$ vertices such that $x,y \notin S$. For each choice of $S$, we first determine $V_i \in \cal F$ such that $x,y \in V_i$ and $S \cap V_i = \emptyset$. Then we check if $x$ and $y$ are connected in $G_i = G[V_i]$ by using its connectivity sketch $D_i$.
    If $x$ and $y$ are connected in $G_i$ for every choice of $S$, then we conclude that edge $\{x,y\}$ is in $\tilde{G}$ and increment $d_{\tilde{G}}(x)$ and $d_{\tilde{G}}(y)$ by 1 each. Furthermore, if the pair $x,y$ is the input pair $u,v$, then we also remember if this particular pair forms an edge or a non-edge.
    Finally, when all vertex pairs have been processed, we check if $d_G(v)$ is equal to $d_{\tilde{G}}(v)$ for all $v \in V(G)$. If not, then we output that $G$ is not a $t$-flow graph. Otherwise, we output that $G$ is not a $t$-flow graph, and the succinct representation queries, given $u,v$, whether or not $\{u,v\} \in E(\tilde{G})$.
    Observe that, it is clear that the above algorithm requires $n^{t + \OO(1)}$ time and $\tilde\OO(n)$ space.
    \end{itemize}

    We now argue that above reconstruction algorithm is correct. Towards this, we first show that $E(G) \subseteq E(\tilde{G})$, irrespective of whether or not $G$ is a $t$-flow graph. Consider any edge $\{x,y\} \in E(G)$, and consider any subset of vertices $S \subseteq V(G)$ such that $|S| \leq t$ and $x,y \notin S$. Then, the separating-family $\cal F$ contains $V_i$ such that $x,y \in V_i$ and $S \cap V_i = \emptyset$. Since $\{x,y\} \in E(G_i)$ and $D_i$ is a spanning forest of $G_i = G[V_i]$ (from the connectivity sketch by Proposition~\ref{prop:dynamicConnectivitySketch}), we have that $x,y$ must in the same connected component of $D_i$. Hence, $x,y$ are connected in $G_i$. Since this holds for every choice of $S$, it follows that $\{x,y\} \in E(\tilde{G})$. By applying this argument to every edge in $E(G)$, we obtain that $E(G) \subseteq E(\tilde{G})$.

    Next, let us argue that $E(\tilde{G}) \subseteq E(G)$, assuming that $G$ is a $t$-flow graph. Towards this, consider an edge $\{x,y\} \in E(\tilde{G})$. As $G'$ is a subgraph of both $G$ and  $\tilde{G}$, if $\{x,y\} \in E(G')$ then $\{x,y\} \in E(G)$. Otherwise, $\{x,y\} \notin E(G')$. Then, for every $S \subseteq V(G)$ such that $|S| \leq t$ and $x,y \notin S$, consider $V_i \in \cal F$ such that $x,y \in V_i$ and $S \cap V_i =\emptyset$. Observe that, as $\{x,y\} \in E(\tilde{G})$, there is a path in $G_i = G[V_i]$ from $x$ to $y$. This path is also present in the graph $G' - S$. Since this holds for every choice of $S$, the minimum-separator of $x$ and $y$ in $G'$ has $t+1$ or more vertices. Therefore, assuming that $G$ is a $t$-flow graph, $\{x,y\} \in E(G)$. Finally, this argument holds for every edge in $\tilde{G}$, and hence $E(\tilde{G}) \subseteq E(G)$. Combined with the above we obtain that $E(G) = E(\tilde{G})$, when ever $G$ is a $t$-flow graph.

    Finally, let us argue that our reconstruction algorithm is correct. 
    In one direction let $G$ be a $t$-flow graph. Then, as $E(G) = E(\tilde{G})$, we have $d_G(v) = d_{\tilde{G}}(v)$ for every vertex $v \in V(G)$. Hence the reconstruction algorithm correctly concludes that $G$ is a $t$-flow graph (and the succinct representation is correct).
    In the reverse direction, suppose that the algorithm concludes that $G$ is a $t$-flow graph. This means that $d_G(v) = d_{\tilde{G}}(v)$ for every vertex $v \in V(G)$. We give a proof by contradiction, and let us assume that $G$ is not a $t$-flow graph. This means there exists $u,v \in V(G)$ such that $\{u,v\} \notin E(G)$ and for every $S \subseteq V(G)$, $|S| \leq t$, $u$ and $v$ are connected in $G - S$. Then $\{u,v\} \in E(\tilde{G})$, and as $E(G) \subseteq E(\tilde{G})$, it must be the case that $d_{\tilde{G}}(v) \geq d_{G}(v) + 1$. But this is a contradiction since $d_G(v) = d_{\tilde{G}}(v)$. Therefore, $G$ must be a $t$-flow graph.
    This concludes the proof of this lemma.
\end{proof}

We are now ready to present a reconstruction algorithm for $t$-block graphs.
Towards this we require a few results about chordal graphs.
\begin{definition}\label{defn:peo}
    A \emph{perfect elimination ordering} of a graph $G$ is an ordering $v_1, v_2, \ldots v_n$ of the vertex set $V(G)$ such that for any $i \in \{1,2, \ldots n\}$, $N_{G_i}(v_i)$ is a clique in the graph $G_i = G - \{v_1, v_2, \ldots v_{i-1}\}$.
\end{definition}

The following is a well-known characterization of chordal graph.
Recall that a vertex $v$ in a graph $G$ is called \emph{simplician} if $N(v)$ induces a clique in $G$.
\begin{proposition}[\cite{brandstadt1999graph}]\label{prop:chordal-peo}
A graph $G$ is a chordal graph if and only if it has a perfect elimination ordering. Furthermore, given a chordal graph $G$ such an ordering can be obtained (greedily) by repeatedly picking a simplicial vertex in the current graph and deleting it.
\end{proposition}

\begin{lemma}\label{lemma:BVD:reconst}
    Let $t \geq 1$ be an integer. The class of $t$-block graphs has a 1-pass randomized reconstruction algorithm, with $n^{t+\OO(1)}$ post-processing time and $\widetilde\OO(n)$ space.
\end{lemma}
\begin{proof}
    Let $G$ denote the input graph.
    We begin by applying the reconstruction algorithm for $t$-flow graph from  Lemma~\ref{lem:recon-t-flow} to $G$. If it concludes that $G$ is not a $t$-flow graph, then we output that $G$ is not a $t$-block graph. Otherwise, $G$ is a $t$-flow graph we obtain a succinct representation of $G$.
    Now ,our goal is to test if $G$ is also a chordal graph.
    Here we apply the algorithm of Proposition~\ref{prop:chordal-peo}, which requires that we repeatedly pick a simplicial vertex  in the current graph and removing it. Observe that it is enough to describe a $\widetilde\OO(n)$ algorithm to find a simplicial vertex in the current remaining graph. This algorithm is as follows. We examine the vertices one by one, and for each vertex $v$ we first enumerate its neighborhood $N(v)$, and then test if every pair of vertices in it forms an edge or not. If so, we conclude that $v$ is a simplicial vertex. Otherwise, $v$ is not a simplicial vertex. It is clear that this runs in $n^{\OO(1)}$ time and $\tilde\OO(n)$ space.

    Our algorithm tests if the $t$-flow graph is chordal by attempting to compute a perfect elimination ordering of it. If it succeeds, then we output that $G$ is a $t$-block graph, and also output its succinct representation from Lemma~\ref{lem:recon-t-flow}. Otherwise, we output that $G$ is not a $t$-block graph. To argue the correctness of this algorithm, observe that if the algorithm does succeed in computing a perfect elimination ordering of $G$, then it is a chordal by Proposition~\ref{prop:chordal-peo}. Hence, $G$ is a $t$-block graph by Observation~\ref{obs:chordal-t-block}. In the reverse direction, if $G$ is a $t$-block graph, then it is a chordal $t$-flow graph. Hence, by Proposition~\ref{prop:chordal-peo}, our algorithm produces a perfect elimination ordering of $G$. Hence it outputs that $G$ is a $t$-block graph. This concludes the proof of this lemma.
\end{proof}

Let us remark that the above lemma, which describes a linear-space algorithm to recognize whether a given $t$-flow graph is also chordal, crucially depends on the fact that the $t$-flow graphs have a linear-space reconstruction algorithm.
In contrast, for general graphs we obtain a lower-bound $\Omega(n^2)$ for recognizing chordal graphs (see Theorem~\ref{thm:chordalLowerBound}).

\subsubsection{Post-processing Algorithm}

In this section we present the post-processing algorithm for {\sc Block Vertex Deletion} (BVD). This algorithm takes as input a graph $G$ and an integer $k$, and outputs a vertex subset $S \subseteq V(G)$ such that $G - S$ is a block graph, if one exists. This algorithm uses $\widetilde\OO(n)$ space, and in particular we assume that the input graph $G$ is presented in read-only memory.
Our algorithm follows the approach of ~\cite{AgrawalKLS16} (to obtain the fastest possible algorithm) with suitable modifications to use only $\tilde\OO(n)$ space.

Let us first consider the case when the input graph $G$ is $\{C_4, D_4\}$-free. Then, we have the following proposition from~\cite{AgrawalKLS16}.

\begin{proposition}[Lemma 6~\cite{AgrawalKLS16}]
\label{prop:BVD:cliqueBound}
    Let $G$ be a graph that is $\{C_4, D_4\}$-free.
    Then $G$ has at most $n^2$ maximal cliques, and any two maximal cliques intersect in at most one vertex.
\end{proposition}

Given a graph $G$ that is $C_4$ and $D_4$ free, consider the \emph{auxiliary graph} $H$ constructed from $G$ as follows. $H$ is a bipartite graph with vertex set $V(H) = V(G) \cup V_{\cal C}$ where $\cal C$ is the collection of all maximal cliques in $G$ and $V_{\cal C}$ is a collection of vertices, with one vertex $v_C$ for each element $C$ of $\cal C$. The edge set $E(H)$ consists of all pairs $(v, v_C)$, where $v \in V(G)$, $v_C \in V_{\cal C}$ and $C \in \cal C$, such that the vertex $v$ belongs to the clique $C$.

\begin{proposition}[Lemma 7 ~\cite{AgrawalKLS16}]
\label{prop:BVD:FVS}
    Let $G$ be a $\{C_4, D_4\}$-free graph, and let $H$ be the auxiliary graph of $G$. Then, $S \subseteq V(G)$ is a block vertex deletion set of $G$ if and only if $S$ is a feedback vertex set of $H$.
\end{proposition}

Agrawal, Kolay, Lokshtanov and Saurabh~\cite{AgrawalKLS16} explicitly constructed $H$ and used an algorithm for {\sc Feedback Vertex Set} on $H$ (while using weights to ensure that only vertices from $V(G)$ are picked). We cannot do this, since just storing the vertex set of $H$ takes $\OO(n^2)$ space. So, we have a more difficult task of resolving the {\sc Feedback Vertex Set} problem on $H$ without even constructing its vertex set explicitly.

For our algorithm, we require the following subroutine that counts the number of maximal cliques in $G$ that contain a specified vertex $v$.

\begin{lemma}
\label{lemma:cliqueCountAlgo}
    Let $G$ be a graph and let $v$ be any vertex. Then there is an algorithm that counts the number of maximal cliques in $G$ that contain $v$ using $\tilde\OO(n \cdot \log \alpha)$ space in time $\alpha \cdot n^{\OO(1)}$, where $\alpha$ denotes the number of maximal cliques in $G$ that contain $v$.
\end{lemma}
\begin{proof}
    We present the algorithm in the following pseudo-code. The input is a graph $G$ and a vertex $v$. We assume that the vertices of $G$ are numbered from $1$ to $n$ in some arbitrary order.

\begin{algorithm}
\caption{Counting the number of Maximal Cliques in $G$ containing vertex $v$}
\SetAlgoLined
    \KwIn{Graph $G$ on $n$ vertices, vertex $v$}
    \KwOut{Number of maximal cliques in $G$ containing vertex $v$}
    \tcp*[h]{GLOBAL VARIABLES}\\
    $\sigma \gets$ an arbitrary ordering of the vertices in $V(G)$ where $\sigma(v) = 1$\\
    \tcp*[h]{bitvector of the current partial maximal clique}\\
    $X \gets $ bit-vector of length $n$ initialized to all $0$\\
    $X[1] \gets 1$ \tcp*[h]{Add $v$ to $X$}\\
    \SetKwFunction{FMain}{Count-X-Cliques}
    \SetKwProg{Fn}{Function}{:}{}
    \Fn{\FMain{}}{
        $Stop \gets \sf{False}$\\
        $uPrev \gets 0$\\
        $numCliques \gets 0$\\
        \While{not Stop}{
            $u \gets$ the common neighbor of $X$ in $V(G) \setminus X$ minimizing $\sigma(u)$ such that $\sigma(u) > \sigma(uPrev)$ and $\sigma(u) > \sigma(w) ~\forall w \in X$\\
            \uIf{$u = \sf{NULL}$}{
                \If{$numCliques = 0$} {
                    \tcp*[h]{then the input $X$ cannot be extended to a larger clique}\\
                    $numCliques \gets 1$\\
                }
                $Stop \gets \sf{True}$\\
            }
            \Else{
                $uPrev \gets u$\\
                $X[\sigma(u)] \gets 1$ \tcp*[h]{Extend $X$ by the vertex $u$}\\
                \tcp*[h]{Recursively count maximal cliques containing $X$ extended by $u$}\\
                $uCliques \gets$ \FMain{}\\
                $numCliques \gets numCliques + uCliques$\\

                $X[\sigma(u)] \gets 0$ \tcp*[h]{Restore $X$}\\
            }

        }
        \Return{$numCliques$}
    }
\end{algorithm}

The algorithm is implemented a simple branching algorithm.
We start by defining an arbitrary ordering $\sigma$ of the vertices in $V(G)$ such that $v$ is the first vertex; note that this requires $\OO(n)$ space and $n^{\OO(1)}$ time. We then define a global length-$n$ bit-vector $X$ that contains the vertices of the current (partial) maximal clique. The algorithm is implemented via the function \texttt{Count-X-Cliques()} that makes recursive calls at each branching. When this function is called, it iterates over all the common neighbors $u$ of $X$ in $V(G) \setminus X$ such that $\sigma(u) > \sigma(w) ~ \forall w \in X$, one by one from smallest to largest. For each $u$, it counts the number of maximal cliques containing $X \cup \{u\}$. It adds up all these numbers and returns the sum as the number of maximal cliques in $G$ containing $X$. Observe that, to determine the next common neighbor of $X$, it only needs to remember the previous common neighbor and check all larger vertices in $V(G) \setminus X$ one by one. Note that this can be done in $n^{\OO(1)}$ time and $\OO(\log n)$ space.  If no common neighbor of $X$ exists, then $X$ itself is a maximal clique and we return $1$. It is clear that every maximal clique in $G$ that contains $v$ is generated and counted once, in the lexicographic order induced by $\sigma$ on subsets of $V(G)$. Indeed, for any two common neighbors $u, u'$ of $X$ such that $\sigma(u) < \sigma(u')$, the maximal cliques that contains both vertices are only generated in the branch for $X \cup u$ (and not in the branch for $X \cup u'$). Next, we show the claimed bounds on space and time.

Observe that each call to the function \texttt{Count-X-Cliques()} requires only a constant number of local-variables, that each require at most $\log \alpha$ bits to store. Further, the depth of recursion is upper-bounded by $n$ since each recursive call sets a $0$-bit of $X$ to $1$. Hence the total space used by the algorithm is $\tilde\OO(n \cdot \log \alpha)$.
Finally, observe that after any $n$ successive calls made to \texttt{Count-X-Cliques()}, we must encounter a maximal clique of $G$ that contains $v$, and each call internally requires $n^{\OO(1)}$ time.
Hence the total time required by this algorithm is $\alpha \cdot n^{\OO(1)}$.
\end{proof}

As a corollary of Proposition~\ref{prop:BVD:cliqueBound} and Lemma~\ref{lemma:cliqueCountAlgo}, we obtain the following.
\begin{lemma}
\label{lemma:BVD:cliqueCount}
    Let $G$ be a graph that is $\{C_4, D_4\}$-free. Then in $n^{\OO(1)}$ time and $\tilde\OO(n)$ space, we can count and store the number of maximal cliques in $G$ that contain $v$ for every vertex $v \in V(G)$.
\end{lemma}

Next, we present a few lemmas and reduction rules that will be useful for our algorithm.

\begin{lemma}\label{lem:bvd:soln-prop}
    Let $G$ be a graph that is $\{C_4, D_4\}$-free. Then the following statements hold.
    \begin{itemize}
        \item Let $v \in V(G)$ be a vertex that lies in only one maximal clique of $G$. Then, $S \subseteq V(G)$ is a block vertex deletion set of $G$ if and only if it is a block vertex deletion set for $G - v$.

        \item Let $G'$ be the graph obtained from $G$ by iteratively removing all vertices that lie in only one maximal clique. Then, $G'$ can be computed in $n^{\OO(1)}$ time and using $\tilde\OO(n)$ space. Moreover, if $G$ is a block-graph then $G'$ is the empty graph.

        \item Let $H'$ be the auxiliary graph of $G'$. Then, every vertex of $H'$ 
        has at least two distinct neighbors.
    \end{itemize}
\end{lemma}
\begin{proof}
    Consider the first statement. One direction is clear---if $S$ is a solution to $G$, then it is also a solution to $G - v$. In the reverse direction, consider the block graph $(G - v) - S$. Then observe that, adding back $v$ to $(G - v) - S$ results in a block graph as well, since the neighborhood of $v$ in $G$ is a clique.
    For the second statement, to compute the graph $G'$, we find a vertex $v$ in $G$ that lies in only one maximal clique. Here we apply Lemma~\ref{lemma:cliqueCountAlgo}, and it is clear that the required space and time bounds are obtained. Next, if $G$ were a block graph then it has a perfect elimination ordering (see Definition~\ref{defn:peo} and Proposition~\ref{prop:chordal-peo}) that is obtained by greedily removing simplicial vertices from $G$ one by one. Hence, iteratively removing all vertices of $G$ that lie in only one maximal clique results in an empty graph.
    Finally, for the third statement observe that, as each vertex in $G$ is part of two distinct maximal cliques, in $H$ it has at least two neighbors. Similarly, any maximal clique in $G$ contains at least two vertices, and hence it has at least two neighbors in $H$.
    %
\end{proof}

\begin{lemma}\label{lem:bvd:H-fvs-prop}
    Let $G$ be a $\{C_4, D_4\}$-free graph such that every vertex of $G$ lies in at least two maximal cliques. Further, suppose that there is no connected component of $G$ that is just a cycle. Let $H$ be the auxiliary graph of $G$.
    \begin{itemize}
        \item[(i)] In $n^{\OO(1)}$ time and $\OO(n)$ space, we can identify and mark each vertex $v \in V(G)$ such that $d_G(v) = 2$ and $d_G(u) = d_G(u') = 2$ where $u,u'$ are the two neighbors of $v$ in $G$.

        \item[(ii)] Let $P$ be a maximal path in $H$ such that all of its vertices are of degree 2 in $H$, $P$ is not a cycle and the end-vertices are from $V(G)$. Then all internal vertices of $P$ that lie in $V(G)$ are marked, and the two end-vertices are unmarked.

        \item[(iii)] If $v$ is a marked vertex, then there is a maximal path $P$ in $H$ such that all vertices of $P$ are of degree $2$, $P$ is not a cycle, the endpoints of $P$ lie in $V(G)$ and they are unmarked.

        \item[(iv)] Let $S \subseteq V(G)$ be a feedback vertex set of $H$.
        Then, there exists another feedback vertex set $S' \subseteq V(G)$  whose vertices are all unmarked and such that $|S'| \leq |S|$.
    \end{itemize}
\end{lemma}
\begin{proof}
    For $(i)$, we simply enumerate each vertex $v$ of degree $2$ in $G$, and then test if the two neighbors of $v$ are also of degree $2$. Here we note that every vertex of $G$ lies in at least 2 maximal cliques, and therefore there are no vertices of degree 1 in $G$. After this marking procedure, observe that if $Q$ is a maximal path in $G$ such that every vertex of $Q$ has degree $2$ in $G$, then all internal vertices of $Q$ are marked, while the two endpoints are unmarked. Since $V(G) \subseteq V(H)$, we can carry over the marking in $G$ to the auxiliary graph $H$.

    For $(ii)$, let is let $u,v$ be the two endpoints of $P$ and note that $u,v \in V(G)$. Let $u'$ and $u''$ be the neighbors of $u$ in $V(H) \setminus V(P)$ and $V(P)$, respectively. Note that $u',u'' \in V_{\cal C}$ correspond to maximal cliques of $G$. Consider $u'$ and let $K \in {\cal C}$ be the maximal clique of $G$ corresponding to it. Note that we have two cases: either the $K$ contains at least 3 vertices, or only 2 vertices. In the first case $u \in V(K)$ has at least $3$ neighbors in $G$, 2 in $V(K)$ and another in the maximal clique corresponding to $u''$. Hence $u$ is unmarked (indeed, we can further argue that $P$ must exclude $u$). In the second case $u'$ is of degree 2 in $H$, and it has a neighbor $x \in V(G)$ in $H$. Then $u$ and $x$ are adjacent in $G$ (forming the clique $C$ corresponding to $u'$). However, as $P$ is maximal and $x \notin V(P)$ it must be the case that $d_H(x) \geq 3$. Then it is clear that $d_G(x) \geq 3$. Hence, in the graph $G$, while $d_G(v) = 2$, one of it's neighbors $x$ has degree at least $3$. Therefore, once again, $u$ is not marked. Similar arguments apply to the vertex $v$, and hence $u,v$ are not marked.
    Now consider any internal vertex $w$ of $P$ that lies in $V(G)$, and observe that $d_G(w) =2$ and for the two neighbors $w',w''$ of $w$ in $G$, $d_G(w') = d_G(w'') = 2$. Hence, every internal vertex of $P$ that lie in $V(G)$ are marked.

    For $(iii)$, recall that $v$ is marked in $G$, only if $d_G(v) = 2$ and the two neighbors of $v$ in $G$, $u, w$ are also of degree $2$ in $G$. Then observe that, corresponding to the path $u-v-w$ in $G$, we have a path $P'$ in $H$ that consists of $u-x-v-y-w$ where $x,y \in V_{\cal C}$ correspond to the maximal cliques $\{u,v\}$ and $(v,w)$ of $G$, respectively. Next consider a maximal path $P''$ of $H$ that contains $P'$ such that all vertices of $P''$ are of degree $2$ in $H$. If $P''$ is a cycle of $H$ then by construction it corresponds to a cycle $C$ in $G$ whose vertices have degree $2$ in $G$. This means that $C$ is a connected component of $G$ that is just a cycle, which is forbidden by the premise of the lemma. Hence $P''$ is not a cycle in $H$, and it contains a maximal sub-path $P$ whose both-endpoints are in $V(G)$. By $(ii)$, the two endpoints of $P$ are unmarked.

    For $(iv)$, first observe that any vertex $v \in V(G) \subseteq V(H)$ such that $d_H(v) \geq 3$ must satisfy $d_G(v) \geq 3$, and hence it is unmarked. Next, observe that as $G$ has no connected component that is an isolated cycle, $H$ has no cycle whose all vertices are of degree $2$ in $H$. Therefore, any cycle $C$ of $H$ can be partitioned into maximal paths whose vertices are of degree 2 in $H$ and the remaining vertices that have degree 3 or more. Each such maximal path $P'$ contains another maximal path $P''$ such that the endpoints of $P''$ lie in $V(G)$. By $(ii)$, the endpoints of $P''$ are unmarked. Now consider the feedback vertex set $S \subseteq V(G)$ of $H$ and consider any marked vertex $w \in S$. By $(iii)$, $w$ lies in a maximal path $P$ of $H$ such that all vertices of $P$ are of degree 2, $P$ is not a cycle, the end-vertices of $P$ lie in $V(G)$ and they are unmarked.
    Since any cycle containing $w$ must contain both $u$ and $v$, it is clear $(S \setminus \{w\}) \cup \{u\}$ is another feedback vertex set of $H$ of same or smaller size. Therefore, by replacing every marked vertex in $S$, we can obtain a feedback vertex set $S'$ that contains only unmarked vertices and $|S'| \leq |S|$.
    This concludes the proof.
\end{proof}

Let us now describe our algorithm for {\sc Block Vertex Deletion} on the class of $\{C_4, D_4\}$-free graphs.

\begin{lemma}\label{lemma:BVD:special-case}
    Let $G$ be a $\{C_4, D_4\}$-free graph and let $k$ be an integer. Then there is a randomized algorithm that can decide if there is a block vertex deletion set of cardinality at most $k$ in $G$, and compute one if it exists. This algorithm runs in $17^k n^{\OO(1)}$ time, uses $\tilde\OO(nk)$ space and succeeds with constant probability $> 2/3$.
\end{lemma}
\begin{proof}
    Our algorithm is based on Proposition~\ref{prop:BVD:FVS} and a simple randomized branching algorithm for {\sc Feedback Vertex Set}~(\cite{becker2000randomized}; see also Chapter~5 in~\cite{CyganFKLMPPS15}). Let $H$ be the auxiliary graph of $G$, and recall that we seek a feedback vertex set of cardinality at most $k$ in $H$ that is a subset of $V(G)$. Note that, as the number of maximal cliques in $G$ could be $\OO(n^2)$, it is not possible for us to explicitly construct the graph $H$. Therefore, we must extract the required information about $H$ from $G$ in $\tilde\OO(n)$ space and $n^{\OO(1)}$ time.

    To simplify the algorithm, let us immediately deal with a special case. Suppose that $G$ contains a connected component that is an isolated cycle $C$ of length $4$ or more. Then, observe that corresponding to $C$ we obtain a connected component $C'$ of $H$ that is also a cycle whose vertices correspond to $V(C)$ and $E(C)$. For such cycles, we simply pick an arbitrary vertex of $C$ into our solution, decrease $k$ by 1, and remove the rest of $C$ from $G$. The optimality and correctness of this step is clear. When this reduction is not applicable, we can assume that there is no connected component of $G$ that is just a cycle. Hence, there is a vertex of degree at least 3 in every connected component of $G$.

    The randomized branching algorithm for feedback vertex set~\cite{CyganFKLMPPS15}, first applies some simple reduction rules to ensure that every vertex has at least $3$ distinct neighbors. Then, it can be shown that if $S$ is any feedback vertex set, then at least half the edges have an endpoint in $S$. In our setting, where space is bounded, applying these reduction rules becomes a little more challenging.
    Towards this, we first apply Lemma~\ref{lem:bvd:soln-prop} and Lemma~\ref{lem:bvd:H-fvs-prop}, which can be done in $\widetilde\OO(n)$ space and $n^{\OO(1)}$ time. Further, Lemma~\ref{lem:bvd:H-fvs-prop} ends up marking a subset of vertices of $V(G) \subseteq V(H)$. Then, we prove the following claim.

    \begin{claim}
    Let $S \subseteq V(G)$ is a feedback vertex set of $H$ that contains only unmarked vertices, and let $\tilde{E} \subseteq E(H)$ be the collection of all those edges that have an unmarked endpoint. Then if we pick an edge $e \in F$ uniformly at random, then with probability at least $1/17$, one endpoint of $e$ lies in $S$.
    \end{claim}
    To prove this claim, observe that by Lemma~\ref{lem:bvd:soln-prop}, every vertex in $H$ has two distinct neighbors. Moreover, by Lemma~\ref{lem:bvd:H-fvs-prop}, if $P$ is a maximal path made of degree-2 vertices in $H$, then all vertices of $P \cap V(G)$ except perhaps the two end-vertices of $P$ and their immediate neighbors in $P$ are marked. Let us consider every vertex in $V_{\cal C} \subseteq V(H)$ of degree-2 in $H$ whose both neighbors from $V(G)$ have been marked, as another marked vertex. Then for every maximal path $P$ of degree-2 vertices in $H$, all vertices are marked except for at most 6 vertices, the two endpoints and the 4 vertices in $P$ that closest to them. Note that, we don't need to explicitly mark these vertices, but only require them for analysis. Also note that if $v \in V(H)$ is marked then $d_H(v) = 2$. The proof of this claim, now easily follows from Lemma~5.1 of \cite{CyganFKLMPPS15} that states that when graph $X$ has minimum degree 3, and $Y$ is any feedback vertex set of $X$, then a randomly sampled edge of $X$ has an endpoint in $Y$ with probability at least $1/2$.

    Indeed, consider the graph $H'$ obtained by contracting each edge whose both endpoints are marked vertices in $H$ one by one, until every edge has at least one unmarked endpoint. Observe that there is a bijection between the edges of $H'$ and the edges of $H$ with at least one unmarked endpoint. Further observe that any maximal path of degree-2 vertices of $H'$ have length at most $7$. Moreover, any feedback vertex set of $H'$ is a feedback vertex set of $H$, and any feedback vertex set of $H$ that contains only unmarked vertices is a feedback vertex set of $G$. Now consider a feedback vertex set $X$ of $H$, and consider the forest $H' - X$. Let $Y = V(H') \setminus X$, and we claim that $|Y| \leq 16|E(X, Y)|$ where $E(X,Y)$ denotes the edges of $H'$ with one endpoint in $X$ and the other in $Y$. Let $V_1, V_2, V_3$ denote the subsets of vertices of $H' - X$ of degree $1$, $2$ and $\geq 3$ respectively. Note that $|V_3| \leq |V_1|$, by a standard result on forests. And, as the minimum degree of $H'$ is two, each vertex in $V_1$ has a distinct edge in $H'$ whose other endpoint lies in $X$, i.e. $|V_1| \leq |E(X,Y)|$. Hence, $|V_3| + |V_1| \leq 2|E(X,Y)|$. Next, to bound $|V_2|$ consider the forest $H''$ obtained from $H' - X$, by contracting every vertex of degree $2$. Then $|E(H'')| \leq |V_3| + |V_1| - 1$. Observe that each edge of $H''$ corresponds to a maximal path of degree 2 vertices in $H'-Y$ whose length is upper-bounded by $7$, and these paths partition $V_2$. Hence $|V_2| \leq 7 |E(H'')| \leq 7(|V_3| + |V_1| -1) \leq 14|E(X,Y)|$, and therefore we conclude that $|Y| \leq 16|E(X,Y)|$. Then, we conclude that $|E(H' - X)| \leq |Y| \leq 16|E(X,Y)|$. Finally, let $E(X)$ be the set of edges with at least one endpoint in $X$ and note that $|E(X)| \geq |E(X,Y)| \geq |V_1|$. And let $E(Y)$ be the set of edges with both endpoints in $Y$; $|E(Y)| = |E(H' - X)| \leq 16|E(X,Y)|$. Therefore the probability that an edge that is sampled uniformly at random from $E(H')$ has an endpoint in $X$ is at least $1/17$.

    Following this claim, and Proposition~\ref{prop:BVD:FVS}, a simple randomized algorithm could be as follows: sample an edge $e \in \tilde{E}$, then select the endpoint $v$ of $e$ from $V(G)$ into the solution, and then solve $(H-v, k-1)$. This algorithm would produce a solution in polynomial time with probability $1/17^k$.
    %
    Unfortunately, as $H$ is not explicitly constructed, we cannot easily sample edge of $\tilde{E}$. However, observe the following. Let $V' \subseteq V(G)$ be all the unmarked vertices. The probability that a vertex $v \in V'$ is selected by the above random process is equal to $d_H(v) / |\tilde{E}|$.
    Here we rely on the fact that if $v$ is an unmarked vertex in $H$, then all edges incident on $v$ are present in $\tilde{E}$.
    Hence, if we sample a vertex $v \in V(G)$ with probability $d_H(v)/|\tilde{E}|$, then with probability $1/17$, $v \in S$. To implement this sampling process, we only need to know the degrees of the vertices in $V(G)$ and as $H$ is bipartite we have $|\tilde{E}| = \sum_{v \in V(G)~:~ v \text{ is unmarked}} d_H(v)$. These values can obtained from Lemma~\ref{lemma:BVD:cliqueCount} in $n^{\OO(1))}$ time and $\tilde\OO(n)$ space. After sampling, $v$ we can recursively solve the instance $(H-v, k-1)$.
    Here the base case of this algorithm where $k=0$ is handled by Lemma~\ref{lem:bvd:soln-prop}. It allows us to test if $G$ is a block graph in polynomial time and $\tilde\OO(n)$ space. If this test fails, we output that there is no solution, or else there is a solution of size $0$ for $G$.
    Observe that this leads to a polynomial time algorithm that requires $\tilde\OO(nk)$ space and outputs a feedback vertex set $S \subseteq V(G)$ of size $k$ with probability at least $1/17^k$. By repeating this algorithm $\OO(17^k)$ times we obtain a simple randomized algorithm for {\sc Block Vertex Deletion} in $\{C_4, D_4\}$-free graphs that $17^k \cdot n^{\OO(1)}$ time and uses $\tilde\OO(nk)$ space. By standard arguments, it follows that this algorithm succeeds with probability $>2/3$ (which can be boosted to any constant by further repetitions).
\end{proof}

Finally, we extend the above lemma to an algorithm for {\sc Block Vertex Deletion}.

\begin{lemma}\label{lemma:BVD:post-algo}
    {\sc Block Vertex Deletion} admits a randomized FPT-algorithm that runs in time $17^k n^{\OO(1)}$ time $\tilde\OO(nk)$ space, where $k$ is the parameter denoting the solution size. It outputs a solution with constant probability.
\end{lemma}
\begin{proof}
    Let $G$ be the input graph. Our algorithm is implemented in two stages. In the first stage, we have a simple branching algorithm that finds $C_4$ or a $D_4$ in $G$ and branches on which of it's vertices lies in a solution of size $k$.
    This process can be represented via a branching tree of depth $k$, where each node has four child-nodes. At the leaf of each node, either we obtain a subgraph $G'$ that is $\{C_4, D_4\}$-free, or conclude that the choices made to arrive at that node does nnt lead to a solution of size $k$. Note that this stage of the algorithm can be implemented in time $\OO(4^k n^{4})$ and space $\OO(k \log n)$. In the next stage, we apply Lemma~\ref{lemma:BVD:special-case} to the graph $G'$ along with the parameter $k'$. Here $k' = k ~- ($the number of vertices selected in Stage 1$)$. It follows that Stage~$2$ runs in time $17^k n^{\OO(1)}$ and $\tilde\OO(nk)$ space, and succeeds with constant probability.
    Hence, we obtain the required algorithm for {\sc Block Vertex Deletion}.
\end{proof}

From Theorem~\ref{lem:reductionLemma}, Lemma~\ref{lemma:BVD:reconst} and Lemma~\ref{lemma:BVD:post-algo} we obtain Theorem \ref{thm:blockGraphs}. Here we also rely on the fact that the class of block graphs are hereditary.

\section{Cut Problems}\label{sec:cutProblems}
Our starting point is the sampling primitive of Guha, McGregor and Tench~\cite{GuhaMT15} (with appropriately set parameters). The main idea here is to sample roughly $\wtilde\bigoh(k^{\bigoh(1)})$ vertex subsets of the input graph so that with good probability, for every set $S\subseteq V(G)$ that is disjoint from some forbidden substructure (of potentially unbounded size) in $G$, there is a set of sampled subgraphs that are disjoint from $S$ and whose union also contains a forbidden substructure for the problem concerned. This allows us to produce $\wtilde\bigoh(k^{\bigoh(1)}n)$ subgraphs that can still certify when a vertex set is not a solution for the problem at hand. However, this is only the first part of our approach and problem-specific steps need to be taken following this.

\thinblackbox{{\sf Sampling Primitive~($V(G),n,k,q$)\footnote{We assume that $n,k,q\geq 1$.}}:  Sample	$\ell=64qk^3 \ln n$ vertex subsets $V_1,\dots, V_\ell$ of $G$. Each $V_i$ is sampled by independently picking each vertex of $G$ with probability $1/2k$. Let $G_i=G[V_i]$ for every $i\in[\ell]$, let ${\cal L}=\{G_1,\dots, G_\ell\}$ and for every  $S\subseteq V(G)$, let ${\cal L}_S=\{G_i\in {\cal L}\mid V(G_i)\cap S=\emptyset \}$. }

We now summarize a useful property of the sampling primitive, which is required for the correctness of some of our algorithms.

\begin{lemma}[Properties of {\sf Sampling Primitive~$(V(G),n,k,q)$}]\label{lem:samplingCorrectness}
Fix a set
 $\{\CC_1,\dots, \CC_r\}$ where $r\leq n^q$ and each $\CC_i$ is a set of subgraphs of $G$.
With probability at least $1-2/n$, the following event occurs:
%
 For every $k$-vertex set $S\in V(G)$ and for every $\CC_i$, if $S$ is disjoint from some $C\in \CC_i$, then $\bigcup_{G\in {\cal L}_S}G$ contains some $C'\in \CC_i$.
\end{lemma}

\begin{proof}
	Let $\cS_k$ denote the set of all $k$-vertex sets in $G$.  For every fixed $S\in \cS_k$, the expected number of sampled vertex subsets disjoint from $S$,  $$\mean{|{\cL}_S|}=(1-1/2k)^{k}\cdot \ell\geq \ell/4=16qk^3 \ln n.$$

Using standard Chernoff bounds we get: $$\prob{|{\cal L}_S|<8qk^3 \ln n}\leq e^{2qk^3 \ln n}\leq  1/n^{2k}.$$

Let ${\bf Y}$ denote the event that for every $S\in \cS_k$, $|{\cal L}_S|\geq 8qk^3 \ln n$. Taking the union bound over all $S\in \cS_k$, we conclude that:

$$\prob{~\exists S\in \cS_k: |{\cal L}_S|<8qk^3 \ln n}=\prob{\overline {\bf Y}}\leq 1/n^{k}.$$
We say that $S\in \cS_k$ is {\em good} if for every $i\in [r]$ at least one element of $\CC_i$ disjoint from $S$ is contained in the union of the graphs in ${\cal L}_S$. We say that $S\in \cS_k$ is {\em bad} otherwise. Let $\bf X$ denote the event that every $S\in \cS_k$ is good. Recall that: $$\prob{\overline {\bf X}}\leq \prob{\overline {\bf Y}}+\prob{\overline{{\bf X}}\mid {\bf Y}}.$$


Now, we observe the following for a fixed $S\in \cS_k$:
$$\prob{S\textrm{ is bad}\mid {\bf Y}}\leq n^q\cdot n^2\cdot (1-1/(2k)^2)^{8qk^3 \ln n}\leq 1/n^{2k}.$$

Here, for each $i\in [r]$, we fixed a graph $C$ in $\CC_i$, used union bound along with the fact that the number of edges in $C$ is at most $n^2$ and then took union bound over the $\CC_i$'s. Finally, taking the union bound over all $S\in \cS_k$, we conclude that:

$$\prob{~\exists S\in
\cS_k: S\textrm{ is bad}\mid {\bf Y}}=\prob{\overline{{\bf X}}\mid {\bf Y}}\leq 1/n^{k}.$$

Therefore, with probability at least $1-2/n^k\geq 1-2/n$, we conclude that every $k$-vertex set in $\cS_k$ is good. This completes the proof of the lemma.\end{proof}

 In the rest of this section, we use the above sampling primitive as the first step in our fixed-parameter semi-streaming algorithms for well-studied cut problems. As discussed above, this primitive identifies $\wtilde\bigoh(k^{\bigoh(1)}n)$ subgraphs that can still certify when a vertex set is not a solution for the problem at hand. {\em However, these subgraphs may still be large and dense}. Therefore, for the problems we consider in this section, we provide a problem-specific sparsification step that allows us to shrink the number of edges in each subgraph while ensuring that non-solutions can be witnessed by a substructure in the union of the sparsified subgraphs.

  \subsection{Odd Cycle Transversal}

Recall that in the {\sc Odd Cycle Transversal} (OCT) problem, the input is a graph $G$ and integer $k$ and the goal is to decide whether there is a set of at most $k$ vertices that intersect every odd cycle in $G$.

Before we proceed, we recall the notion of the {\em bipartite double cover} of a graph.

\begin{definition}[\cite{AhnGM12}]\label{def:bipartiteDoubleCover}
	For a graph $G$, the {\em bipartite double cover} of $G$ is the graph $D$  constructed from $G$ as follows. For each $v\in V(G)$, construct $v_a,v_b\in V(D)$ and for each edge
$(u, v) \in  E(G)$, create two edges $(u_a, v_b)$ and $(u_b, v_a)$.
\end{definition}

The above definition was used by Ahn, Guha and McGregor~\cite{AhnGM12} to give the first  dynamic semi-streaming algorithm to test bipartiteness. In this work, we use the properties of bipartite double covers to generalize their result and give a fixed-parameter semi-streaming algorithm for checking if a graph is at most $k$ vertices away from being bipartite, i.e., OCT.

\oct*

\begin{proof} If $k=0$, then we simply need to check whether the input graph is bipartite. For this, we use the aforementioned algorithm of Ahn, Guha and McGregor~\cite{AhnGM12}. Hence, we may assume that $k\geq 1$. Let $G$ be the input graph. We run {\sf Sampling Primitive}~($V(G),n,k,2$) to generate the vertex sets $V_1,\dots, V_\ell\subseteq V(G)$ and then process the input stream as follows.
	   For each sampled vertex set $V_i$, we use Proposition~\ref{prop:dynamicConnectivitySketch} to maintain a connectivity sketch  for the bipartite double cover of $G_i$, which is denoted by $G_{i}^{\sf aux}$ (see Definition~\ref{def:bipartiteDoubleCover}). Recall that $G_{i}^{\sf aux}$ has exactly $2|V_i|$ vertices. Therefore, the  space requirement for each $V_i$ is bounded by $\wtilde{\bigoh}(n)$.
	   This completes the processing step.
  Since $\ell=\wtilde{\bigoh}(k^{\bigoh(1)})$, the space used by our processing step is $\wtilde{\bigoh}(k^{\bigoh(1)}\cdot n)$ as required.

 It remains for us to describe our fixed-parameter post-processing routine.
Recall that for any $c>0$ and $i\in [\ell]$, with probability at least $1-1/n^c$, we can recover a spanning forest of $G_{i}^{\sf aux}$. By choosing $c$ to be large enough and taking the union bound over polynomially many subgraphs, we obtain a spanning forest $T_i$ of $G_{i}^{\sf aux}$ for each $i\in [\ell]$ and conclude that with probability at least $1-1/n$,
each $T_i$ is indeed a spanning forest of $G_{i}^{\sf aux}$. Now, we define $H_i$ to be the subgraph of $G_i$ induced by edges that contribute a copy to $T_i$. That is, an edge $e=(u,v)\in E(G_i)$ is added to $H_i$ if and only if $(u_a,v_b)$ or $(u_b,v_a)$ is contained in  $E(T_i)$.
Notice that the number of edges in each $H_i$ is bounded by $\bigoh(n)$, implying that the number of edges in the graph $\bigcup_{i\in [\ell]}H_i$ is $\wtilde\bigoh(k^{\bigoh(1)}n)$.

\begin{claim}\label{clm:thm:nodeULC}
With probability at least $1-3/n$, the following holds: For every $S\subseteq V(G)$ of size $k$, if there is an odd cycle in $G-S$, then there is an odd cycle present in the subgraph  $\bigcup_{G_i\in {\cal L}_S}H_i$.
\end{claim}

 \begin{proof}
 Observe that an undirected graph has an odd cycle if and only if it has an odd-length closed walk and let $\CC_1$ denote the set of odd-length closed walks in  $G$. Invoking Lemma~\ref{lem:samplingCorrectness} with $r=q=1$, we conclude that with probability at least $1-2/n$, if $S$ is a set of size at most $k$ such that there are no odd-length closed walks in the graph $\bigcup_{G_i\in {\cal L}_S}G_i$, then $S$ hits all odd-length closed walks in $G$.

 Towards the proof of the claim we next argue that for each $i\in [\ell]$ and $(u,v)\in E(G_i)$, there is an odd $u$-$v$ walk in $H_i$. That is, even if we have ignored the edge $(u,v)$ in the sparsification, we would have at least preserved an odd $u$-$v$ walk, which we later argue would suffice for our purposes.
  Indeed, notice that since $(u,v)\in E(G_i)$, it follows that $(u_a,v_b)\in E(G_{i}^{\sf aux})$. Therefore, it must be the case that $u_a$ and $v_b$ lie in the same tree in the spanning forest $T_i$. Call this tree $\Gamma_i^{u,v}$ and let $e_1,\dots, e_r$ be the edges that comprise the unique $u_a$-$v_b$ path in $\Gamma_i^{u,v}$ and notice that the edges of $G_i$ corresponding to the edges in $\{e_1,\dots, e_r\}$ are present in $H_i$ by definition. Moreover, notice that $r$ must be odd since $u_a$ and $v_b$ lie on opposite sides of the bipartite graph $G_{i}^{\sf aux}$.  Hence, we obtain a $u$-$v$ walk in $H_i$ containing an odd number of edges as required. Given this, we are ready to complete the proof of the claim.

  We have already observed that with probability at least $1-2/n$, for any vertex set $S$ of size at most $k$, if there is an odd cycle in $G-S$ (and hence an odd-length closed walk), then there is an odd-length closed walk $C^S=x^S_1,\dots, x^S_{r_S},x^S_1$ in $\bigcup_{G_i\in {\cal L}_S}G_i$.
  Fix an $S$ such that $G-S$ contains an odd cycle and let $e_1,\dots, e_{r_S}$ be the edges in the odd-length closed walk $C_S$. We have argued above that for every $j\in [r_S]$ and $e_j=(x_j,x_{j+1 ({\sf mod }~r_S)})$, there is an odd $x_j$-$x_{j+1 ({\sf mod}~r_S)}$ walk in $\bigcup_{G_i\in {\cal L}_S}H_i$, implying an odd-length closed walk in $\bigcup_{G_i\in {\cal L}_S}H_i$ as required. This completes the proof of the claim.
   \end{proof}

 Given the above claim, we conclude that with probability at least $1-3/n$, any set $S\subseteq V(G)$ is a solution for the OCT instance $(G,k)$ if and only if  $S$ is a solution for the OCT instance $(G',k)$ where $G'=\bigcup_{G_i\in {\cal L}}H_i$. Therefore, our post-processing algorithm computes the sparsified instance $(G',k)$ with $\wtilde\bigoh(k^{\bigoh(1)}n)$ edges and invokes the $\bigoh(3^kk^3(m+n))$-time $k^{\bigoh(1)}(m+n)$-space algorithm from \cite{KolayMRS20} for {\sc OCT}.
 This completes the proof of the theorem.
\end{proof}

\subsection{Subset Feedback Vertex Set}
Recall that in the {\sc Subset Feedback Vertex Set} problem, the input is a graph $G$, a set $T\subseteq V(G)$ of vertices called terminals and an integer $k$. The goal is to determine whether there is a set $S\subseteq V(G)$ of size at most $k$ that intersects all cycles in $G$ that contain a vertex of $T$. We assume the following input model for instances of this problem. For each terminal, when the first edge incident on it is inserted, it is inserted along with a pair of bits that indicate whether (and if so, which of) the endpoints of this edge are terminals.

\subsetFVS*

\begin{proof}
Let $G$ be the input graph.
We run {\sf Sampling Primitive}~($V(G),n,k+1,2$) to generate the vertex sets $V_1,\dots, V_\ell\subseteq V(G)$ and then
 process the input stream as follows. We use Proposition~\ref{prop:kSparseRecovery} to construct a $(k+1)n$-sparse recovery data structure $\cX$ of size $\wtilde{\bigoh}(kn)$ and for each terminal, we insert all edges adjacent to it, into this data structure. We also have  a $\bigoh(\log n)$-bit counter that counts the number of such edges. Moreover,  for each vertex set $V_i$ produced by the sampling primitive, we maintain a dynamic connectivity sketch of $G_i=G[V_i]$ (call this data structure $D_i$) using Proposition~\ref{prop:dynamicConnectivitySketch}.
This completes the processing.  Since $\ell=\wtilde{\bigoh}(k^{\bigoh(1)})$ and $\cX$ has size $\wtilde{\bigoh}(kn)$ the space used by our processing step is $\wtilde{\bigoh}(k^{\bigoh(1)}\cdot n)$ as required.

   In the post-processing step, we first check the counter to determine whether the number of edges incident on terminals exceeds $(k+1)n$. If the answer is yes, then we return {\sc No}. Otherwise, we do the following. We recover all edges incident on $T$ by extracting all elements from $\cX$ (see Proposition~\ref{prop:kSparseRecovery}). Call this set $E_{\cX}$.
   For every $i\in [\ell]$, we construct a spanning forest of $G_i$ from $D_i$ (Proposition~\ref{prop:dynamicConnectivitySketch}). Choosing $c$ to be large enough and applying the union bound over all $\{G_i\mid i\in [\ell]\}$, we conclude that with probability at least $1-1/n$, for every $i\in [\ell]$, $H_i$ is a spanning forest of $G_i$.
We now run the $25.6^k (m+n)$-time $k^{\bigoh(1)}(m+n)$-space randomized algorithm from \cite{LokshtanovRS18}\footnote{This algorithm is in fact a $k^{\bigoh(1)}(m+n)$-time (and so, also a $k^{\bigoh(1)}(m+n)$-space) algorithm with success probability $1/25.6^k$.} for $\bigoh(\log n)$ repetitions to solve the instance $(E_{\cX}\cup \bigcup_{i\in [\ell]}H_i,T,k)$.

  We now prove the correctness of our algorithm.

  \begin{claim} The following statements hold.
  \begin{enumerate}
  \item If $(G,T,k)$ is a yes-instance, then the number of edges incident to $T$ is at most $(k+1)n$.
\item With probability at least $1-3/n$, every   	$S\subseteq V(G)$ is a solution for $(G,T,k)$ if and only if $S$ is a solution for $(E_{\cX}\cup \bigcup_{G_i\in {\cal L}}H_i,T,k)$.

\end{enumerate}
  \end{claim}

	\begin{proof}
		For the first statement, fix a solution $S$ for $(G,T,k)$ and consider the edges incident on $T$ that do not have a vertex in $S$ as an endpoint. It is sufficient to prove that the number of such edges is at most $n$. And indeed, this is the case since each edge incident on a terminal must form a biconnected component of $G-S$ on its own and the number of biconnected components of a graph is bounded by $n$.

		Consider the second statement.
	 	 Let $G'$ be the subgraph of $G$ obtained by taking the subgraph $\bigcup_{G_i\in {\cal L}}H_i$ and then adding the edges in
	$E_{\cX}$.  Since $G'$ is a subgraph of $G$, it follows that if $S$ is a solution for $(G,T,k)$, it is also a solution for $(G',T,k)$.

	 We now consider the converse direction.
	For every $u,v\in V(G)$, let $\CC_{uv}$ denote the set of all $u$-$v$ paths in $G$. Consider the set $\{\CC_{uv}\mid u,v\in V(G)\}$ which has size at most $n^2$.	Invoking Lemma~\ref{lem:samplingCorrectness}  guarantees that with probability at least $1-2/{n}$, for every set $P\subseteq V(G)$ of size at most $k+1$, and every $u,v\in V(G)$, if there is a $u$-$v$ path in $G-P$ then there is a $u$-$v$ path in $G'-P$. Henceforth, we assume that this event indeed occurs.
 Now, suppose that $S$ is a solution for $(G',T,k)$, but is not a solution for $(G,T,k)$. Let $t\in T$ be such that $C$ is a $t$-cycle in $G-S$ which is not present in $G'-S$. Recall that every block of $G'-S$ that $t$ is contained in, consists of a single edge incident on $t$. Moreover, since $G'$ contains every edge in $G$ that is incident on $t$, it follows that every edge in $G$ that is incident on $t$ but not incident on $S$ is already present in $G'-S$.
	We now observe that for $t$ to have a flow of 2 to some vertex in $G-S$, there must be a pair of adjacent vertices $u,v\notin T$ that lie on a $t$-cycle in $G-S$,
	but which have no path between them in $G'-S-\{t\}$ (i.e., they lie in distinct blocks of $G'-S$ that are separated by $t$).
Consider the set $S'=S\cup \{t\}$ and notice that we have inferred the existence of a $u$-$v$ path in $G-S'$. As already argued using Lemma~\ref{lem:samplingCorrectness}, there must be a $u$-$v$ path in $G'-S'$, a contradiction.
	 This completes the proof of the claim.
		\end{proof}
Since the instance $(E_{\cX}\cup \bigcup_{i\in [\ell]}H_i,T,k)$ has $n$ vertices and $\wtilde\bigoh(k^{\bigoh(1)}n)$ edges, the theorem follows.
\end{proof}

\subsection{Multiway Cut}
Recall that in the {\sc Multiway Cut} problem, the input is a graph $G$, a set $T\subseteq V(G)$ of vertices called terminals and an integer $k$. The goal is to determine whether there is a set $S\subseteq V(G)$ of size at most $k$ such that there is no connected component of $G-S$ containing two terminals. As was the case for {\sc Subset Feedback Vertex Set}, we assume the following input model for instances of this problem. For each terminal, when the first edge incident on it is inserted, it is inserted along with a pair of bits that indicate whether (and if so, which of) the endpoints of this edge are terminals.

\multiwaycut*

\begin{proof}
	We run {\sf Sampling Primitive}~($V(G),n,k,2$) to generate the vertex sets $V_1,\dots, V_\ell\subseteq V(G)$ and then
 process the input stream as follows. For each vertex set $V_i$ produced by the sampling primitive, we maintain a dynamic connectivity sketch of $G_i=G[V_i]$ (call this data structure $D_i$) using Proposition~\ref{prop:dynamicConnectivitySketch}.
This completes the processing.  Since $\ell=\wtilde{\bigoh}(k^{\bigoh(1)})$,  the space used by our processing step is $\wtilde{\bigoh}(k^{\bigoh(1)}\cdot n)$ as required.

   In the post-processing step, for every $i\in [\ell]$, we construct a spanning forest of $G_i$ from $D_i$ (Proposition~\ref{prop:dynamicConnectivitySketch}). Choosing $c$ to be large enough and applying the union bound over all $\{G_i\mid i\in [\ell]\}$, we conclude that with probability at least $1-1/n$, for every $i\in [\ell]$, $H_i$ is a spanning forest of $G_i$.
We now run the known $\bigoh(4^kk^{\bigoh(1)}(m+n))$-time $k^{\bigoh(1)}(m+n)$-space fixed-parameter algorithm for {\sc Multiway Cut} (see, for example, \cite{CyganFKLMPPS15}) on the instance $(\bigcup_{i\in [\ell]}H_i,T,k)$ which has $n$ vertices and $\wtilde\bigoh(k^{\bigoh(1)}n)$ edges.

  We now prove the correctness of our algorithm.

  \begin{claim}
 With probability at least $1-3/n$, every   	$S\subseteq V(G)$ is a solution for $(G,T,k)$ if and only if $S$ is a solution for $(\bigcup_{G_i\in {\cal L}}H_i,T,k)$.
  \end{claim}

	\begin{proof}

	Let $G'=\bigcup_{G_i\in {\cal L}}H_i$.  Since $G'$ is a subgraph of $G$, it follows that if $S$ is a solution for $(G,T,k)$, it is also a solution for $(G',T,k)$.

	 We now consider the converse direction.
	For every $u,v\in V(G)$, let $\CC_{uv}$ denote the set of all $u$-$v$ paths in $G$. Consider the set $\{\CC_{uv}\mid u,v\in V(G)\}$ which has size at most $n^2$.	Invoking Lemma~\ref{lem:samplingCorrectness}  guarantees that with probability at least $1-2/{n}$, for every set $P\subseteq V(G)$ of size at most $k$, and every $u,v\in V(G)$, if there is a $u$-$v$ path in $G-P$ then there is a $u$-$v$ path in $G'-P$. Henceforth, we assume that this event occurs.
	Now, suppose that $S$ is a solution for $(G',T,k)$, but is not a solution for $(G,T,k)$. Let $t_1,t_2\in T$ be such that $Q$ is a $t_1$-$t_2$ path in $G-S$.
	But then,  there must be a $t_1$-$t_2$ path in $G'-S'$ also, yielding a contradiction.
	 This completes the proof of the claim.
		\end{proof}
Since the instance $(\bigcup_{i\in [\ell]}H_i,T,k)$ has $n$ vertices and $\wtilde\bigoh(k^{\bigoh(1)}n)$ edges, the theorem follows.
\end{proof}

Notice that in contrast to the static setting, the {\em edge deletion} variants of the above problems cannot be trivially reduced to the vertex deletion variants by, say, subdividing edges. This is because the number of vertices now blows up, plus the two edges that correspond to a single original edge may now appear far apart in time. Hence, one has to treat the edge deletion variants separately. Our focus in this paper is on the vertex deletion variants and we leave the edge deletion variants for future work.

\section{Our Refinement of The {\semips} Class}\label{sec:refiningSemiPS}

We recall the class {\semips} introduced by Chitnis and Cormode~\cite{ChitnisC19} and then formally define our new notions related to parameterized semi-streaming algorithms.

\begin{definition}[{\semips}]{\sf \cite{ChitnisC19}}
	The class {\semips} comprises precisely those parameterized problems that are solvable using $\wtilde{\bigoh}(f(k)\cdot n)$ bits for some computable function $f$, allowing unbounded computation at each edge update, and also at the end of the stream, i.e., in post-processing.
\end{definition}

\begin{definition}[{\fptsemips} and FPSS algorithms]
	The class {\fptsemips} comprises precisely those parameterized graph problems that are solvable using a $\wtilde{\bigoh}(f(k)\cdot n)$-space streaming algorithm for some computable function $f$,
	allowing $f(k)n^{\bigoh(1)}$-time at each edge update, and also at the end of the stream, i.e., in post-processing. We call algorithms of this form, {\em fixed-parameter semi-streaming algorithms}, or {\em {FPSS}-algorithms}.
\end{definition}

Clearly, {\fptsemips} is contained in {\semips}, but the converse (restricted to graph problems) is not true as {\semips} contains W-hard problems, whereas {\fptsemips} is clearly contained in {\sf FPT}.  We recall that, in this paper all FPSS-algorithms we designed only require $\wtilde{\bigoh}(k^{\bigoh(1)}\cdot n)$-space and polynomially bounded time at each edge update. Therefore, we in fact have worked within a further refinement of {\fptsemips} that is already very rich in the amount and type of graph problems it contains as we have shown.

We next refine the class {\fptsemips} by introducing the notion of (polynomial) semi-streaming kernelization and also introduce the weaker notion of (polynomial) semi-streaming compression. 	Recall that Fafianie and Kratsch~\cite{fafianie2014streaming} first introduced the notion of {\em streaming kernelization}, which, for a parameterized problem $P$ is a polynomial-time $k^{\bigoh(1)}\log n$-space streaming algorithm $\cA$  that, on input $(I,k)$, outputs an instance $(I',k')$ of $P$, where $(I,k)$ is a yes-instance if and only if $(I',k')$ is a yes-instance and $|I'|,k'\leq f(k)$ for some polynomially bounded function $f$. However, for vast numbers of well-studied parameterized graph problems, this definition places insurmountable constraints and so, we introduce an appropriate notion of {\em semi-streaming kernelization}. We show through Theorem~\ref{thm:FPT-semi-PS-Kernels-equivalence} that this definition is robust and consistent with our definition of parameterized semi-streaming algorithms and mimics the relation between (static) kernelization and (static) parameterized algorithms.

\begin{definition}
	[Semi-streaming Kernelization]
	Let $P$ be a graph problem parameterized by  $k\in {\mathbb N}$. A {\em semi-streaming kernelization} for $P$ is an $\widetilde{\OO}(k^{\bigoh(1)}n)$-space streaming algorithm $\cA$ with polynomially bounded time at each edge update and polynomial post-processing time that, on input $(I,k)$, outputs an instance $(I',k')$ of $P$, where $(I,k)$ is a yes-instance if and only if $(I',k')$ is a yes-instance and $|I'|,k'\leq f(k)$ for some computable function $f$. If $f$ is a polynomial function, then we say that $\cA$ is a {\em polynomial} semi-streaming kernelization. 	\end{definition}

	\begin{definition}
	[Semi-streaming Compression]
	Let $P$ be a graph problem parameterized by  $k\in {\mathbb N}$ and let $Q$ be any language. A {\em semi-streaming compression} of $P$ into $Q$ is an  $\widetilde{\OO}(k^{\bigoh(1)}n)$-space streaming algorithm $\cA$ polynomially bounded time at each edge update and polynomial post-processing time that, on input $(I,k)$, outputs an instance $I'$ of $Q$ such that $(I,k)$ is a yes-instance of $P$ if and only if $I'$ is a yes-instance of $Q$ and $|I'|\leq f(k)$ for some computable function $f$. If $f$ is a polynomial function, then we say that $\cA$ is a {\em polynomial} semi-streaming compression of $P$ into $Q$. We also simply say that $\cA$ is a {(polynomial) semi-streaming compression} for $P$.
	\end{definition}

	\begin{theorem}\label{thm:FPT-semi-PS-Kernels-equivalence}
		Let $P$ be a parameterized graph problem. Then, $P$ is in {\fptsemips} if and only if it is decidable and has a semi-streaming kernelization.
	\end{theorem}

	\begin{proof}
	Let $P$ be parameterized by  $k\in {\mathbb N}$.
		Suppose that $P\in$ {\fptsemips} and let $\cA$ be a $\wtilde{\bigoh}(f(k)\cdot n)$-space streaming algorithm for $P$ that uses $f(k)n^{\bigoh(1)}$-time for post-processing and between edge updates. Decidability trivially follows from the existence of $\cA$. If $f(k)>\log n$, then, we can simply store the input graph using space bounded by $2^{2f(k)}$ and output it as the required equivalent instance of $P$. On the other hand, if $f(k)<\log n$, then $\cA$ itself is a $\wtilde{\bigoh}(n)$-space streaming algorithm with polynomially bounded post-processing time and polynomially bounded time between edge updates. Hence, we can use $\cA$ to solve the input instance and output an equivalent trivial instance of $P$. This gives a semi-streaming kernelization for $P$.

		Conversely, suppose that $P$ has a semi-streaming kernelization $\cA$. Then, we run $\cA$ on the input instance, say $(I,k)$, and use decidability to solve the computed equivalent instance $(I',k')$ (whose size is bounded by some $f(k)$) using additional space bounded by $g(k)$ for some computable $g$. This is a fixed-parameter semi-streaming algorithm for $P$ as required.
	\end{proof}

As in the case of static kernelization, the primary classification goal in semi-streaming kernelization is to identify problems that have polynomial semi-streaming kernelization and for those that do, to obtain outputs whose size can be bounded by the slowest-growing polynomial function possible. 
We extend the definitions of semi-streaming kernelization and compression from single passes to multiple passes in the natural way. 

\section{Lower Bounds}\label{sec:lowerBounds}

In this section, we show lower bounds implying that (i) there are problems in {\sf FPT} that are not in {\fptsemips}, and (ii) our corollary for Theorem~\ref{thm:Hcovering} when $\cH$ is the class of split graphs is tight (for $\cH$ being cluster graphs, a lower bound is already known from \cite{ChitnisCEHMMV16}).

We use the  two party communication model, introduced by Yao~\cite{yaocc} to prove
lower bounds on streaming algorithms.  In this model of communication, there are two players, Alice and Bob, holding inputs $x\in X$ and $y\in Y$ respectively, where $X$ and $Y$ are two arbitrary sets. Their objective is  to compute a given function $f~:~X\times Y \rightarrow \{0,1\}$, by communicating as few bits as possible. It is assumed that both players have infinite computational power.
The minimum number of bits communicated, for any pair of inputs $(x,y)$, to compute the function $f$, is called the (deterministic) communication complexity of $f$, denoted by $D(f)$.
For a randomized protocol $P$, we say $P$ has error probability at most $\varepsilon$ if $\Pr(P (x, y) = f (x, y)) \geq 1-\varepsilon$ for all inputs $x\in X$ and $y\in Y$. Here, the randomness is over the private coin tosses of Alice and Bob. The  randomized communication complexity of $f$, denote by $R(f)$, is the minimum communication cost over all the randomized  protocols for $f$ with error probability at most $1/3$.
In the one-way communication model, only Alice is allowed to send one message to Bob and no other message is allowed to be sent.   We use $D^{1-\text{way}}(f)$ and $R^{1-\text{way}}(f)$ to denote the one way deterministic and randomized communication complexity of $f$, respectively.
Notice that $D(f) \geq R(f)$ and$D^{1-\text{way}}(f) \geq R^{1-\text{way}}(f)$.

Lower bounds on $D^{1-\text{way}}(f)$ and $R^{1-\text{way}}(f)$ can be used for proving lower bounds on
deterministic and randomized $1$-pass streaming algorithms via standard reductions. That is, given a function $f$, Alice creates a stream $\sigma_x$ from her input $x$ and Bob creates a stream $\sigma_y$ from his input $y$.  First Alice runs the streaming algorithm on $\sigma_x$ and passes the state of the algorithm to Bob. Then, Bob  continues the execution of the streaming algorithm on $\sigma_y$.  Finally,   the output of the streaming algorithm is used to solve the problem $f$. This provides
space complexity lower bound of the streaming algorithm to be at least $D^{1-\text{way}}(f)$ and $R^{1-\text{way}}(f)$.

Next, we describe some well known communication problems and state their communication lower bounds. 

\paragraph*{INDEX.}
The {\sf INDEX} problem is one of the  most commonly used communication problems  to prove lower bounds on streaming algorithms. In the {\sf INDEX} problem, Alice is given an $n$-length bit string $x_1\ldots x_n \in \{0,1\}^n$ and Bob is given an index $i\in [n]$.  Bob wants to find the bit $x_i$.

\begin{proposition}[\cite{KremerNR99}]
\label{prop:index}
$R^{1-\text{way}}({\sf INDEX})=\Omega (n)$.
\end{proposition}

\paragraph*{GRAPH INDEX.}
In this problem, Alice is given a bipartite graph $G$ with bipartition $V(G)=A\uplus B$ such
$\vert A\vert =\vert B\vert=n$. Bob is given a pair $(a,b)\in A\times B$ and he wants to test whether $\{a,b\}\in E(G)$ or not. A reduction from {\sf INDEX} implies that $R^{1-\text{way}}({\sf GRAPH\; INDEX})=\Omega (n^2)$ (see Proof of Theorem 2.1~\cite{ChakrabartiG0V20}).

\paragraph*{Chordal Recognition.}
We prove that any $1$-pass streaming algorithm to recognize a chordal graph requires $\Omega(n^2)$ space.


  \begin{theorem}\label{thm:chordalLowerBound}
Any $1$-pass streaming algorithm to test whether a graph is chordal  requires $\Omega(n^2)$ space.
\end{theorem}

\begin{proof}
We give a reduction from {\sf GRAPH INDEX}. In this problem, Alice is given a bipartite graph $G$ with bipartition $A \cup B$, where $\vert A\vert=\vert B\vert=n$,  $A=\{a_1,\ldots,a_n\}$ and $B=\{b_1,\ldots,b_n\}$. Bob is given a pair $(a_i,b_j)\in A\times B$ and he want to test whether $(a_i,b_j)\in E(G)$. Alice and Bob will construct a graph $H$  such that $H$ is chordal if and only if $(a_i,b_j)\in E(G)$ as follows. The vertex set of $H$ is $A\cup B \cup \{x_i,y_i ~\colon~i\in [n]\}$. Alice will add all the edges of $G$ to $H$ and make $A\cup \{x_1,\ldots,x_n\}$ a clique. Then, Alice will add the set of edges $\{(x_r,y_{r'})~\colon~ r\neq r'\}$.
Recall that $(a_i,b_j)$ is the pair Bob has. Bob will add the edges $(a_i,y_i)$, $(y_i,b_j)$, and $\{(b_j,x_r)~\colon~r\in [n]\}$. This completes the construction of $H$.

 Next we prove that $H$ is a chordal graph if and only if $(a_i,b_j)\in E(G)$. From the construction of $H$, we have that $C=a_i,y_i,b_j,x_i,a_i$ is a cycle of length $4$ and $(x_i,y_i)\notin E(H)$. Now, if $(a_i,b_j)\notin E(G)$, then $(a_i,b_j)\notin E(H)$ and hence $C$ is an induced $C_4$ in $H$. That is, if $(a_i,b_j)\notin E(G)$, then $H$ is not a chordal graph. Next, we prove the other direction of the correctness proof. That is, if $(a_i,b_j)\in E(G)$, then $H$ is a chordal graph. Let $e=(y_i,b_j)$.  First, it is easy to see that $H-e$ is a split graph with $A'=A\cup \{x_1,\ldots,x_n\}$ being a clique  and $B'=B\cup \{y_1,\ldots,y_n\}$ being an independent set in $H-e$. That is, $e$ is the only edge present in $H[B']$. For the sake of contradiction, suppose $H$ is not a chordal graph. Then, by the structure of $H$, there is an induced $C_4$ in $H$ containing $e$.  Let $y_i,b_j,q,q'$ be an induced $C_4$ in $H$. As $H'$ is a split graph, we have that $q,q'\in A'$.
 As $\{(b_j,x_r)~\colon~r\in[n]\}\subseteq E(H)$, $q'\in A'$, and $y_i,b_j,q,q'$ is an induced $C_4$ in $H$, we have that $q'\in A$. Moreover, by the construction of $H$, the only one vertex in $A$ which is adjacent to $y_i$ is $a_i$. Therefore, $q'=a_i$. Then, since $(a_i,b_j)=(q',b_j)\in E(H)$, we got a contradiction that $y_i,b_j,q,q'$ is an induced $C_4$ in $H$.
%
\end{proof}

 \paragraph*{Split Graph Recognition.}
Here, we prove a streaming lower bound for recognizing split graphs that implies that our corollary for Theorem~\ref{thm:Hcovering} when $\cH$ is the class of split graphs, is essentially tight. 

 \begin{theorem}
Any $1$-pass streaming algorithm to test whether the graph is a split graph  requires $\Omega(n)$ space.
\end{theorem}


\begin{proof}[Proof sketch]
We give a reduction from the {\sf INDEX} problem. In this problem Alice has a bit string $x_1\ldots x_n$ and Bob has an index $i\in [n]$. Bob wants to know that whether $x_i=1$ or not. Alice and Bob will construct a graph $G$ as follows. The vertex set of $G$ is $\{y_1,\ldots,y_n\}\cup \{z_1,\ldots,z_n\}$. For each $1\leq j<j'\leq n$ such that $x_i=x_j=1$ Alice will add an edge $(y_j,y_{j'})$. Recall that
Bob has an index $i$. He will  add edges between $x_i$ and $\{z_1,\ldots,z_n\}$.
Notice that if $x_i=1$, then $G$ graph is a split graph. Otherwise $G$ is  not a split graph.
Thus, by Proposition~\ref{prop:index}, the theorem follows.
\end{proof}

A similar lower bound for {\sc Custer Vertex Deletion} is already known from \cite{ChitnisCEHMMV16}.



 \section{Concluding Remarks and Discussion}\label{sec:conclusion}

In this paper, 
we have given a framework to design semi-streaming algorithms for NP-hard cut problems and proved two meta theorems that reduce the task of solving, for any graph class $\cH$, the corresponding {\probVDH} problem,
to just the recognition problem (i.e. solving it for just for $k=0$)
when: ~~{\em (i)}~$\cH$ is characterized by a finite set of forbidden induced subgraphs, or ~~{\em (ii)}~$\cH$ is hereditary, where additionally ``reconstruction'' rather than just recognition is required. The second case encompasses graph classes that are characterized by a possibly infinite set of obstructions.  We have exemplified the usage of our theorems with several applications to well-studied graph problems, and we believe that these will serve as building blocks for many more applications in the future. 
We remark that the running times in many of our results are also optimal under ETH.

A point to note here is that considering {\probVDH} where $\hh$ is defined by excluding a set of forbidden graphs as {\em minors} or {\em topological minors} is not very interesting in our setting. In particular, if $\hh$ is defined by excluding a {\em finite} set of  forbidden graphs as minors or topological minors (e.g., $\hh$ is the class of planar graphs), then the 
number of edges in any $G\in \cH$ is upper bounded by $\cO(n)$. Thus, given an instance $(G,k)$ of \probVDH\, one can check whether the number of edges in $G$ is $\cO(kn)$. If the answer is no, then we have a no-instance. Otherwise, we can simply store the entire graph.
%
%
This means that \probVDH for such $\cH$ trivially belongs to the class  {\semips}~\cite{ChitnisC19}. Further, we know that for all such $\cH$,  \probVDH is in {\sf FPT} (see, for example, \cite{CyganFKLMPPS15}), which implies that in fact, whether or not \probVDH belongs to the class {\fptsemips} (which is contained in {\semips}) simply depends on whether there is a {\em static} linear-space FPT algorithm for the problem that can be used in post-processing. While this is an interesting research direction in itself, it presents different challenges to those we encounter when the goal is to sketch the input graph stream.
On the other hand, suppose that $\hh$ is defined by excluding an {\em infinite} number of forbidden graphs as minors. It is well-known from the Graph Minors project that if a graph class is characterized by an infinite set of forbidden minors, then it is also characterized by a {\em finite} set of forbidden minors, as a result of which we fall back into the previous case. Finally, to the best of our knowledge, there are no natural graph classes characterized by an infinite set of forbidden topological minors.

Let us conclude the paper with a few open questions and further research directions.

\begin{itemize}
\item A central tool in Parameterized Complexity is the usage of graph decompositions. Which among them is relevant to parameterized problems in the semi-streaming model? Also, which of the width-parameters corresponding to them is relevant? With respect to tree decompositions, once a graph has bounded treewidth $w$, its number of edges is already $\OO(nw)$, which means that it can just be stored explicitly. More interestingly, we believe that the ``unbreakability decompositions'' (which generalize tree decompositions) given in \cite{DBLP:journals/siamcomp/CyganLPPS19,CyganKLPPSW21} can be computed in $\widetilde{\OO}(k^{\OO(1)}n)$ space. If this is achieved, then the classic {\sc Minimum Bisection} problem (see Appendix for the formal definition) can be dealt in the same manner as we deal with other cut problems in this paper, yielding a parameterized algorithm for it in the semi-streaming model. We leave this as an interesting question.

\item One of the main inspirations behind Parameterized Complexity is the graph minors project of Robertson and  Seymour~\cite{DBLP:conf/birthday/Downey12,DBLP:conf/birthday/Lokshtanov0Z20}. Here, the central algorithms are for {\sc Minor Containment} and {\sc Disjoint Paths}, also on specific graph classes. Could one develop a theory for algorithmic graph minors in the semi-streaming model?

\item Our work naturally motivates revisiting and improving classic algorithmic methods in Parameterized Complexity where space exceeds $\widetilde{\OO}(k^{\OO(1)}n)$. For example, consider the celebrated matroid-based technique to design polynomial kernels by Kratsch and Wahlstr\"{o}m \cite{DBLP:journals/jacm/KratschW20,DBLP:journals/talg/KratschW14}---can it be implemented in $\widetilde{\OO}(k^{\OO(1)}n)$ space? If the answer is positive, then, for example, composed with our sketch for {\sc Odd Cycle Transversal}, we may derive a semi-streaming kernelization  for this problem. Here, and in many other foundational techniques, completely new ideas may be required. As we mentioned above, in light of semi-streaming in particular, and big data in general, increasing importance must be given to space complexity.

\item In recent years, parameterized approximation algorithms and lossy kernelization have become a very active subject of research (see, e.g., \cite{marx2008parameterized,DBLP:journals/algorithms/FeldmannSLM20,DBLP:journals/siamcomp/ChalermsookCKLM20,LokshtanovPRS17,KulikS20,DBLP:journals/jacm/SLM19,FPTApprox21,LokshtanovSS20}). Indeed, we do not have to choose between parameterized analysis and approximation, but combine them and ``enjoy the best of both worlds''---for example, we can develop a $(1+\epsilon)$-approximation algorithm that runs in $2^{\OO(k)}\cdot n^{\OO(1)}$ time for some problems that are both APX-hard and W[1]-hard. Naturally, we can use approximation also in the context of parameterized algorithms in the semi-streaming model. Then, we may use the same definition, but seek an approximate rather than exact solution.

\item  Lastly, we would like to remark that the study of parameterized (graph) problems in the semi-streaming model can also be of significant interest when the input graph is restricted to belong to a certain graph class. Here, of course, it only makes sense to consider dense classes of graphs. For example, most geometric intersection graphs classes are dense (containing cliques of any size), and, recently, they are being extensively studied from the viewpoint of Parameterized Complexity (see, e.g., \cite{giannopoulos2008parameterized,DBLP:conf/compgeom/FominLP0Z20,DBLP:conf/stoc/BergBKMZ18,DBLP:conf/soda/Panolan0Z19,DBLP:journals/dcg/FominLPSZ19}). So, they may be natural candidates in this regard. In particular, recognizing the family of interval graphs in the semi-streaming model is an interesting open question.
\end{itemize}

\bibliographystyle{alpha}
\bibliography{references,references-old,ref-reconst,ref-relatedwork,referencesConclusion,elimination}
\appendix
\section{Problem Definitions}\label{ap:problems}


%
%

%

\defparproblem{\sc Feedback Vertex Set on Tournaments}{A tournament $T$ and integer $k$.}{$k$}
{Is there a set $S \subseteq V(G)$ of size at most $k$ such that $T-S$ is acyclic.}

\defparproblem{\sc  Split Vertex Deletion}{A graph $G$ and integer $k$.}{$k$}
{Is there a set $S \subseteq V(G)$ of size at most $k$ such that $G-S$ is a split graph, i.e., the vertex set can be partitioned into two -- one inducing a complete graph and the other, an independent set.}

\defparproblem{\sc  Threshold Vertex Deletion}{A graph $G$ and integer $k$.}{$k$}
{Is there a set $S \subseteq V(G)$ of size at most $k$ such that $G-S$ is a threshold graph, i.e., a split graph that also excludes the induced path on 4 vertices.}

\defparproblem{\sc  Cluster Vertex Deletion}{A graph $G$ and integer $k$.}{$k$}
{Is there a set $S \subseteq V(G)$ of size at most $k$ such that $G-S$ is a cluster graph, i.e., every connected component is a complete graph.}

\defparproblem{\sc Minimum Bisection}{A graph $G$ and an integer $k$.}{k}
{Does there exists a partition $(A,B)$ of $V(G)$ such that
$||A| - |B|| \leq 1$ and $E(A,B) \leq k$? }


%


\end{document}